\def\final{1}
\def\ijqi{1}

\documentclass[11pt]{article}

\usepackage{gastex}
\usepackage{amsthm}
\usepackage{mathpazo}
\usepackage{times}
\usepackage{subfigure}
\usepackage{amsfonts}
\usepackage{multicol}
\usepackage{graphicx}
\usepackage{amssymb}
\usepackage{amsmath}
\usepackage{latexsym}
\usepackage[usenames]{color}
\usepackage{paralist}
\usepackage{soul}
\usepackage{verbatim}
\usepackage{mdwlist}
\usepackage{algorithmic}
\usepackage{algorithm}
\usepackage[numbers]{natbib}
\usepackage{hyperref}

 \hypersetup{pdfpagemode=UseNone}

\newtheorem{theorem}{Theorem}[section]
\newtheorem*{theorem*}{Theorem}
\newtheorem{lemma}[theorem]{Lemma}
\newtheorem{definition}[theorem]{Definition}
\theoremstyle{definition}
\newtheorem{remark}{Remark}
\newtheorem{claim}[theorem]{Claim}
\newtheorem{obs}[theorem]{Observation}

\newtheorem{cor}[theorem]{Corollary}
\newtheorem{prop}[theorem]{Proposition}
\newtheorem*{prop*}{Proposition}
\newtheorem{assumption}{Assumption}
\newtheorem{example}{Example}

\newcommand{\mypar}[1]{\smallskip \noindent {\bf {#1}.}}

\newenvironment{pffs}[3][]{
\begin{figure}[ht!]
\begin{center}
   \begin{tabular}{|ll|}\hline
     \hspace{.1ex}
\begin{minipage}{.9\linewidth}\vspace{1ex}
       {\begin{center}{\bf #2}
           #3 \end{center}}\vspace{-1ex}
	}{%
       \vspace{-2ex}
       \smallskip
     \end{minipage} & \\
     \hline
   \end{tabular}
   \end{center}
\end{figure}
}

\ifnum\final=0
\newcommand{\mynote}[2]{\marginpar{\tiny {\bf\sc #1}: {#2}}}
\else
\newcommand{\mynote}[2]{}
\fi
\newcommand{\anote}[1]{\mynote{ads}{#1}}
\newcommand{\fnote}[1]{\mynote{fs}{#1}}

\newcommand{\asnref}[1]{Assumption~\ref{asn:#1}}

\newcommand{\ip}[2]{\left\langle #1, #2\right\rangle}

\newcommand{\setft}[1]{\mathrm{#1}}
\newcommand{\lin}[1]{\setft{L}\left(#1\right)}
\newcommand{\density}[1]{\setft{D}\left(#1\right)}

\newcommand{\class}[1]{\textup{#1}}

\newcommand{\kb}[1]{|#1\rangle\langle #1|}
\def\auth{\bar A}
\def\dauth{\bar B}
\def\negl{\mathrm{negl}}

\newcommand{\etal}{\emph{et al.}}

\newcommand{\model}[3]{\mathcal{M}_{#1,#2,#3}}

\newcommand{\reg}[1]{\mathsf{#1}}
\newcommand{\rspace}[1]{\mathsf{#1}}
\newcommand{\rspacen}[1]{\mathsf{#1}_{n}}
\newcommand{\linop}[2]{\setft{L}\left( #1, #2\right)}

\newcommand{\cprot}[3]{ {#1}^{{#2}/{#3}}} % composed protocol

\newcommand{\cqsa}{\text{\tt C-QSA}}
\newcommand{\sqsa}{\text{\tt S-QSA}}
\newcommand{\qsa}{\text{\tt QSA}}
\newcommand{\cquc}{\text{\tt C-QUC}}
\newcommand{\squc}{\text{\tt S-QUC}}
\newcommand{\ccuc}{\text{\tt C-CUC}}
\newcommand{\cqucwi}{\text{\tt C-QUC-WI}}
\def\wiq{\text{\tt QC-WI}}
\def\td{\text{\sc TD}}

\def\bbN{\mathbb{N}}

\def\calA{\mathcal{A}}
 
\def\calE{\mathcal{E}}
\def\calF{\mathcal{F}}
\def\calG{\mathcal{G}}
\def\calM{\mathcal{M}}
\def\calO{\mathcal{O}}
\def\calS{\mathcal{S}}
\def\calW{\mathcal{W}}

\def\calZ{\mathcal{Z}}

\def\pgen{\text{\bf PG}}
\def\gen{\text{\bf Gen}}
\def\enc{\text{\bf Enc}}
\def\dec{\text{\bf Dec}}

\def\bfA{\mathbf{A}}
\def\bfB{\mathbf{B}}
\def\bfP{\mathbf{P}}
\def\bfV{\mathbf{V}}
\def\bfS{\mathbf{S}}
\def\bfcomm{\mathbf{comm}}
\def\zka{\mathrm{ZK}_1}

\def\I{\mathbb{1}}
\def\({\left(}
\def\){\right)}

\def\a{{aux}_1}
\def\abar{\overline\a}
\def\b{{aux}_2}
\def\bbar{\overline\b}
\def\c{state}
\def\cbar{\overline\c}

\def\veps{\varepsilon}

\def\acc{\text{\sf acc}}
\def\rej{\text{\sf rej}}

\def\mac{M}
\def\fzk{\calF_{\tt ZK}}
\def\fcp{\calF_{\tt CP}}
\def\fcf{\calF_{\tt CF}}
\def\fot{\calF_{\tt OT}}
\def\fcom{\calF_{\tt COM}}
\def\wqc{\approx_{wqc}}
\def\qc{\approx_{qc}}
\def\qcii{\approx_{qci}}
\def\qsii{\approx_{qsi}}
\def\dmm{\approx_{\diamond}}
\def\trm{\approx_{tr}}

\def\exec{\mathbf{EXEC}}
\def\ideal{\mathbf{IDEAL}}

\begin{document}

%-----------------------------------------------------------------------------%
%\title{Quantum Secure Two-party Computation}
\title{Classical Cryptographic Protocols in a Quantum World\footnote{A preliminary version of this work appeared in Advances in Cryptology - CRYPTO 2011.}}
%-----------------------------------------------------------------------------%

\author{Sean Hallgren \thanks{Department of
    Computer Science and Engineering, Pennsylvania State University,
    University Park, PA, U.S.A.}\  \ \thanks{Partially supported by National Science
    Foundation award CCF-0747274 and by the National Security
    Agency (NSA) under Army Research Office (ARO) contract number
    W911NF-08-1-0298.}
    \and Adam Smith\ \footnotemark[2]\ \ \thanks{Partially supported by
    National Science Foundation award CCF-0747294.} \and Fang Song \thanks{Department of Combinatorics \& Optimization and Institute for Quantum Computing, University of Waterloo, Canada. Partially supported by Cryptoworks21, NSERC, ORF and US ARO. Most work was conducted while at the Pennsylvania State University.}}
\date{}
\maketitle

%\vspace{-0.75in}

\begin{abstract}
Cryptographic protocols, such as protocols for secure function
evaluation (SFE), have played a crucial role in the development of
modern cryptography. The extensive theory of these protocols,
however, deals almost exclusively with classical attackers. If we
accept that quantum information processing is the most realistic
model of physically feasible computation, then we must ask: what
classical protocols remain secure against quantum attackers?

Our main contribution is showing the existence of classical
two-party protocols for the secure evaluation of any polynomial-time
function under reasonable computational assumptions (for example, it
suffices that the learning with errors problem be hard for quantum
polynomial time). Our result shows that the basic two-party
feasibility picture from classical cryptography remains unchanged in
a quantum world.

\end{abstract}

\newpage
\tableofcontents
\newpage

\ifnum\final=0

%\mypar{Fang's To-Do List}

%\begin{enumerate}
%\item Stand-alone model: make ref. system explicit; consistent with CryptoConf version.

%\item CF protocol in ZK-hybrid: modify the definition of FCF to allow abort. 

%\item ZK protocol in CF-hybrid: Change UC-WI to a stand-alone definition and explain how to use it to construct UC-secure protocols.  

%\item Introduction: \\
  %-- add a few recent works, e.g., quantum SFE and Abstract Crypto framework (requested by Review 2)\\
  %-- fix a few undefined terms 

%\item Prelim: re-order refs.; and fix  def. 2.4.  

%\item Style: review 2 requested moving proofs into main body since it's a journal publication. I am thinking about merging appendix A,B, and C into main body. 

%\item update pics for emulation: make ref. system explicit

%\item for revision submission: response to review comments

%\item Misc. dash issue. quantum UC, quantum-UC or quantum-UC-emulates... 

%\end{enumerate}

\mypar{Questions}
\begin{enumerate}
%\item Intro. needs updating. (TBD)
%\item %Indistinguishability:  clearer organization. 
%prelim: defs of functionalities zk etc. 
%\item security model
  %\begin{enumerate}
%\item Semi-honest adversary coherent (side note~\ref{anote:coherent}): what does coherent mean precisely? how to achieve? reference?
%  \item work out an light-weight version of that in thesis 
%\item 
%: where to put UC?
  %\end{enumerate}
 % \item Functionalities and SFE: inconsistency. sometimes say SFE, sometimes we say any 2-party functionality. Be precise what we can achieve. %Related issues include separating poly-time, well-defined, reactive  functionalities, etc. 
 
%\item Simple hybrid argument Sect.~\ref{sec:quc}
    %\begin{enumerate}
   % \item simple hybrid machines: a. state precisely what it means that a quantum machine takes a classical sample. b. In Pic~\ref{fig:sime}: $\sigma$ (a bit string here rather than a quantum state) may be confusing.
%    \end{enumerate}
%\item Lifting CLOS: elaborate: in each step, what kind of hybrids 
%\item Equivalence between $\fzk$ and $\fcf$: fill in details

%\item Fang's affiliation. should I say "most of this work was done while at Penn state"

% \item p21. @Adam  Reference: semi-honest OT from dense-crypto [Adam:
%   added a reference]
%\item p35. @Adam Reference: concurrent WI [Adam: Not sure. Asking some friends.]

%\item p8. @Sean Machine models: are these equivalence? I am not sure. maybe not mentioning at all. 

%\item which proofs go into main body? (as requested by Review2)

\item p14. semi-honest quantum adversary: ``we focus on first (should be second?) model''

% \item p24. @Adam. GUC what to say there? [Adam: I think what we've
%   said is fine.]
 
%\item p17. UC composition theorem. Justify the Proof if Z takes quantum advice.  

%\item New stuff on Quantum-WI in Sect.~\ref{subsec:cf2zk}. does it look right to you? 

%\item Any more literature review, e.g. quantum random-oracle and other related work by Mark Zhandry? 
%
%\item Appendix~\ref{sec:variants-app} (relations between models)
%\begin{enumerate}
%	\item Reflections: do we want to keep it? what are the big messages we want to deliver? e.g., including quantum advice to adversaries.
%\end{enumerate}

%\item  update pics: emulation definition and mod composition?
% \item bib
\end{enumerate}

\
\newpage
\fi

%---------------------------------------------------------------------------%
\section{Introduction}
\label{sec:intro}
%---------------------------------------------------------------------------%

%\section{Motivation}
Cryptographic protocols, such as protocols for secure function
evaluation (SFE), have played a crucial role in the development of
modern
cryptography.  % Secure function evaluation (SFE) allows a group of
% players, each holding a secret input (e.g., a vote) to jointly
% evaluate some function of their inputs (say, the votes' tally)
% without revealing anything except the function's value. A special
% case of this is a zero-knowledge (ZK) proof system, which allows a
% prover $P$ who knows a short proof of a statement to interactively
% prove the statement to a computationally-bounded verifier $V$
% without revealing anything except the statement's veracity. The very
% possibility of such protocols is counterintuitive. But
Goldreich, Micali and Wigderson~\cite{GMW87}, building on the
development of zero-knowledge (ZK) proof systems \cite{GMR85,GMW91},
showed that SFE protocols exist for any polynomial-time function
under mild assumptions (roughly, the existence of secure public-key
cryptosystems).
%,CCD88,BGW88,RB89},
% and, assuming only the existence of one-way functions, ZK proof systems are
% possible for any language in NP~\cite{GMW91,HILL,Naor}.
Research into the design and analysis of such protocols is now a
large subfield of cryptography; it has also driven important
advances in more traditional areas of cryptography such as the
design of encryption, authentication and signature schemes.

The extensive theory of these protocols, however, deals almost
exclusively with classical attackers. However, given our current
understanding of physics, quantum information processing is the most
realistic model of physically feasible computation. It is natural to
ask: \emph{what classical protocols remain secure
  against quantum attackers?} In many cases, even adversaries with modest
quantum computing capabilities, such as the ability to share and
store entangled photon pairs, are not covered by existing proofs of
security.

Clearly not all protocols are secure: we can rule out anything based
on the computational hardness of factoring, the discrete
log~\cite{Shor97}, or the principal ideal problem~\cite{Hallgren07}.
More subtly, the basic techniques used to reason about security may
 not apply in a quantum setting. For example, some
 information-theoretically secure two-prover
 ZK and commitment protocols are analyzed by viewing the
  provers as long tables that are fixed before queries are chosen by
  the verifier;
  quantum entanglement breaks that analysis and some protocols are
  insecure against colluding quantum provers (Cr{\'e}peau \etal,~\cite{CSST11}).

  In the computational realm, \emph{rewinding} is a key technique for
  basing the security of a protocol on the hardness of some
  underlying problem. Rewinding proofs consist of a mental experiment
  in which the adversary is run multiple times using careful
  variations of a given input. At first glance, rewinding seems impossible with
  a quantum adversary since running it multiple times
  might modify the entanglement between its internal storage and an outside reference
  system, thus changing the overall system's behavior.

  In a breakthrough paper, Watrous~\cite{Wat09} showed that a specific
  type of zero-knowledge proof (3-round, GMW-style protocols) can be
  proven secure using a rewinding argument tailored to quantum
  adversaries.  Damg{\aa}rd and Lunemann~\cite{DL09} showed that a
  similar analysis can be applied to a variant of Blum's coin flipping
  protocol.  Hallgren \etal~\cite{HKSZ08} showed
  certain classical transformations from
  honest-verifier to malicious-verifier ZK can be modified to provide
  security against malicious quantum verifiers.
  Some information-theoretically secure classical protocols are also
  known to resist quantum
  attacks~\cite{CGS02,BCGHS06,FS09,Unr10}. Finally, there is a longer
  line of work on protocols that involve quantum communication, dating
  back to \citet*{bennett1984quantum}. Overall, however, little is known
  about how much of the classical theory can be carried over to
  quantum settings. See ``Related Work'', below, for more detail.

%---------------------------------------------------------------------------%
\subsection{Our Contributions} \label{ssec:contr}
%---------------------------------------------------------------------------%

Our main contribution is showing the existence of classical
two-party protocols for the secure evaluation of any polynomial-time
function under reasonable computational assumptions (for example, it
suffices that the learning with errors problem~\cite{Reg09} be hard
for quantum polynomial time). Our result shows that \emph{the basic
two-party feasibility picture from classical cryptography remains
unchanged in a quantum world}. The only two-party general SFE
protocols which had previously been analyzed in the presence of
quantum attackers required quantum computation and communication on
the part of the honest participants~(e.g.~\cite{CDMS04,DFLSS09}).

%\fnote{SFE vs. Functionality}
% We show that a large class of classical security
% analyses remain valid in the presence of quantum attackers as long as
% the underlying computational primitives (encryption schemes,
% pseudorandom generators, etc) resist quantum attack. We provide a
% building block, ZK arguments of knowledge secure against quantum
% attackers, which we show can be combined with existing protocols to
% yield our main result.

In what follows, we distinguish two basic settings: in the
\emph{stand-alone} setting, protocols are designed to be run in
isolation, without other protocols running simultaneously; in
\emph{network} settings, the protocols must remain secure even when
the honest participants are running  other protocols (or copies of
the same protocol) concurrently. Protocols proven secure in the
\emph{universal composability} (UC) model~\cite{Can01} are secure in
arbitrary network settings, but UC-security is impossible to achieve
in many scenarios.

Our contributions can be broken down as follows:

\mypar{General Modeling of Stand-Alone Security with Quantum Adversaries}
We describe a security model for two-party protocols in the presence
of a quantum attackers. Proving security in this model amounts to
showing that a protocol for computing a function $f$ behaves
indistinguishably from an ``ideal'' protocol in which $f$ is computed
by a trusted third party, which we call the ideal functionality
$\calF$. Our model is a quantum analogue of the model of stand-alone
security developed by Canetti~\cite{Can00} in the classical
setting. It slightly generalizes the existing model of \citet{DFLSS09} in two ways. First, our model allows for protocols in which the ideal functionalities process quantum information (rather than only
classical functionalities). Second, it allows for adversaries that take
arbitrary quantum advice, and for arbitrary entanglement between
honest and malicious players' inputs. Our model may be
viewed as a restriction of the quantum UC model of \citet{Unr10} to
noninteractive distinguishers, and we use that connection in our
protocol design (see below). We also discuss possible variants of quantum stand-alone models and initiate a study on their relationships, which connects to interesting questions in a broad scope. 

We show a sequential modular composition theorem for protocols
analyzed in our model. Roughly, it states that one can design
protocols modularly, treating sub-protocols as equivalent to their
ideal versions when analyzing security of a high-level protocol. While
the composition result of \citet{DFLSS09} allows only for classical
high-level protocols, our result holds for arbitrary quantum
protocols.

\ifnum\final=0
Moreover, we initiate a study on definitional issues of quantum
stand-alone security models in a unified framework. We show the
equivalence of several variants of our model (the variants have to do
with how auxiliary information is provided and whether it can be
entangled with an outside system). The equivalences demonstrate the
robustness of our model, and also simplify the literature. 

The equivalences are not trivial; for example, they use
quantum authentication schemes (\citet{BCGST02}) to prevent certain interactions between the environment and the auxiliary information.
Along the way, we generalize the soundness property of authentication
schemes to a ``network'' setting with computationally bounded
attackers. Roughly speaking, if a tampering operation on an
authenticated message causes no difference to a poly-time observer,
then the operation is essentially identity, again from a poly-time
observer's view. The requirement must hold even when the message being
authenticated is part of an entangled state. This in a way helps our
understanding of the power of entanglement. Our study also relates to other basic questions in complexity theory and quantum information. 
\fi
\mypar{Classical Zero-knowledge Arguments of Knowledge Secure Against
  Quantum Adversaries}
We construct a classical zero-knowledge argument of
knowledge (ZKAoK) protocol that can be proven secure in our
stand-alone model.
Our construction is
``witness-extendable''~(\citet{Lin03}), meaning that one can
simulate an interaction with a malicious prover and simultaneously
extract a witness of the statement whenever the prover succeeds.
Our security proof overcomes a limitation of the previous construction of (two-party) quantum proofs of knowledge
(\citet{Unr-qpok}), which did not have a simulator for malicious
provers. 
Such a simulator is important since it allows one to analyze
security when using a proof of knowledge as a subprotocol. As in
the classical case, our ZKAoK protocol is an important building
block in designing general SFE protocols.

The main idea behind our construction is to have the prover and
verifier first execute a weak coin-flipping protocol to generate a
public key for a special type of encryption scheme.  The prover
encrypts his witness with respect to this public key and proves
consistency of his ciphertext with the statement $x$ using the ZK
protocols analyzed by Watrous~\cite{Wat09}. A simulator playing the
role of the verifier can manipulate the coin-flipping phase to
generate a public key for which she knows the secret key, thus
allowing her to extract the witness without needing to rewind the
prover. A simulator playing the role of the prover, on the other hand,
cannot control the coin flip (to our knowledge) but can ensure that
the public key is nearly random. If the encryption scheme satisfies
additional properties (that can be realized under widely
used lattice-type assumptions), we show that the verifier's view can
nonetheless be faithfully simulated. %\citet{LN11} independently gave a similarly-flavored construction of ZKAoK for quantum adversaries; see ``Related Work''.

\mypar{Classical UC Protocols in a Quantum Context: Towards
  Unruh's Conjecture}
We show that a large class of protocols which are UC-secure against
computationally bounded classical adversaries are also UC-secure
against quantum adversaries. Unruh~\cite{Unr10}
showed that any classical protocol which is proven UC-secure against
unbounded classical adversaries is also UC-secure against unbounded
quantum adversaries. He conjectured (roughly, see~\cite{Unr10} for
the exact statement) that classical arguments of
\emph{computational} UC security should also go through as long as
the underlying computational primitives are not easily breakable by
quantum computers.

We provide support for this conjecture by describing a family of
classical security arguments that go through verbatim with quantum
adversaries. We call these arguments ``simple hybrid arguments''.
They use rewinding neither in the simulation nor in any of the steps
that show the correctness of simulation.\footnote{In general, it is
hard to clearly define what it means for a security proof to ``not
use rewinding''. It is not enough for the protocol to have a
straight-line simulator, since the proof of the simulator's
correctness might still employ rewinding. Simple hybrid arguments
provide a clean, safe subclass of arguments that go through with
quantum adversaries.}

Our observation allows us to port a general result of \citet{CLOS02}
to the quantum setting.  We obtain the following: in the
$\calF_{ZK}$-hybrid model, where a trusted party implementing ZKAoK
is available, there exist classical protocols for the evaluation of
any polynomial-time function $f$ that are UC-secure against quantum
adversaries under reasonable computational assumptions. As an
immediate corollary, we get a classical protocol that quantum-UC
emulates the ideal functionality $\calF_{CF}$ for coin-flipping,
assuming UC-secure ZK. 

\mypar{New Classical UC Protocols Secure Against Quantum Attacks}
\anote{This paragraph is a bit unclear. How does it relate to Fang's
 work on unconditional reductions in the quantum UC model?}
\fnote{what would you suggest to mention here? I have a few words
  about our TCC13 work in Section 5}
 We construct new two-party protocols that are UC-secure against quantum adversaries. Adapting ideas from~\citet{Lin03}, we show a \emph{constant-round} classical coin-flipping protocol
from ZK (i.e. in $\fzk$-hybrid model). Note that the general feasibility result from above  already implies the existence of a quantum-UC secure coin-flipping protocol, but it needs polynomially many rounds.  
Conversely, we can also construct a \emph{constant-round} classical protocol for ZKAoK that is UC-secure against quantum adversaries, assuming a trusted party implementing coin-flipping, i.e. in the $\calF_{CF}$-hybrid model (essentially equivalent to the \emph{common reference string} model, where all participants have
access to a common, uniformly distributed bit string).
%The ``simple hybrid arguments'' mentioned above do not suffice for proving the security of this ZKAoK protocol.  Specifically, one component of our protocol, a \emph{witness-indistinguishable} proof system, needs a new proof of security. The basic strategy is still a hybrid argument, but its analysis requires breaking the space of possible executions into pieces (classically, this involves conditioning on complementary events; quantumly, this involves projecting onto orthogonal subspaces) and arguing that (a) the adversary cannot have a significant advantage in either piece and (b) the original state was a mixture, not a superposition, of the two pieces. 
\fnote{I removed the highlights on analyzing WI protocol.}
This establishes the equivalence between $\calF_{ZK}$ and $\calF_{CF}$ in the quantum UC model, which may be of independent interest, e.g., in simplifying protocol
designs. It has also motivated a subsequent work by Fehr et el.~\cite{FKSZZ13} where they showed interesting connections between ideal functionalities in the quantum-UC model in a systematic way. %\fnote{Fixed! is zk from crs implicit in CanettiFischlin01? not constant round}

\mypar{Implications}  The modular composition theorem in our
stand-alone model allows us to get the general feasibility result
below by combining our stand-alone ZKAoK protocol and the UC-secure
protocols in $\fzk$-hybrid model:

Under standard assumptions, {\em there exist \emph{classical} SFE
protocols in the plain model (without a shared random string) which
are stand-alone-secure against static \emph{quantum} adversaries}.
This parallels the classic result of Goldreich, Micali and
Wigderson~\cite{GMW87}.

The equivalence of zero-knowledge and coin-flipping functionalities
in the UC model also has interesting implications. First,
the availability of a common reference string (CRS) suffices for
implementing quantum-UC secure protocols. Secondly, given our
stand-alone ZKAoK protocol, we get a quantum stand-alone
coin-flipping protocol.

Independently of our work, \citet{LN11}, via a different route, also
showed the existence of  classical two-party SFE protocols secure
against quantum attacks. See the discussion at the end of ``Related Work''.

%---------------------------------------------------------------------------%
\subsection{Related work}
\label{ssec:rwork}
%---------------------------------------------------------------------------%

In addition to the previous work mentioned above, we expand here on
three categories of related efforts.

\mypar{Composition Frameworks for Quantum Protocols} Systematic
investigations of the composition properties of quantum protocols
are relatively recent.
Canetti's UC framework and Pfitzmann and Waidner's closely related
\emph{reactive functionality} framework were extended to the world
of quantum protocols and adversaries by Ben-Or and
Mayers~\cite{BOM04} and Unruh~\cite{Unr04,Unr10}. These
frameworks (which share similar semantics) provide extremely strong
guarantees---security in arbitrary network environments. They were
used to analyze a number of unconditionally secure quantum protocols
(key exchange~\cite{BHLMO05} and multi-party computation with honest
majorities~\cite{BCGHS06}). However, many protocols are not
universally composable, and Canetti~\cite{Can01} showed that
classical protocols cannot UC-securely realize even basic tasks such
as commitment and zero-knowledge proofs without some additional
setup assumptions such as a CRS or public-key infrastructure.

Damg{\aa}rd \etal \cite{DFLSS09}, building on work by Fehr and
Schaffner~\cite{FS09}, proposed a general composition framework
which applies only to secure quantum protocols of a particular form
(where quantum communication occurs only at the lowest levels of the
modular composition). As noted earlier, our model is more general
and captures both classical and quantum protocols. Recently, Maurer and Renner proposed a new composable framework called Abstract Cryptography~\cite{MR11}, and it has been adapted to analyzing quantum protocols as well~\cite{DFPR14}. 

\mypar{Analyses of Quantum Protocols} The first careful proofs of
security of quantum protocols were for key exchange
(Mayers~\cite{Mayers01}, Lo and Chau~\cite{LC99}, Shor and
Preskill~\cite{SP00}, Beaver~\cite{Beaver02}). Research on quantum
protocols for two-party tasks such as coin-flipping, bit commitment
and oblivious transfer dates back farther~\cite{BC90,BBCS91}, though some initially proposed protocols were insecure~\cite{Mayers01}. The
first proofs of security of such protocols were based on
computational assumptions~\cite{DMS00,CDMS04}. They were highly
protocol-specific and it was not known how well the protocols
composed. The first proofs of security using the simulation paradigm
were for information-theoretically-secure protocols for multi-party
computations assuming a strict majority of honest
participants~\cite{CGS02,CGS05,BCGHS06}. More recently, Dupuis et al.~\cite{DNS10,DNS12} constructed  two-party quantum protocols for evaluating arbitrary unitary operations, which they proved secure under reasonable computational assumptions in a simulation-based definition similar to what we propose in this work. There was also a line of
work on the \emph{bounded quantum
  storage} model~\cite{DFSS08,DFSS07,FS09,Unr-bqs} developed tools for
reasoning about specific types of composition of two-party
protocols, under assumptions on the size of the adversary's quantum
storage. Many tools have been developed in recent years on modeling and analyzing composable security for protocols of device-independent quantum key-exchange and randomness expansion~\cite{FGS13,VV14,MS14,CSW14}.  %Unruh's UC security work, mentioned above, was the first we are aware of that was sufficiently general to encompass classical and quantum protocols and generic composition. 

\mypar{Straight-Line Simulators and Code-Based Games} As mentioned
above, we introduce ``simple hybrid arguments'' to capture a class
of straightforward security analyses that go through against quantum
adversaries. Several formalisms have been introduced in the past to
capture classes of ``simple'' security arguments. To our knowledge,
none of them is automatically compatible with quantum adversaries.
For example, \emph{straight-line black-box simulators}~\cite{KLR10}
%\fnote{Fixed. slbb simulator reference}
do not rewind the adversary nor use an explicit description of its
random coins; however, it may be the case that rewinding is
necessary to prove that the straight-line simulator is actually
correct. In a different vein, the \emph{code-based games} of Bellare
and Rogaway~\cite{BR06} capture a class of
hybrid arguments that can be encoded in a clean formal language;
again, however, the arguments concerning each step of the hybrid may
still require rewinding.

\mypar{Independent Work} Lunemann
and Nielsen~\cite{LN11} independently obtained  similar results to
the ones described here,  via a slightly
different route.  Specifically, they start by constructing a
stand-alone coin-flipping protocol that is fully simulatable against
quantum poly-time adversaries. Then they use the coin-flipping
protocol to construct a stand-alone ZKAoK protocol, and finally by
plugging into the GMW construction, they get quantum
stand-alone-secure two-party SFE protocols as well. The computational
assumptions in the two works are similar and the round complexities of the
stand-alone SFE protocols are both polynomial in the security
parameter. Our approach to composition is more general, however,
leading to results that also apply (in part) to the UC model.

\subsection{Future Directions}
\label{ssec:open}
Our work suggests a number of straightforward
conjectures. For example, it is likely that our techniques in fact
apply to all the results in CLOS (multi-party, adaptive adversaries)
and to corresponding results in the ``generalized'' UC
model~\cite{CDPW07}. Essentially all protocols in the semi-honest
model seem to fit the simple hybrids framework, in particular
protocols based on Yao's garbled-circuits framework
(e.g.~\cite{BMR90}). It is also likely that existing proofs in
security models which allow super-polynomial simulation
(e.g.,~\cite{Pas03,PS04,BS05}) will carry through using a similar
line of argument to the one here.

However, our work leaves open some basic questions: for example, can
we construct constant-round ZK with negligible completeness and
soundness errors against quantum verifiers? Watrous's technique does
not immediately answer it since sequential repetition seems necessary
in his construction to reduce the soundness error.  A quick look at
classical constant-round ZK (e.g.,~\cite{FS89}) suggests that
witness-indistinguishable proofs of knowledge are helpful. Is it
possible to construct constant-round witness-extendable WI proofs of
knowledge? Do our analyses apply to extensions of the UC framework,
such the \emph{generalized UC} framework of \citet{CDPW07}?
Finally, more generally, which other uses of rewinding can be adapted to
quantum adversaries? Aside from the original work by
Watrous~\cite{Wat09}, Damg{\aa}rd and Lunemann~\cite{DL09} and
Unruh~\cite{Unr-qpok} have shown examples of such adaption.

\mypar{Organization} The rest of the paper is organized as follows:
Section~\ref{sec:pre} reviews basic notations and definitions. In
Section~\ref{sec:qssc}, we propose our quantum
stand-alone security model. We show our main result in Section~\ref{sec:q2pc}.
Specifically, Section~\ref{subsec:qclos} establishes quantum-UC secure protocols in $\fzk$-hybrid model. A quantum stand-alone-secure ZKAoK protocol is developed in
Section~\ref{subsec:zkaok}. 
Finally in Section~\ref{sec:zk=cf}, we discuss equivalence of $\fzk$ and
$\fcf$.
%introduces the ``simple hybrid argument''
%framework, and we show classical SFE protocols that are quantum UC
%secure in the $\calG_{ZK}$-hybrid model.
%Finally in Section~\ref{sec:app}, we obtain, among other consequences, classical SFE that are quantum stand-alone-secure with no set-up assumptions. 
%, and conclude with future directions.

%---------------------------------------------------------------------------%
\section{Preliminaries}
\label{sec:pre}
%---------------------------------------------------------------------------%
% \mypar{Basic Notation}
For $m \in  \bbN$, $[m]$ denotes the set $\{1, \ldots, m\}$. We
use $n \in \mathbb{N}$ to denote a {\em security parameter}. The
security parameter, represented in unary, is an implicit input to
all cryptographic algorithms; we omit it when it is clear from the
context. Quantities derived from protocols or algorithms
(probabilities, running times, etc.)  should be thought of as
functions of $n$, unless otherwise specified. A function $f(n)$ is
said to be negligible if $f = o(n^{-c})$ for any constant $c$, and
$\text{negl}(n)$ is used to denote an unspecified function that is
negligible in $n$. We also use $poly(n)$ to denote an unspecified
function $f(n) = O(n^c)$ for some constant $c$. When $D$ is a
probability distribution, the notation $x \gets D$ indicates that
$x$ is a sample drawn according to $D$. When $D$ is a finite set, we
implicitly associate with it the uniform distribution over the set.
If $D(\cdot)$ is a probabilistic algorithm, $D(y)$ denotes the
distribution over the output of $D$ corresponding to input $y$. We
will sometimes use the same symbol for a random variable and for its
probability distribution when the meaning is clear from the context.
Let $\mathbf{X} = \{X_n\}_{n\in \bbN}$ and $\mathbf{Y} =
\{Y_n\}_{n \in \bbN}$ be two ensembles of binary random
variables. We call $\mathbf{X,Y}$ {\em indistinguishable}, denoted
$\mathbf{X} \approx \mathbf{Y}$, if $ \left|\Pr(X_n = 1) - \Pr(Y_n =
1)\right| \leq \text{negl}(n)$.

We assume the reader is familiar with the basic concepts of
 quantum information theory (see, e.g., \cite{NC00}).
We use a {sans serif} letter (e.g. $\reg{X}$) to denote both a quantum
register and the corresponding Hilbert space. We use $\rspacen{X}$ if
we want to be specific about the security parameter. Let $\mathcal{H}_n$
denote the space for $n$ qubits.
%a quantum register and for each $n$, we use %script letter (e.g.,
%$\mathcal{X}(n)$) to denote the corresponding %Hilbert space.
Let $\density{\rspace{X}}$ be the set of density operators acting on space
$\rspace{X}$ and $\linop{\rspace{X}}{\rspace{Y}}$ be the set of linear operators from space $\rspace{X}$ to $\rspace{Y}$.

%A superoperator $\Phi:
%D(\mathcal{H}_1)\to D(\mathcal{H}_2)$ is a completely positive trace
%preserving mapping that represents physically realizable operations
%(i.e., unitary transformations, measurements, and initializations
%and removals of qubits). We say a superoperator  is \emph{classical}
%if it can be implemented by a classical Turing machine.

\mypar{Quantum Machine Model}
We adapt Unruh's machine model in~\cite{Unr10} with minor changes. A
{\em quantum interactive machine} (QIM) $M$ is an ensemble of
interactive circuits $\{M_x\}_{x \in I}$. The index set $I$ is typically the natural numbers $\bbN$ or a set of strings $I\subseteq \{0,1\}^*$ (or both).  We give our description here with respect to $\{M_n\}_{n\in \bbN}$. For each value $n$
of the security parameter, $M_n$ consists of a sequence of circuits
$\{M_n^{(i)}\}_{i=1,...,\ell(n)}$, where $M_n^{(i)}$ defines the
operation of $M$ in one round $i$ and $\ell(n)$ is the number of
rounds for which $M_n$ operates (we assume for simplicity that
$\ell(n)$ depends only on $n$). We omit the scripts when they are
clear from the context or are not essential for the discussion. Machine $M$
(or rather each of the circuits that it comprises) operates on three
registers: a state register $\reg{S}$ used for input and workspace; an
output register $\reg{O}$; and a network register $\reg{N}$ for communicating
with other machines. %Hence $M_n$ is effectively a mapping from state
%space $\s(n)$ to output space $\mathcal{O}(n)$.\anote{Fixed!This sentence
%``Hence, $M_n$ ....'' doesn't make sense, because it ignores interaction.}
The size (or running time) $t(n)$ of $M_n$ is the sum of the sizes
of the circuits $M_n^{(i)}$. We say a machine is polynomial time if
$t(n)=poly(n)$ and
% . We say a QIM is uniformly generated (in polynomial time) if
there is a deterministic classical Turing machine that computes the
description of $M_n^{(i)}$ in polynomial time on input $(1^n,1^i)$.
%Sometimes we also consider a machine that is a collection $\{M_x\}_{x\in S}$ indexed by a set of strings $S \subseteq \{0,1\}^*$. In this case the complexity is parameterized by $|x|$, the length of the index string. 

%\anote{Fixed! I'm not sure if it makes sense to assume all machines are uniformly generated.}
%\fnote{more careful about the complexity definition. maybe introduce classes q(t)/poly etc.}

When two QIMs $M$ and $M'$ interact, they share network register $\reg{N}$. The circuits $M_n^{(i)}$ and ${M'}_n^{(i)}$ are executed
alternately for $i=1,2,...,\ell(n)$. When three or more machines
interact, the machines may share different parts of their network
registers (for example, a private channel consists of a register
shared between only two machines; a broadcast channel is a register
shared by all machines). The order in which machines are activated
may be either specified in advance (as in a synchronous network) or
adversarially controlled.

A non-interactive quantum machine (referred to as QTM hereafter) is a
QIM $M$ with network register empty and it runs for only one round (for
all $n$). This is equivalent to the {\em quantum Turing machine}
model (see~\cite{Yao93}).

% \fnote{our QTM and standard QTM are they indeed equivalent?}
%\anote{I think that \emph{with advice}, they are. In both cases you get the class of quantum circuits (this follows from Yao, 1993, ``Quantum circuit complexity'', who showed  that QTM's can be simulated by poly-size circuits). I am not  sure that they are equivalent without advice, though. Sean?} 

A classical interactive Turing machine is a special case of a QIM, where
the registers only store classical strings and all circuits are
classical. This is also called an interactive Turing machine
(ITM) with advice (\citet{Can00,Can01}).

\mypar{Indistinguishability of Quantum States}
Let $\rho = \{\rho_n\}_{n\in \bbN}$ and $\eta = \{\eta_n\}_{n\in \bbN}$ be ensembles of
mixed states indexed by $n\in\bbN$, where $\rho_n$ and $\eta_n$ are
both $r(n)$-qubit states for some polynomially bounded function
$r$. We first define a somewhat weak notion of indistinguishability of
quantum state ensembles.

\begin{definition}[$(t, \veps)$-weakly  indistinguishable states]
\label{def:wqc} 
We say two quantum state ensembles $\rho =
\{\rho_n\}_{n\in \bbN}$ and $\eta = \{\eta_n\}_{n\in \bbN}$
are \emph{$(t,\veps)$-weakly indistinguishable}, denoted $\rho
\wqc^{t, \veps} \eta$, if for every $t(n)$-time QTM $\calZ$, 
%\anote{Fixed. $d$ notation not yet defined for quantum states}
$$\left|\Pr[\calZ(\rho_n) =1] - \Pr[ \calZ(\eta_n) =
1]\right|  \leq \veps(n)\, .$$
\end{definition}

The states $\rho$ and $\eta$ are called {\em weakly computationally indistinguishable}, denoted $\rho \wqc \eta$, if for every polynomial $t(n)$, there exists a negligible $\veps(n)$ such that $\rho_n$ and $\eta_n$ are $(t,\veps)$-weakly computationally indistinguishable.

A stronger notion of indistinguishability of quantum states was proposed by Watrous~\cite[Definition 2]{Wat09}. The crucial distinction is that a distinguisher is allowed to take quantum advice.

\begin{definition}[$(t, \veps)$-indistinguishable states]
\label{def:qc}
We say two quantum state ensembles $\rho =
\{\rho_n\}_{n\in \bbN}$ and $\eta = \{\eta_n\}_{n\in \bbN}$
are $(t,\veps)$-indistinguishable, denoted $\rho \qc^{t, \veps} \eta$, if for every $t(n)$-time QTM $\calZ$ and
any mixed state $\sigma_n$,% on $t(n)$ qubits, 
$$\left|\Pr[\calZ(\rho_n\otimes\sigma_n) =1] - \Pr[ \calZ(\eta_n\otimes\sigma_n) =
1]\right|  \leq \veps(n)\, .$$
\end{definition}

The states $\rho$ and $\eta$ are called {\em quantum computationally indistinguishable}, denoted $\rho \qc \eta$, if for every polynomial $t(n)$, there exists a negligible $\veps(n)$ such that $\rho_n$ and $\eta_n$ are $(t,\veps)$-indistinguishable.
%  with $t(n), s(n) \leq poly(n)$ and
% $\veps(n) = negl(n)$.

The two definitions above subsume classical distributions as a special case, since classical
distributions can be represented by density matrices
%$\{\rho_n\}_{n\in\bbN}$ and $\{\eta_n\}_{n\in\bbN}$
that are diagonal with respect to the standard basis.

\mypar{Indistinguishability of Quantum Machines} Now we introduce
the notion of distinguishing two QTMs.

\begin{definition}[$(t, \veps)$-indistinguishable QTMs]
\label{def:qcm} 
We say two QTMs $M_1$ and $M_2$ are
{\em$(t,\veps)$-indistinguishable}, denoted $M_1 \approx_{qc}^{t,
  \veps} M_2$, if for any $t(n)$-time QTM $\calZ$ and any mixed state
$\sigma_n \in \density{\rspacen{S} \otimes \rspacen{R}}$, where
$\rspacen{R}$ is an arbitrary reference system,

$$ \left| \Pr[ \calZ((M_1\otimes \mathbb{1}_{\lin{\rspace{R}}})\sigma_n) = 1 ] - \Pr [ \calZ((M_2\otimes \mathbb{1}_{\lin{\rspace{R}}})\sigma_n) = 1] \right| \leq \veps(n) \, .$$
Machines $M_1$ and  $M_2$ are called {\em quantum computationally
indistinguishable}, denoted $M_1 \qc M_2$, if for every polynomial $t(n)$, there exists a negligible $\veps(n)$ such that $M_1$
and $M_2$ are $(t,\veps)$-computationally indistinguishable. 
\end{definition}
This definition is equivalent to quantum computationally indistinguishable super-operators proposed by Watrous~\cite[Definition 6]{Wat09}.
If we do not restrict the running time of the distinguisher, we obtain a statistical notion of indistinguishability. Let $\td(\cdot,\cdot)$ be the trace distance between density operators. 

\begin{definition}[$ \veps$-indistinguishable QTMs in diamond norm]
\label{def:diamm} 
We say two QTMs $M_1$ and $M_2$ are {\em $\veps$-indistinguishable in diamond norm}, denoted $M_1 \approx_{\diamond}^{\veps} M_2$, if for any $\sigma_n
\in \density{ \rspacen{S} \otimes \rspacen{R}}$, 
 $\rspace{R}$ being an arbitrary reference system%where $\reg{R}$ is a reference system with same dimension as $\reg{S}$ for each $n$
,
$$ \td[ (M_1\otimes \mathbb{1}_{\lin{\rspace{R}}})\sigma_n, (M_2\otimes \mathbb{1}_{\lin{\rspace{R}}})\sigma_n] \leq \veps(n)\, .$$
QIMs $M_1$ and  $M_2$ are said to be {\em indistinguishable in diamond norm}, denoted $M_1 \approx_{\diamond} M_2$, if there exists a negligible $\veps(n)$ such that $M_1$ and $M_2$ are $\veps$-indistinguishable in diamond norm.
\end{definition}

\ifnum\final=0

The two definitions above allow the distinguisher to keep a reference system. Namely part of a possibly entangled state is given as input to the machines and the other part is passed onto $\calZ$ to assist his/her decision later. If this is not permitted, we may obtain weaker definitions (strictly weaker in the statistical case. See a discussion in~\cite{Wat1120}.). We introduce them below because they are useful sometimes such as simplifying the analysis. However we stress that Definitions~\ref{def:qcm} and~\ref{def:diamm} should be the considered as the standard notions of indistinguishability in the quantum setting. 

\begin{definition}[$(t, \veps)$-weakly indistinguishable QTMs]
\label{def:wqcm} 
%Let $M_1$ and $M_2$ be two QTMs with state space $\reg{S}$ and output space $\reg{O}$. 
We say two QTMs $M_1$ and $M_2$ are {\em
$(t,\veps)$-weakly computationally indistinguishable}, denoted $M_1 \wqc^{t, \veps}
M_2$, if for any $t(n)$-time QTM $\calZ$ and any $\sigma_n
\in \density{ \rspacen{S}}$, 
$$ \left| \Pr[ \calZ(M_1(\sigma_n)) = 1 ] - \Pr [ \calZ(M_2(\sigma_n)) = 1] \right| \leq \veps(n).$$
Machines $M_1$ and  $M_2$ are called {\em weakly quantum computationally
indistinguishable}, denoted $M_1 \wqc M_2$, if for every polynomial $t(n)$, there exists a negligible $\veps(n)$ such that $M_1$
and $M_2$ are $(t,\veps)$-weakly indistinguishable.
\end{definition}

\begin{definition}[$\veps$-indistinguishable QTMs in trace norm]
\label{def:trm} 
We say two QTMs $M_1$ and $M_2$ are {\em $\veps$-indistinguishable in trace norm}, denoted $M_1 \approx_{tr}^{\veps} M_2$, if for any $\sigma_n \in\density{ \rspacen{S}}$,
$$ \td[M_1(\sigma_n), M_2(\sigma_n)] \leq \veps(n) \, $$
$M_1$ and  $M_2$ are said to be {\em indistinguishable in trace norm}, denoted $M_1 \approx_{tr} M_2$, if there is a negligible $\veps(n)$ such that $M_1$ and $M_2$ are $\veps$-indistinguishable in trace norm. 
\end{definition}
%Definition~\ref{def:trm} is strictly weaker than Defintion~\ref{def:diamm}. See an example in~\cite{Wat1120}. 
\fi
\mypar{Indistinguishability of QIMs} Next, we generalize the
definitions of 
indistinguishability above to interactive quantum machines. Let $\calZ$ and $M$ be
two QIMs, we denote $\ip{\calZ(\sigma)}{M}$ as the following process:
machine $\calZ$ is initialized with $\sigma$, %and $M$ collectively are given input state $\sigma$. 
it then provides input to $M$ and interacts with $M$. In the end, the output register of $M$ is given to $\calZ$ and $\calZ$ generates one classical bit on its own output register.
\begin{definition}[$(t,\veps)$-indistinguishable QIMs]
\label{def:qci} We say two QIMs $M_1$ and $M_2$ are {\em
$(t,\veps)$-interactively indistinguishable}, denoted $M_1
\approx_i^{t, \veps} M_2$, if for any quantum $t(n)$-time interactive
machine $\calZ$ and any mixed state $\sigma_n$ on $t(n)$ qubits,
$$ \left| \Pr[ \langle \calZ(\sigma_n), M_1 \rangle  = 1 ] - \Pr [ \langle \calZ(\sigma_n), M_2 \rangle= 1] \right| \leq \veps(n).$$
QIMs $M_1$ and  $M_2$ are called {\em quantum computationally interactively
indistinguishable}, denoted $M_1 \qcii M_2$, if for every $t(n)
\leq poly(n)$, there exists a negligible $\veps(n)$ such that $M_1$
and $M_2$ are $(t,\veps)$-interactively indistinguishable.
\end{definition}

\noindent We may call such $\calZ$ an interactive distinguisher. We can likewise define \emph{statistically interactively indistinguishable} QIMs, denoted $M_1\qsii M_2$, if we allow unbounded interactive distinguisher $\calZ$. 

\begin{remark} Quantum interactive machines, as we defined earlier,
  actually can be seen as a subset of quantum strategies, formulated
  in~\cite{GW07}. Namely, a QIM is a strategy in which each channel
  can be implemented by a uniformly generated circuit. Therefore we
  can as well define statistically interactively indistinguishability
  using the $\|\cdot\|_{\diamond r}$ norm for quantum
  strategies. See~\citet{Gut12} and \citet{CGD+08} for details about characterizing distinguishability of quantum strategies using the $\|\cdot\|_{\diamond r}$ norm.
\end{remark}
\mypar{Ideal functionalities} We sketch ideal functionalities, i.e., the programs of a trusted party in an ideal protocol, for a few basic cryptographic tasks. 

\noindent$\bullet$~Commitment $\fcom$: At ``Commit'' stage, Alice (the committer) inputs a bit $b$ and Bob (the receiver) receives from $\fcom$ a notification that a bit was committed. At ``Open'' stage, Alice can input the command {\sf open} to $\fcom$ who then sends Bob $b$. 

\noindent$\bullet$~Oblivious Transfer $\fot$: Alice (the sender) inputs $2$ bits
  $(s_0, s_1)$ and Bob (the receiver) inputs a selection bit
  $c\in\{0,1\}$. Bob receives $s_c$ from $\fot$. 
 
 \noindent$\bullet$~Zero-knowledge $\fzk$:  Let $R_L$ be an \class{NP} relation. Upon receiving $(x,w)$ from Alice, $\fzk$ verifies $(x,w) \stackrel{?}{\in}
R_L$. If yes, it sends $x$ to Bob; otherwise it instructs Bob to reject.

\noindent$\bullet$~Coin Flipping $\fcf$: Alice and Bob input the
request $1^n$ to $\fcf$, and $\fcf$ randomly chooses $r\gets\{0,1\}^n$
and sends it to Alice. Alice responds $\fcf$ with ``$\acc$'' or ``$\rej$'' indicating continuing or aborting respectively. In the case of ``$\acc$'', $\fcf$ sends $r$ to Bob and otherwise sends Bob $\perp$. Note the functionality is asymmetric in the sense that Alice gets the coins first. This avoids the complicated issue about fairness, which has been an active line of research in classical cryptography (see for example~\cite{Cleve86,Katz07,MNS09,GHKL11}) and is beyond the scope of this paper.

%Here we only describe the program of a functionality $\calF$ when both parties are honest. More precisely, one should also specify the program of $\calF$ in case of dishonest parties corrupted by an adversary. Take $\fcf$ for example, suppose that the adversary corrupts Alice and decides to abort the protocol, then $\fcf$ will send Bob $\perp$. More formal descriptions can be found in~\cite{Can01,CLOS02}. 
  
%---------------------------------------------------------------------------%
\section{Modeling Security in the Presence of Quantum Attacks}
\label{sec:qssc}
%---------------------------------------------------------------------------%

In this section, we propose a stand-alone security model for
two-party protocols in the presence of quantum attacks and show a
modular composition theorem in this model, which allows us to use
secure protocols as ideal building blocks to construct larger
protocols. 
\ifnum\final=0
In Section~\ref{subsec:variants}, we discuss variants of
our stand-alone model, and the equivalences and separations
between them.  
\else
We also discuss variants of
our stand-alone model in a unified framework. 
\fi
To be self-contained, we review in
Section~\ref{subsec:qucmodel} the quantum universal-composable (UC)
security model, which is a generalization of classical UC model to the quantum setting. 

%---------------------------------------------------------------------------%
\subsection{A General Quantum Stand-Alone Security Model}
\label{subsec:qsamodel}
%---------------------------------------------------------------------------%

\def\old{0}
\ifnum\old=1
%---------------------------------------------------------------------------%
\subsection{Security Definition}
\label{ssec:sdef}
%---------------------------------------------------------------------------%
% \mypar{Adversary Model}
A two-party protocol $\Pi$ consists of two quantum interactive
machines $\bfA$ and $\bfB$. Two players Alice and Bob that execute $\Pi$
are called honest if they run machines $\bfA$ and $\bfB$ respectively.
An adversary in $\Pi$ is one entity that corrupts some player and
controls its behavior.
We consider both {\em semi-honest} (a.k.a. {\em honest-but-curious})
and {\em malicious} adversaries. In the quantum setting, a
semi-honest adversary runs the honest protocol \emph{coherently},
that is, replacing measurements and classical operations with
unitary equivalents.
\anote{Give a reference for this
  notion (Watrous?) and a citation for ``running circuits coherently''
  (Nielsen-Chuang?).\label{anote:coherent}}

In this work, we consider only {\em static} adversaries, which corrupt a set of
players before the protocol execution starts, but do not perform
further corruptions during the protocol execution.
%
%For ease of exposition, we merge the identities of an adversary
%and the party it corrupts. Machines run by an %adversary are indicated by a $\hat{~}$
%symbol (e.g., $\hat{\bfB}$).
% We will abuse notation and use the machine name to denote the entity
% that runs the machine as well.

% \mypar{Real-world and Ideal-world executions}
Our definition of security follows the {\em simulation paradigm}
in which two modes of execution called {\em real-world} and
{\em ideal-world} are considered. A {\em real-world} execution is an
interaction between an honest party and a real-world adversary, who corrupts the other party. In the {\em ideal world}, there is a trusted party
that communicates with $\bfA$ and $\bfB$
%$\bfA_I$ and $\bfB_I$ (subscript $I$ indicates
%entities in the ideal world)
through private channels and completes
the desired task. We model the trusted party as a (classical) interactive
machine, and call it an {\em ideal functionality} $\calF$.  For example,
in the secure evaluation of a function $f$, an ideal functionality $\calF$
takes inputs $(x,y)$ from $\bfA$ and $\bfB$ respectively,
computes $(f_A, f_B) := f(x,y) $%$(x,y) \stackrel{f}{\mapsto}(f_A,f_B)$
 and hands $f_A$ (resp. $f_B$) to $\bfA$ (resp. $\bfB$). 
We then say a protocol $\Pi$ securely realizes a given task,
formulated by an ideal functionality $\calF$, if
%no
%adversary $\hat{\bfB}$ can gain more from an attack on a real
%execution of the protocol, than from an attack on an ideal world
%execution. In other words, 
for any adversary $\calA$ attacking a real-world execution, there
exists an ideal-world adversary (a.k.a {\em simulator}) $\calS$
%$\hat\bfB_I$
emulating ``equivalent''
attacks in the ideal-world. Equivalent means on any input state the
output states in the real world and ideal world are indistinguishable.

To be more specific, in the real world $\calA$ corrupts one party, say $\bfA$. Registers $S_\calA$ and $S_\bfB$ are initialized with an $n$-qubit state $\sigma_n \in \text{D}(\mathcal{H}_n)$. Then $\calA$  interacts with honest $\bfB$ and they output on $O_\calA$ and $O_\bfB$ respectively. Let the joint output state be $\sigma_n'$. Finally a QTM $\calZ$, which we call an {\em environment}, takes
$\sigma_n'$ as input and outputs one classical bit. Abstractly, we treat $\calA$ and ${\bfB}$ collectively as a noninteractive machine
$M_{\calA}$ with state space ${S}_\calA \otimes
{S}_{\bfB}$ and output space ${O}_\calA \otimes {O}_\bfB$. Namely, $M_{\calA}$ takes an input state $\sigma_n$ and internally simulates the interaction between $\calA$ and $\bfB$ and outputs $\sigma_n'$.
%\fnote{Fixed! describe how to treat machines as a single one precisely %\label{fnote:mm}}
In the ideal world, trusted party $\calF$ interacts with an ideal-world adversary $\calS$ and an honest $\bfB$. Note that there is no interaction between $\calS$ and $\bfB$ in the ideal world. Analogously, we can model $\calS$, $\bfB$ and $\calF$ as a single QTM $M_{\calS}$ with state space ${S}_{\calA}
\otimes {S}_{\bfB}$ and output space ${O}_{\calA}
\otimes {O}_{\bfB}$. Note that, although $\calS$ typically need more work space than $\calS_{\calA}$ in order to simulate $\calA$ in the ideal world, the input and output spaces should still be in accordance with that of $\calA$. We then define $\mathbf{EXEC}_{\Pi,
\calA, \calZ} :=\{\calZ(M_{\calA}(\sigma_n))\}_{n\in\bbN,\sigma_n \in \text{D}(\mathcal{H}_n)}$ and
$\mathbf{IDEAL}_{\calF, \calS, \calZ} := \{ \calZ(M_{\calS} (\sigma_n))\}_{n\in\bbN, \sigma_n \in \text{D}(\mathcal{H}_n)}$ to be the binary
distribution ensembles of $\calZ$'s output in the real-world and ideal-world executions respectively. See Fig.~\ref{fig:ri}
for an illustration of real-world and ideal-world executions.
\begin{figure}[h!] \centering \caption{Real-world and
Ideal-world Executions}\mbox{ \subfigure[Real-world execution
$\mathbf{EXEC}_{\Pi, \hat\bfB, \calZ}$]
{\includegraphics[width=0.45\columnwidth]{real}}\label{fig:real}
\qquad \subfigure[Ideal-world execution $\mathbf{IDEAL}_{\calF,
\hat{\bfB}_I, \calZ}$]{\includegraphics[width=0.45\columnwidth]{ideal}
}\label{fig:ideal}}\label{fig:ri}
\end{figure}%\vspace{-5mm}

\begin{definition}
\label{def:qsa}(Quantum Stand-alone Secure Emulation). Let $\calF$ be a poly-time
two-party functionality and let $\Pi$ be a two-party protocol. We
say $\Pi$ {\em quantum stand-alone-emulates} $\calF$, if for any
poly-time QIM $\calA$, there is a poly-time QIM $\calS$, such
that for any poly-time QTM $\calZ$,  $\mathbf{EXEC}_{\Pi, \calA, \calZ} \approx
\mathbf{IDEAL}_{\calF, \calS, \calZ}$. \fnote{somewhere say poly-time
protocol: honest machines are poly in n}
%, where all machines are polynomial-time.
\end{definition} %\fnote{Fixed. say focus on computational security. do we want to consider environments that take quantum aux input?}

%\anote{Fixed. This is the place to explain the equivalent formulation in
 % terms of distinguishability of QTMs.}

\mypar{Remark}
\begin{inparaenum}[\upshape(I\upshape)]
\item Equivalently, the definition can be formulated as: for any $\calA$, there exists
$\calS$, such that QTMs $M_{\calA}$ and $M_{\calS}$ are
indistinguishable, i.e., $M_{\calA} \approx_{wqc} \calM_{\calS}$ as per Definition~\ref{def:dm}.
\item We focus on computational security in this work, and the model
extends to information-theoretical setting straightforwardly.
%\item  We stress that $\sigma$
%not only encodes the inputs to the players, but %also contains
%auxiliary system $\calW$ that might be entangled %with the inputs and
%moreover serves as quantum advice to later %assist $\calZ$ in
%distinguishing the two worlds. There are other possible choices in
%the definition, e.g., disallowing auxiliary %system $\calW$ and only
%giving $\calZ$ classical advice, which may give %rise to variants that
%coincide with or subsume existing models. See Appendix~\ref{app:def}
%for a thorough discussion.
\item Despite the simplicity of the model we present here, it is actually equivalent to a few seemingly stronger variants. For example, the environment $\calZ$ may take a quantum advice which could be even entangled with the input state. See Appendix~\ref{app:def} for a thorough discussion. \item Though this work mainly focuses on feasibility of functionalities that evaluates a \emph{classical, deterministic} function, our model captures functionalities $\calF$ that can be randomized and reactive (i.e., consisting of multiple stages). We can even model a functionality as a quantum interactive machine, and study secure emulation of quantum functionalities (e.g., jointly evaluating a quantum circuit) in our model. 
\item For technical reasons, we require functionalities to be {\em
well-formed} and protocols to be {\em nontrivial}, which are
satisfied by all functionalities and protocols in our paper.
(See~\cite[Section 3]{CLOS02} for details.) 
\end{inparaenum} \anote{Remark needs updating.}\fnote{how about now?}

%---------------------------------------------------------------------------%
\subsection{Modular Composition}
\label{ssec:mcomp}
%---------------------------------------------------------------------------%

% \mypar{Hybrid Model}
It is common practice in the design of large protocols that a given task is divided into several subtasks. We first realize each subtask, and then use these modules as building blocks (subroutines) to realize the
initial task. We formalize this paradigm by {\em hybrid
models}\footnote{In contrast, we call it a {\em plain model} if
there are no trusted parties and no trusted setup assumptions like
common reference string or public-key infrastructure, etc.}. A
protocol in the {\em $\calG$-hybrid model}, denoted $\Pi^\calG$, has
access to a trusted party that implements ideal functionality $\calG$.
As before, for each adversary $\calA$
%$\hat\bfB_H$ (subscript $H$ indicates
%entities in a hybrid model)
in the $\calG$-hybrid model, we can define
$M_{\calA}$ and $\mathbf{EXEC}_{\Pi^{\calG}, \calA, \calZ}$ likewise. Then we say $\Pi^{\calG}$ quantum stand-alone-emulates ideal
functionality $\calF$ {in the $\calG$-hybrid model} if for any poly-time QIM $\calA$, there is a poly-time QIM $\calS$
$\mathbf{EXEC}_{\Pi^{\calG}, \calA, \calZ} \approx \mathbf{IDEAL}_{\calF,
\calS, \calZ}$.

%\mypar{Composing Protocols Modularly}
Now suppose we have $\Pi_1^\calG$
in the $\calG$-hybrid model and a protocol $\Pi_2$ realizing $\calG$. The
operation of replacing an invocation of $\calG$ with an invocation of
$\Pi_2$ is done in the natural way: machines in $\Pi_1$ initialize
machines in $\Pi_2$ and pause; machines in $\Pi_2$ execute $\Pi_2$
and generate outputs; then $\Pi_1$ resumes with these outputs. We
denote the composed protocol $\Pi_1^{\Pi_2}$. We can show that our model allows for modular compositions.

\begin{theorem}(Modular Composition Theorem) Let $\Pi_1^\calG$ be a two-party
  protocol that quantum stand-alone-emulates $\calF$ in the
  $\calG$-hybrid model and let $\Pi_2$ be a two-party protocol that
  quantum stand-alone-emulates $\calG$.  Then the composed
  protocol $\Pi_1^{\Pi_2}$ quantum stand-alone-emulates $\calF$.
\label{thm:mcomp}
\end{theorem}

\mypar{Remark} The proof comes in Appendix~\ref{app:pmct}. It is easy
to extend our analysis to a more general case where $\Pi$ can invoke
$\calG$ multiple times and also access polynomially many ideal
functionalities ($\calG_1, \calG_2, \ldots$). However, we stress that at
each round, only one functionality is invoked for at most once.
\else

%===== new section on security model =====%

Our model follows the \emph{real-world/ideal-world} simulation paradigm. It proceeds in three high-level steps: 
\begin{inparaenum}[(i\upshape)]
\item Formalizing the process of executing a protocol in the presence of adversarial activities. 
\item Formalizing an \emph{ideal-world} protocol for realizing the desired task. This is an (imaginary) idealized protocol which captures the security guarantees we want to achieve. 
\item Finally we say a ({real-world}) protocol realizes a task securely if it ``behaves similarly'' to the {ideal-world} protocol for that task (Definition~\ref{def:qsae}). ``Behaving similarly'' is formalized by the notion of stand-alone \emph{emulation} between protocols (Definition~\ref{def:qcsa},~\ref{def:qssa}). 
\end{inparaenum}  %Here we give a brief overview of the model, and the detailed description can be found in Section~\ref{subsec:model}

Our definition can be viewed in two ways: either as a quantum analogue
of Canetti's classical stand-alone model~\cite{Can00} or as a relaxed
notion of (a variant of) Unruh's quantum UC
security~\cite{Unr10}. Prior to our work, stand-alone security
definitions for quantum attacks were largely developed \emph{ad
  hoc}\footnote{E.g., \citet{FS09} write: "It is still common practice
  in quantum cryptography that every paper proposes its own security
  definition of a certain task and proves security with respect to the
  proposed definition. However, it usually remains unclear whether
  these definitions are strong enough to guarantee any kind of
  composability, and thus whether protocols that meet the definition
  really behave as expected."}; the first systematical treatments
appear in~\cite{FS09,DFLSS09}. Our model generalizes the existing
model of Damg{\aa}rd \etal~\cite{DFLSS09} in two ways. First, our
model allows protocols in which the functionalities can process
quantum information (rather than only classical
functionalities). Second, it allows adversaries that take arbitrary
quantum advice, and for arbitrary entanglement between honest and
malicious players' inputs.  This distinction is reflected in the
composability that the model provides (see details in
Section~\ref{sssec:sacomposition}). While the composition results of
Damg{\aa}rd \etal~allow only for classical high-level protocols, our
result holds for arbitrary quantum protocols.

%-------------------------------------------------%
\subsubsection{The Model}
\label{sssec:model}
%-------------------------------------------------%

We describe our model for the two-party setting; it is straightforward
to extend to multi-party setting. We first introduce a few important
objects in our model. We formalize a cryptographic task by an
interactive machine called a \emph{functionality}. It contains the
instructions to realize the task, and we usually denote it by $\calF$ or
$\calG$. While our model applies to both classical and quantum functionalities, our focus in this work will be efficient classical
functionalities. Namely $\calF$ is a classical probabilistic
polynomial-time machine. A two-party protocol for a task $\calF$
consists of a pair of interactive machines $(A,B)$. We call a protocol
poly-time if $(A,B)$ are both poly-time machines. We typically use
Greek letters (e.g., $\Pi$) to denote protocols. If we want to
emphasize that a protocol is classical, i.e., computation and all
messages exchanged are classical, we then use lower-case letters
(e.g., $\pi$). Finally, an adversary, usually denoted $\calA$ or
$\calS$, is another interactive machine that intends to attack a
protocol. Very often we abuse notation and do not distinguish a
machine and the player that runs the machine. This should not cause
any confusion.

\mypar{Protocol Execution}  We consider executing a protocol $\Pi=(A,B)$ in the presence of an adversary $\calA$. Their state registers are initialized by a secure parameter $1^n$ and a joint quantum state $\sigma_n$. 
Adversary $\calA$ gets activated first and coordinates the execution. Specifically, the operations of each party are: 
\begin{itemize}
	\item Adversary $\calA$: it may either \textbf{deliver} a message to some party or \textbf{corrupt} a party. Delivering a message is simply instructing the designated party (i.e., the receiver) to read the proper segment of his network register. We assume all registers are authenticated so that $\calA$ cannot modify them and in particular if the register is private to the party, $\calA$ may not read the content. Other than that, $\calA$ can for example schedule the messages to be delivered in any arbitrary way. If $\calA$ corrupts a party, the party passes all of its internal state to $\calA$ and follows the instructions of $\calA$. In the two-party setting, corrupting a party can be simply thought of as substituting the machine of $\calA$ for the machine of the corrupted party.

	\item Parties in $\Pi$: once a party receives a message from $\calA$, it gets activated and runs its machine.  At the end of one round, some message is generated on the network register. Adversary $\calA$ is activated again and controls message delivery. At some round, the party generates some output and terminates. 
\end{itemize}

Clearly, we can view $\Pi$ and $\calA$ as a whole and model the
composed system as another QIM, call it $\mac_{\Pi, \calA}$. Then
executing $\Pi$ in the presence of $\calA$ is just running $\mac_{\Pi, \calA}$ on some input state, which may be entangled with a reference system that will be handed to the distighuisher. 

\mypar{Protocol Emulation} As indicated earlier, a secure protocol is supposed to ``emulate'' an idealized protocol. Here we formally define emulation between protocols. Let $\Pi$ and $\Gamma$ be two protocols. Let $\mac_{\Pi,\calA}$ be the composed machine of $\Pi$ and an adversary $\calA$, and $\mac_{\Gamma,\calS}$ be that of $\Gamma$ and another adversary $\calS$. Informally, $\Pi$ emulates $\Gamma$ if the two machines $\mac_{\Pi,\calA}$ and $\mac_{\Gamma,\calS}$ are indistinguishable. 

%\fnote{Not fixed! define combination-operation; maybe call it composed machine?}

\begin{definition}[Computationally Quantum-Stand-Alone Emulation] Let $\Pi$ and $\Gamma$ be two poly-time protocols. We say $\Pi$ \emph{computationally quantum-stand-alone} (\cqsa) emulates $\Gamma$, if for any poly-time QIM $\calA$ there exists a poly-time QIM $\calS$ such that $\mac_{\Pi,\calA} \qc \mac_{\Gamma,\calS} \ .$
\label{def:qcsa}
\end{definition}

%If we allow computationally \emph{unbounded} adversaries and environments, we get \emph{statistical} emulation.
 
\begin{definition}[Statistically Quantum-Stand-Alone Emulation] Let $\Pi$ and $\Gamma$ be two poly-time protocols. We say $\Pi$ \emph{statistically quantum-stand-alone} (\sqsa) emulates $\Gamma$, if for any QIM $\calA$ there exists an QIM $\calS$ that runs in poly-time of that of $\calA$, such that $\mac_{\Pi, \calA} \dmm \mac_{\Gamma, \calS}$.  
\label{def:qssa}
\end{definition} 

%\fnote{update picture}
\begin{figure}[h!]
\centering
\ifnum\ijqi=1
{\includegraphics[width=3.5in]{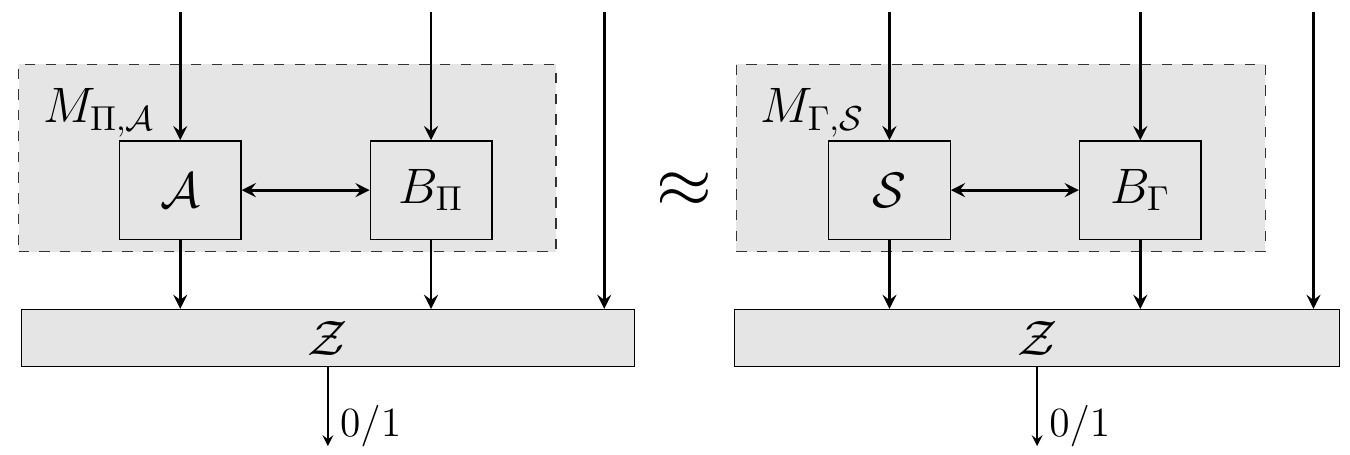}} 
\else
%{\includegraphics[width=3.5in]{Figures/emuprot}} 
{\includegraphics[width=3.5in]{Figures/emuprotref}} 
\fi
\caption{Quantum stand-alone emulation between protocols.} \label{fig:emuprot}
\end{figure}%\fnote{Fig. needs update}

\begin{remark}  \begin{inparaenum}[(i\upshape)] 
\item The adversary $\calS$ is usually called a simulator because typical constructions of $\calS$ simulate the given $\calA$ internally. 
\item In the statistical setting, we require the complexity of $\calS$
  and $\calA$ to be polynomially related. This ensures that the statistical
  notion actually implies the computational one. See \citet{Can00} for discussion
  of this issue in the classical context.
\end{inparaenum}
\end{remark}

\mypar{Ideal-world Protocol and Secure Realization of a Functionality}
We formalized protocol emulation in a general form which applies to
any two protocols. But it is of particular interest to emulate a
special type of protocol which captures the security guarantees we
want to achieve.  We formalize the so-called \emph{ideal-world}
protocol $\tilde \Pi_\calF$ for a functionality $\calF$.  In this
protocol, two (dummy) parties $\tilde A$ and $\tilde B$ have access to
an additional ``trusted'' party that implements $\calF$. We may abuse
notation and call the trusted party $\calF$ too. Basically $\tilde A$
and $\tilde B$ invoke $\calF$ with their inputs, and then $\calF$ runs
on the inputs and sends the respective outputs back to $\tilde A$ and $\tilde B$. An execution of $\tilde \Pi$ with an adversary $\calS$ is similar to our prior description for executing a (real-world) protocol, except that $\calF$ cannot be corrupted. Likewise, we denote the composed machine of $\calF$ and $\tilde \Pi_\calF$ as $\mac_{\calF, \calS}$.  We state the definition in the computational setting; statistical emulation is defined analogously. 

\begin{definition}[\cqsa~Realization of a Functionality] Let $\calF$ be a poly-time
two-party functionality and $\Pi$ be a poly-time two-party protocol. We
say $\Pi$ {\em computationally quantum-stand-alone} realizes $\calF$, if $\Pi$ \cqsa~emulates $\tilde \Pi_\calF$. Namely, for any poly-time $\calA$, there is a poly-time $\calS$ such
that $\mac_{\Pi, \calA} \qc \mac_{\calF, \calS}$. 
\label{def:qsae}
\end{definition}

\begin{figure}[h!]
\centering
\ifnum\ijqi=1
{\includegraphics[width=3.5in]{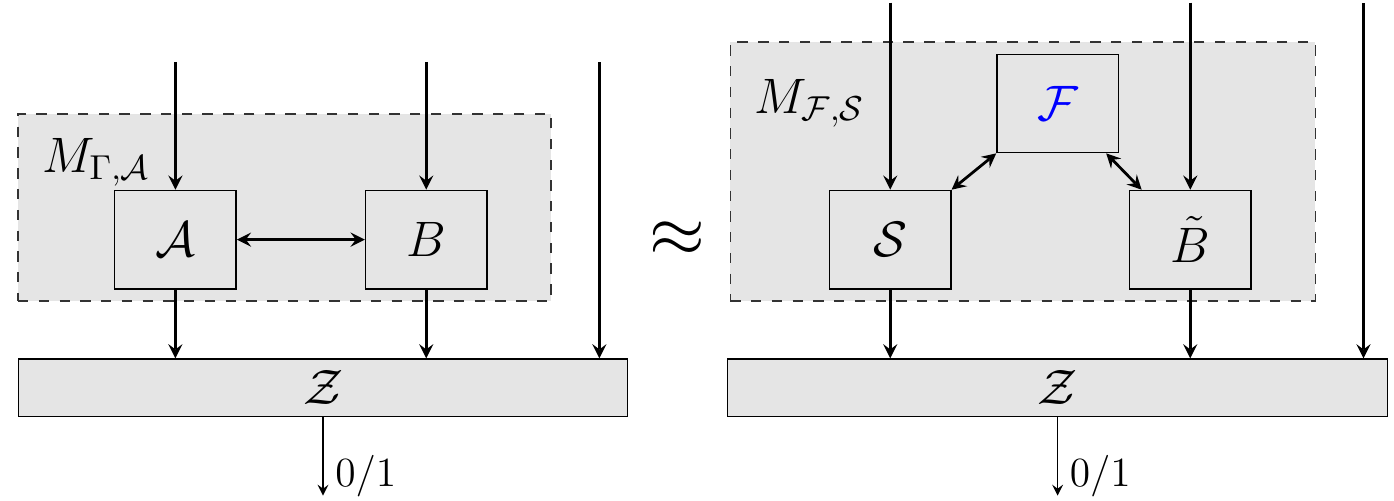}} 
\else
%{\includegraphics[width=3.5in]{Figures/emufunc}} 
{\includegraphics[width=3.5in]{Figures/emufuncref}} 
\fi
\caption{Quantum Stand-alone Realization of a functionality.} \label{fig:emufunc}
\end{figure}%\fnote{Fig. needs update}

It is conventional to use %$\exec_{\Pi,\calA,\calZ} := \{ \calZ(M_{\Pi,\calA} (1^n, \sigma_n)) \}_{n \in \bbN}$ 
$\exec_{\Pi,\calA,\calZ} := \{ \calZ((M_{\Pi,\calA}\otimes\mathbb{1}_{\lin{\rspace{R}}} ) (\sigma_n)) \}_{n \in \bbN}$ to denote the binary output distribution ensemble of $\calZ$ that runs on the output state of an execution of $\Pi$ and $\calA$ with input $(1^n, \sigma_n)$.  Likewise, %$\ideal_{\calF,\calS,\calZ}:= \{\calZ(M_{\calF,\calS} (1^n, \sigma_n)) \}_{n \in \bbN}$
$\ideal_{\calF,\calS,\calZ}:= \{\calZ((M_{\calF,\calS}\otimes \mathbb{1}_{\lin{\rspace{R}}})  (\sigma_n)) \}_{n \in \bbN}$ denotes  the
binary output distribution ensemble of $\calZ$ in an execution of the
ideal-world protocol $\tilde \Pi_{\calF}$.  Definition \ref{def:qsae}
can be restated as requiring that for any poly-time $\calA$ there
exists a poly-time $\calS$ such that, for any poly-time $\calZ$ and
state $\sigma_n$, we have $\exec_{\Pi,\calA,\calZ} \approx \ideal_{\calF, \calS,\calZ} \ .$ 

%We can think of the distinguisher implicit in both the definitions above as part of an \emph{environment}, which prepares an (arbitrary) entangled input for the parties, and processes their outputs to generate a single bit. 
%
\ifnum\final=0
The security definitions above use the strong indistinguishable notions for quantum machines.  It is natural to ask what happens if we apply weaker indistinguishable notions instead. Namely, we could require that $\mac_{\Pi, \calA} \wqc \mac_{\Gamma, \calS}$ and $\mac_{\Pi, \calA} \trm \mac_{\Gamma, \calS}$ respectively.  Somewhat surprisingly, in both computational and statistical settings, the definitions of protocol emulation (hence secure realization of ideal functionality too) are equivalent under two distinguishability notions. 
\begin{theorem}[Informal]
 Definitions \ref{def:qcsa} and \ref{def:qssa} remain the same
  whether or not the distinguisher can run the protocol on a state that is entangled with an outside reference system.
\end{theorem}
We formalize this claim and discuss other variants of the security definitions in Section \ref{subsec:variants}. We remark that because of the equivalence, we usually use the definitions under weak indistinguishability in the proofs and analysis (to ease notation for example). Nonetheless, conceptually we should keep in mind the security notions under the stronger indistinguishability. 
\fi

%Likewise,  we can define $\ideal_{\calF, \calA, \calZ}$ as the output distribution ensemble of $\calZ$ in an execution. Then we can rephrase the condition as \emph{``$\ldots$ for any poly-time $\calA$, there is a poly-time $\calS$, such
%that for any poly-time $\calZ$: $ \exec_{\Pi,\calA,\calZ}  \approx \ideal_{\calF, \calS,\calZ}$''}. 

%\begin{remark}
%\begin{inparaenum}[ i\upshape)]
	% \item If there is a protocol that emulates a functionality $\calF$ according to our definition(s), we often say $\calF$ can be \emph{realized}. 
%	\item The adversary $\calS$ is usually called a simulator because the typical constructions of $\calS$ simulate the given $\calA$ internally.
%	\item We may say ``Quantum Stand-alone (QSA) emulation'' and call our model ``QSA model'' without specifying the exact setting (\emph{computational, statistical, perfect}), when it is clear from context, or the statement applies to all settings. 
%\end{inparaenum}

%\end{remark}
\mypar{Types of Attack} Typically, we need to speak of security
against a specific class of adversaries. We have distinguished two
classes of adversaries according to their computational complexity,
i.e., poly-time vs. unbounded time. We also categorize adversaries
according to how they corrupt the parties and how they deviate from the
honest behavior defined by the protocol. The standard two types of
corruptions considered in the literature are \emph{static}
vs. \emph{adaptive} corruptions. Under {static} corruption, the
identities of corrupted parties are determined before protocol starts. In contrast, {adaptive}
corruption allows an adversary to change the party to corrupt
adaptively during the execution. This work only concerns static
corruption. 

In terms of what dishonest behaviors are permitted for an adversary,
again two classes are considered standard in the literature:
\emph{semi-honest} (a.k.a. \emph{passive or honest-but-curious}) and
\emph{malicious} (a.k.a. \emph{active}). A semi-honest adversary,
after corrupting a party, still follows the party's circuit, except
that in the end it processes the output and the state of the party. A
malicious adversary, however, can substitute any circuit for the
corrupted party. In the definitions of the protocol emulation, unless
otherwise specified, the two adversaries in the real-world and
ideal-world must belong to the same class. For example, if $\calA$ is
semi-honest, $\calS$ must also be
semi-honest. %\textbf{Unless otherwise specified, all statements in this thesis refer to \emph{malicious} adversaries with \emph{static} corruption}.

These notions of different classes of adversaries naturally extend to
quantum adversaries, except for one subtlety in defining semi-honest
quantum adversaries. There are two possible definitions. One
definition, which may be referred to as the Lo-Chau-Mayers semi-honest
model~\cite{LC97,Mayers97}, allows $\calA$ to run the circuit of the
corrupted party, which is specified by the protocol,
\emph{coherently}. Namely $\calA$ purifies the circuit of corrupted
party so that all operations are unitary. For example, instead of
measuring a quantum state, the register is ``copied'' by a CNOT
operation to an ancillary register. Another definition forces the
adversary to exactly faithfully follow the corrupted party's circuit during
the protocol execution, so that any quantum measurement occurs
instantaneously and possibly destructively. In other words, in the second model, a semi-honest
quantum adversary $\calA$ only corrupts a party at the end of the
protocol execution, and then processes the internal state and
transcript that the corrupted party holds.  This second model is generally weaker than
the first, in the sense that the adversary is more restricted.
In
this paper, we focus on the second of these two notions.
\anote{I don't see why the two models are equivalent in general. I
  think we only need the second model, though, for our proofs of
  security. Is that correct?}
\fnote{Yes, I think we are using the 2nd notion in proving semi-honest UC protocols. Where do we use the first notion in the paper?}

%========================%
\subsubsection{Modular Composition Theorem}
\label{sssec:sacomposition}
%========================%

It is common practice in the design of large protocols to divide a
task into several subtasks. We first realize each subtask, and then
use these modules as building blocks (subroutines) to realize the
initial task. In this section, we show that our definition allows such
modular design.

\mypar{Composition Operation} Let $\Pi$ be a protocol that uses
another protocol $\Gamma$ as a subroutine, and let $\Gamma'$ be a
protocol that QSA emulates $\Gamma$. We define the \emph{composed}
protocol, denoted $\cprot{\Pi}{\Gamma}{\Gamma'}$, to be the protocol
in which each invocation of $\Gamma$ is replaced by an invocation of
$\Gamma'$. We allow multiple calls to a subroutine and also using
multiple subroutines in a protocol $\Pi$. However, we require that at
any point, only \textbf{one} subroutine call be in progress; that is,
we handle \emph{sequential} composition. This is weaker than the
``network'' setting, where many instances and subroutines may be
executed \emph{concurrently}.

We can show that our quantum stand-alone model admits a modular composition theorem. 

\begin{theorem}[Modular Composition: General Statement]
Let $\Pi$, $\Gamma$ and $\Gamma'$ be two-party protocols such that $\Gamma'$ \cqsa~(resp. \sqsa) emulates $\Gamma$, then $\cprot{\Pi}{\Gamma}{\Gamma'}$ \cqsa~(resp. \sqsa) emulates $\Pi$. 
\label{thm:sacomposition}
\end{theorem}

The proof can be found in Appendix~\ref{sec:pmct}. Here we discuss an
important type of protocol where the composition theorem is especially
useful.

\mypar{Protocols in a Hybrid Model} We next define a \emph{hybrid}
model, in which the parties can make calls to an ideal-world protocol
$\tilde\Pi_{\calG}$ of some functionality $\calG$\footnote{In
  contrast, we call it the \emph{plain model} if no such trusted
  set-ups are available.}. We call such a protocol a
\emph{$\calG$-hybrid} protocol, and denote it $\Pi^{\calG}$. The
execution of a hybrid-protocol in the presence of an adversary $\calA$ proceeds in the usual way. %As mentioned earlier, we assume an mechanism authenticated  All protocols considered in this thesis % described before. 

Now assume that we have a protocol $\Gamma$ that realizes $\calG$ and
we have designed a $\calG$-hybrid protocol $\Pi^{\calG}$ realizing
another functionality $\calF$.  Then the composition theorem allows us
to treat sub-protocols as equivalent to their ideal versions when
analyzing security of a high-level protocol. %For instance, one can treat a coin-flipping protocol as a trusted party who hands all participants a uniformly random string.
\begin{cor}[Modular Composition: Realizing Functionalities]
Let $\calF$ and $\calG$ be poly-time functionalities. Let $\Pi^{\calG}$ be a $\calG$-hybrid protocol that \cqsa~(resp. \sqsa) realizes $\calF$, and $\Gamma$ be a protocol that \cqsa~(resp. \sqsa) realizes $\calG$, then $\cprot{\Pi}{\calG}{\Gamma}$ \cqsa~(resp. \sqsa) realizes $\calF$. 
\label{cor:sacomposition}
\end{cor}

\begin{figure}[h!]
\centering
\ifnum\ijqi=1
\mbox{ \subfigure{\includegraphics[width=2in]{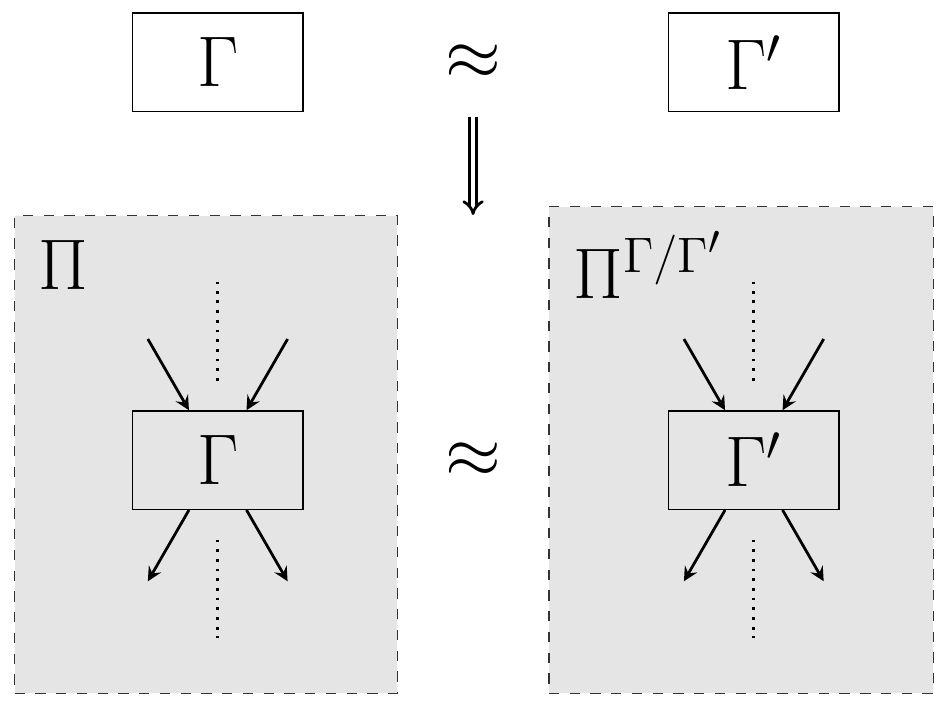}} \qquad
\subfigure{\includegraphics[width=2in]{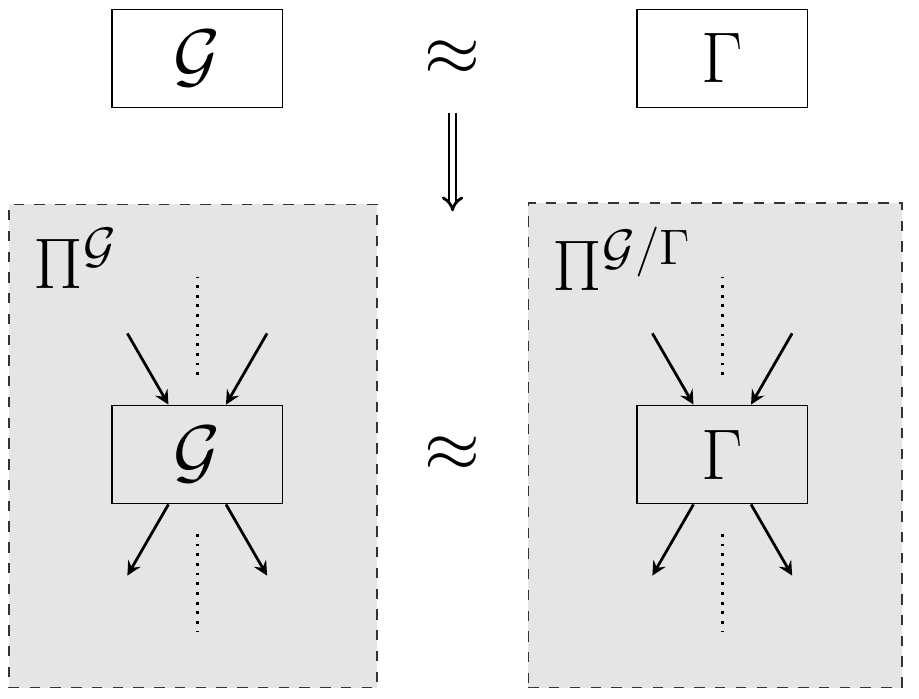}}}
\else
\mbox{ \subfigure{\includegraphics[width=2in]{Figures/compgeneral}}\qquad
\subfigure{\includegraphics[width=2in]{Figures/composeideal}
}}
\fi
\label{fig:modcomp}

 \caption{Illustration of modular composition theorem: the general case (left) and in hybrid model (right).} 
\end{figure}%\fnote{Fig. needs update}
\ifnum\final=0
%============================================%
\subsection{Variants of Quantum Stand-Alone Models: A Unified Framework}
\label{subsec:variants}
%============================================%
\else
%============================================%
\subsubsection{Variants of Quantum Stand-Alone Models: A Unified Framework}
\label{sssec:variants}
%============================================%
\fi
When defining a security model, there are lots of choices qualifying
and quantifying the power of the adversaries to account for various
security requirements. Here we provide an abstract stand-alone model
for both classical and quantum cryptographic protocols, illustrated in
Figure~\ref{fig:def}, which contains three natural choices for the
adversaries which we think are essential. This abstract model captures
all existing stand-alone security models (including ours) and this
allows for a unified study of, and comparison among,
these models. 
The relationship between these models may be interesting beyond the
study of SFE. 
\anote{In fact, it may simplify some of the literature on
  device-independent QKD.}
\fnote{That's interesting! any reference?}

% We believe that studying these models is not only of interest to the cryptography community, but also helps understand fundamental questions in areas like quantum information and quantum computational complexity theory.

\begin{figure}[h!]
\centering
\ifnum\ijqi=1
{\includegraphics[width=2in]{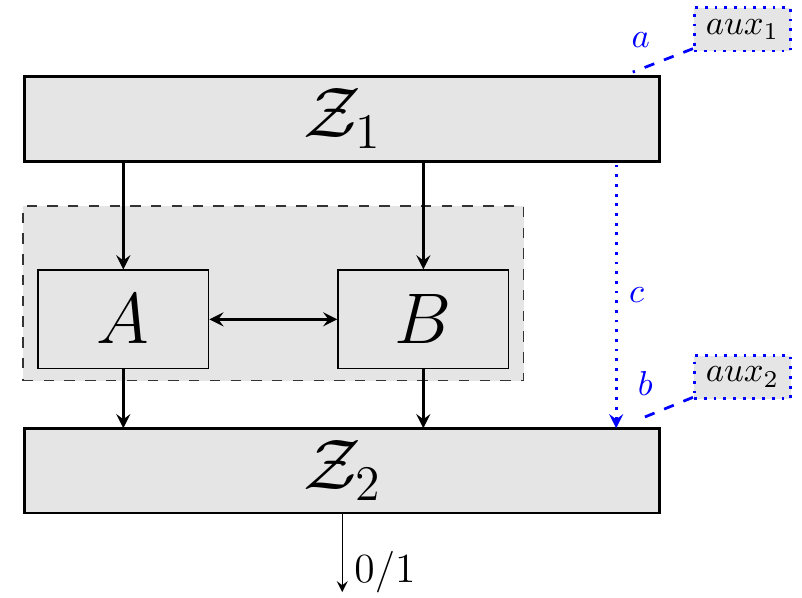}} 
\else
{\includegraphics[width=2in]{Figures/generalmodel}} 
\fi
\caption{Possible choices in defining a security model} \label{fig:def}
\end{figure}%\fnote{Fig. needs update}

The model contains an environment $\calZ$ and a protocol.
Depending on whether the protocol is in real or ideal world,
we have the honest party, the (real-or ideal-world) adversary and possibly the trusted party. Here we think of the environment as two separate machines: $\calZ_1$, which we may call an input sampler,  prepares inputs to the players; and $\calZ_2$ that receives outputs and makes the decision. Now we consider the following choices:

\renewcommand{\labelenumi}{(\alph{enumi})}
\begin{enumerate}%[(a)]
\item Does $\calZ_1$, the input sampler, have a quantum advice $aux_1$? In other words, do we allow arbitrary input states or only states that can be generated efficiently?
\item Does $\calZ_2$, which is essentially a distinguisher, take quantum advice $aux_2$?
\item Does $\calZ_1$ pass a state to $\calZ_2$? Namely, does the environment keep state during the execution?
\end{enumerate}

Notice that positive answers potentially give more power to the adversaries and thus provides stronger security guarantee. Also note that all machines are always allowed to take classical advice. We may denote a security model as $\model{\cdot}{\cdot}{\cdot}$ where the subscripts are from $\{\a,\abar,\b,\bbar,\c,\cbar\}$ indicating each of the choices made for the model. For example $\model{\a}{\bbar}{\c}$ corresponds to the model that $\calZ_1$ gets quantum advice; $\calZ_2$ takes no quantum advice and $\calZ_1$ passes state to $\calZ_2$ --this exactly leads to  our model in Def.~\ref{def:qcsa}. 
Similarly, $\model{\a}{\b}{\c}$ is the model where $\calZ_1$ and $\calZ_2$ both
take quantum advice, and there is state passing from $\calZ_1$ to
$\calZ_2$. 

We say two models $\calM$ and $\calM'$ are equivalent if % it holds that $\Pi$ is secure in model $\calM$ if and only if it is secure in $\calM'$. More generally 
for any two protocols $\Pi$ and $\Gamma$, it holds that $\Pi$ emulates
$\Gamma$ in $\calM$ \emph{if and only if} $\Pi$ emulates $\Gamma$ in $\calM'$. It is conceivable that some of the $2^3=8$ combinations collapse to the
same model. For example, if all players are classical circuits, then
all eight models $\model{\cdot}{\cdot}{\cdot}$
collapse. This is because classical (non-uniform) machines can only measure a
quantum state in computational basis to obtain a classical string
from a certain distribution. But a classical circuit can be hardwired
with any classical string, and so (quantum)
advice gives no extra power to a classical
circuit. Passing state likewise becomes vacuous. 

When we consider an adversary and environment consisting of quantum circuits, the
situation becomes generally more complicated. We can observe that
choice (b) becomes irrelevant once we permit arbitrary input state and
state passing (i.e., $\model{\a}{\b}{\c} \equiv
\model{\a}{\bbar}{\c}$). We conjecture that state passing makes no
difference either.  If this is indeed true, then all the variants
collapse when $\calZ_1$ takes quantum advice. On the other hand, if $\calZ_1$
takes no advice (i.e. only efficiently generated input states are
allowed), we are left with two variants $\model{\abar}{\b}{\cdot}$ and
$\model{\abar}{\bbar}{\cdot}$. The relationship between these two
models is closely related to the fundamental question in quantum
complexity theory regarding
$\class{BQP}/{\mathrm{poly}}\stackrel{?}{=}\class{BQP}/{\mathrm{\mathbf{q}poly}}$. We
leave further investigations as future work. In Appendix~\ref{sec:mc},
we discuss another variant that appears in the
literature~\cite{DFLSS09,FS09}, in which $\calZ_1$ may only generate input
states of a special form. We show that this does not change the model in the case that $\calZ_1$ takes quantum advice. 

%We give almost complete characterizations below, focusing on the computational setting. Namely unless otherwise specified, all parties are poly-time quantum machines. 

%\mypar{Models with Quantum Advice to $\calZ_1$} 
\ifnum\final=0
\subsubsection{Models with Quantum Advice to $\calZ_1$} 
\label{ssec:wadv}

Perhaps surprisingly, we show that
once $\calZ_1$ takes quantum advice, the choices for (b) and (c) no
longer matter, and we end up with a single model. %In particular, this justifies our choice in Definition~\ref{def:qcsa}. 

\begin{theorem} The models where $\calZ_1$ passes state to $\calZ_2$
  are equivalent to the models without state passing. Namely, 
$$\model{\a}{\bbar}{\cbar}\equiv \model{\a}{\bbar}{\c} \equiv \model{\a}{\b}{\c} \equiv\model{\a}{\b}{\cbar} \ . $$
\vspace{-5mm}
\begin{figure}[h!]
\centering
\ifnum\ijqi=1
{\includegraphics[width=5in]{equiv1}} 
\else
{\includegraphics[width=5in]{Figures/equiv1}} 
\fi
\caption{Quantum advice to $\calZ_1$ collapses all models}
\label{fig:equiv1}
\end{figure}
\label{thm:eq_sp_body}
\end{theorem}

%We defer the proof, which is quite technical, to Appendix~\ref{subsec:wadv}.  
The main proof idea is that we can simulate passing of state from $\calZ_1$ to $\calZ_2$ by having $\calZ_1$ authenticate the state (using a quantum authentication scheme) and
pass it to $\calZ_2$ via the adversary. The classical key for the
authentication scheme can be hardwired into the circuits of $\calZ_1$
and $\calZ_2$.  The authentication scheme ensures that a simulator
playing the role of the adversary cannot modify the state en route
(without being easily distinguishable from a real adversary).

Our proof requires a careful definitional treatment of quantum
authentication schemes (we provide an apparently stronger definition
than the one in \citet{BCGST02},
which is nonetheless satisfied by their constructions). %See Section~\ref{subsec:wadv} for details.

\begin{proof}[Proof of Theorem~\ref{thm:eq_sp_body}]
We will show that $\model{\a}{\bbar}{{\cbar}} \equiv \model{\a}{\bbar}{{\c}}$ in Proposition~\ref{prop:aux1_sp}. This is the most interesting as well as technically heavy part. $\model{\a}{\b}{{\cbar}} \equiv \model{\a}{\b}{{\c}}$ can be shown by almost the same argument, and hence we omit the details. $\model{\a}{\b}{{\c}} \equiv \model{\a}{\bbar}{{\c}}$ follows by a simple observation. When $\calZ_1$ takes quantum advice and passes state to
    $\calZ_2$, $\calZ_2$ does not need to explicitly take an extra
    advice, because this state can just be part of $\calZ_1$'s
    quantum advice, and is later passed to $\calZ_2$ by $\calZ_1$. Combining these proves the theorem. 
\end{proof}

\begin{prop}  
$\model{\a}{\bbar}{{\cbar}} \equiv \model{\a}{\bbar}{{\c}}  \ .$ 
\label{prop:aux1_sp}
\end{prop}

%\begin{proof}[Proof of Proposition~\ref{prop:eq_sp}]

Denote $ \calM':= \model{\a}{\bbar}{\cbar}$ and $\calM := \model{\a}{\bbar}{\c}$. The proof goes in two directions.

\begin{lemma}
Let $\Pi$ and $\Gamma$ be two protocols. If $\Pi$ emulates $\Gamma$ in model $\calM$, then $\Pi$ also emulates $\Gamma$ in model $\calM'$.
\label{lemma:eqsp1}
\end{lemma}
This is obvious, because $\calZ_1$'s that do not pass states form a subset of those who do. The nontrivial part is the other direction.

\begin{lemma}
Let $\Pi$ and $\Gamma$ be two protocols. If $\Pi$ emulates $\Gamma$ in model $\calM'$, then $\Pi$ also emulates $\Gamma$ in model $\calM$.\label{lemma:eqsp2}
\end{lemma}

Let $\reg{W}$ be the system that $\calZ_1$ passes to
$\calZ_2$. Intuitively, we could imagine that passing register
$\reg{W}$ from $\calZ_1$ to $\calZ_2$ as in model $\calM'$ can be done
in model $\calM'$  by handing the register to the adversary, and
having the adversary hand it ti $\calZ_2$. Specifically, for any
adversary $\calA$ acting in model $\calM$, we construct an 
adversary $\calA'$ acting in model $\calM'$. Adversary $\calA'$ receives both $\reg{S}_{\calA}$ and $\reg{W}$, applies $\calA$ on $\reg{S}_{\calA}$ and leaves $\reg{W}$ untouched. Since $\Pi$ emulates $\Gamma$ in model $\calM'$ by assumption, there exists an adversary $\calS'$ attacking $\Gamma$ that simulates $\calA'$. Finally we construct an adversary $\calS$ out of $\calS'$, in hope that it simulates $\calA$ according to model $\calM$. However, the difficulty is that $\calS'$ may depend on $W$ in some nontrivial way, and in model $\calM$, $\reg{W}$ is kept by $\calZ_1$ which is not available to $\calS$. The strategy here is to consider a special class of $\calZ_1$'s that always {\em authenticate} $\reg{W}$ before passing to someone else. Note that by hypothesis the construction of $\calS'$ is independent of the environment, and it must work against all possible environments. This would therefore force that $\calS'$ must be essentially the identity operator on $\reg{W}$, because authentication ensures that any tampering on $\reg{W}$ will be otherwise detected by $\calZ_2$. This tells us that the operation of $\calS'$ is independent of $\reg{W}$, and thus can be simulated by $\calS$ who does not have $\reg{W}$.

To prove formally Lemma~\ref{lemma:eqsp2}, we first review what {quantum authentication} is as proposed in~\cite{BCGST02}.

\begin{definition}(Quantum Authentication Scheme)
A quantum authentication scheme (QAS) is a pair of polynomial time
algorithms $A$ and $B$ together with a set of classical keys
$\mathcal{K}$ such that
\begin{itemize}
    \item Let $M$ be a message system on $m$ qubits. $A$ takes $M$ and a key $k \in \mathcal{K}$ as
    input and outputs a transmitted system $T$ of $m+t$ qubits.
    \item $B$ takes as input the (possibly altered) transmitted
    system $T'$ and a classical key $k\in \mathcal{K}$ and outputs
    two systems: a $m$-qubit message state $M$, and a single qubit
    $V$ which indicates acceptance or rejection.
\end{itemize}
For any fixed key $k$, we denote the corresponding superoperators of $A$ and $B$ by $A_k$ and $B_k$.
\label{def:qas}
\end{definition}

Roughly speaking, a QAS  $({A,B})$ is secure if $B$ with high probability can recover faithfully Alice's message or detect tampering if there is any during transmission. The proposition below states an interesting property about a {secure} QAS which is crucial in proving Lemma~\ref{lemma:eqsp2}. Its proof and the formal definition of secure QAS will appear later in Appendix~\ref{ssec:pcqas}. In what follows, ${\bar A}$ denotes the operation that samples a random key $k\gets \mathcal{K}$ and applies $A_k$, and let $\bar B$ denote the corresponding decoding procedure $B_k$. We claim that, from the perspective of a computationally bounded distinguisher, if an operator $\calO$ appears not to touch part of an input, then $\calO$ essentially decomposed to tensor product of identity and another operator $\calO'$, which we can construct efficiently from $\calO$ and the authentication scheme.

\begin{prop}
Let $({A,B})$ be a secure QAS. Given two superoperators $\calE: \text{L}(X) \to \text{L}(X')$ and $\calO: \text{L}(W\otimes X) \to \text{L}(W \otimes X')$, define $\calO': \text{L}(X) \to \text{L}(X')$ by $\rho \mapsto Tr_{W'}\left[\calO \left(\I_{\lin{X}} \otimes {\bar A} (\rho \otimes |0\rangle \langle 0|_{W'})\right)\right]$.
If $\calE\otimes\I_{\lin{W}} \approx_{wqc} \calO$, then $\calE \approx_{qc} \calO'$.

\begin{figure}[h!]

\centering 
\ifnum\ijqi=1
{\includegraphics[width=3in]{oprime}} 
\else
{\includegraphics[width=3in]{Figures/oprime}} 
\fi
\caption{Illustration of $\calO$ and $\calO'$} \label{fig:oprime}
\end{figure}
\label{prop:cqas}
\end{prop}

\anote{I think we need an intuitive explanation here of what the
  proposition ``means''. I had trouble understanding the relevance of
  the result, since it's been a few years!}
\fnote{intuition added. not sure about necessity of authentication...}

\begin{proof}[Proof of Lemma~\ref{lemma:eqsp2}]
Assume $\Pi$ \cqsa~emulates $\Gamma$ in model $\calM'$, our goal is to show that $\Pi$ also \cqsa~emulates $\Gamma$ in $\calM$. Specifically, we need to argue that for any $\calA$ there is an $\calS$ such that $\mac_{\Pi, \calA} \approx_{qc} M_{\Gamma, \calS}$. The complete construction of $\calS$ is shown below (see also
Fig.~\ref{fig:ps}). %where P and Q stand for perfectly and quantum computationally indistinguishable).
\begin{figure}[h!]
\ifnum\ijqi=1
\centering {\includegraphics[width=4in]{passstate}} 

\else
\centering {\includegraphics[width=4in]{Figures/passstate}} 
\fi
\caption{Proof
of Lemma.~\ref{lemma:eqsp2}.} 

\label{fig:ps}
\end{figure}

\anote{The subscripts on the figure do not correspond to the subcripts
  in the proof. Should $\Pi'$ be $\Gamma$?}
\fnote{yep. fixed! thx}
  
Given an arbitrary real-world adversary $\calA$ acting in model
$\calM= \model{\a}{\bbar}{\c}$,
\begin{itemize}
%Let the input state generated by $\calZ_1$ be %$|\bfPsi\rangle \in S_\bfA\otimes S_\bfB\otimes W$.
\item Construct $\calA'$ acting in model $\calM' = \model{\a}{\bbar}{\cbar}$. Adversary $\calA'$
receives some state on registers $\reg{S}_{\calA}$ and $\reg{W}$. It hands
$\reg{S}_{\calA}$ to $\calA$ and outputs what $\calA$ outputs. Leave $\rspace{W}$
untouched. Then $\mac_{\Pi, \calA'} \equiv \mac_{\Pi, \calA} \otimes \I_{\lin{\rspace{W}}}$ by construction.

\item Since we assume $\Pi$ \qsa~emulates $\Gamma$ in $\calM'$, by definition, there is an $\calS'$ such that $M_{\Pi, \calA'} \wqc M_{\Gamma, \calS'}$.

\item Finally construct $\calS$ for
$\calA$. Simulator $\calS$ applies operator $\bar{A}$ that authenticates $|0\rangle_{\reg{W}'}$ with a randomly chosen key. %(i.e., as defined in Definition~\ref{def:qas}). 
Then it runs $\calS'$ on
$\reg{S}_{\calA}$ and $\reg{W}'$, outputs $\reg{O}_{\calA}$ and discards $\reg{W}'$.  From above we know that $\mac_{\Pi,\calA} \otimes \I_{\lin{\rspace{W}}} \wqc \mac_{\Gamma,\calS'}$, we can apply Proposition~\ref{prop:cqas} with $\calE : = \mac_{\Pi,\calA}$ and $\calO := \mac_{\Gamma, \calS'}$.  This gives us %$M_{\Gamma, \calS} \otimes \I_{\reg{W}} \approx_{wqc} \mac_{\Pi, \calS'}$. Combining the sequence of indistinguishability, we obtain that 
$\mac_{\Pi, \calA} \approx_{qc} \mac_{\Gamma,\calS}$. 
\end{itemize}
Thus if $\Pi$ emulates $\Gamma$ w.r.t. the definition in $\calM'$, it also emulates $\Gamma$ w.r.t. $\calM$.
\end{proof} % of lemma:eqsp2
Lemma~\ref{lemma:eqsp1} and Lemma~\ref{lemma:eqsp2} together show Prop.~\ref{prop:aux1_sp}: passing state is irrelevant if $\calZ_1$ takes quantum advice.

%====New proof =====%
\mypar{New proof of Lemma~\ref{lemma:eqsp2}}
Given an arbitrary real-world adversary $\calA$ acting in model
$\calM= \model{\a}{\bbar}{\c}$, we first construct $\calA'$ acting in model $\calM' = \model{\a}{\bbar}{\cbar}$. Let $\reg{W'}$ be the output space of the authentication scheme on $\reg{W}$. Define $\calA': = \I_{\lin{\reg{W}'}} \otimes \calA$. Then there exists an $\calS'$ such that $\mac_{\Pi, \calA'} \wqc \mac_{\Gamma, \calS'}$ by the hypothesis that $\Gamma$ \cqsa-emulates $\Pi$ in model $\calM'$. Finally we construct $\calS$ for
$\calA$ as follows: $\calS$ applies operator $\bar{A}$ that authenticates $|0\rangle_{\reg{R}}$ with a randomly chosen key, where $\reg{R}$ is an auxiliary system with the same dimension as $\reg{W}$. %(i.e., as defined in Definition~\ref{def:qas}). 
Then it runs $\calS'$ on
$\reg{S}_{\calA}$ and $\reg{R}'$, outputs $\reg{O}_{\calA}$ and discards $\reg{R}'$.  

Let $\kb{\psi} \in \density{\reg{W}\otimes\reg{S_\calA}\otimes\reg{S_B}}$ be an arbitrary input state. Let $\rho:= (\I_{\lin{\reg{W}}}\otimes \mac_{\Pi,\calA})(\kb{\psi})$ and  $\sigma:= (\I_{\lin{\reg{W}}}\otimes \mac_{\Gamma,\calS})(\kb{\psi})$. We show below that $\rho \wqc \sigma$, which will prove that [\textbf{well, not really...}] $\mac_{\Pi, \calA} \approx_{qc} \mac_{\Gamma,\calS}$. 

Define the following states: 

\begin{align*}
\hat \rho &: = (\dauth\otimes\I_{\lin{\reg{O_\calA}\otimes\reg{O_B}}})\mac_{\Pi,\calA'}(\auth\otimes\I_{\lin{\reg{S_\calA}\otimes\reg{S_B}}})(\kb{\psi})\\
\eta &:= (\dauth\otimes\I_{\lin{\reg{O_\calA}\otimes\reg{O_B}}}){\mac_{\Gamma,\calS'}}(\auth\otimes\I_{\lin{\reg{S_\calA}\otimes\reg{S_B}}})(\kb{\psi})\\
\gamma &:= Tr_\reg{R}\left( (\I \otimes\dauth)(\I_{\lin{\reg{W}}}\otimes \calS'\otimes \Gamma_B)(\I\otimes \auth)(\kb{\psi}\otimes\kb{0}_{\reg{R}}) \right)
\end{align*}
Note that $\gamma$ coincides with the state that in $\calS$, instead of discarding (i.e.tracing out) $\reg{R}'$, we apply $\dauth$ on $\reg{R}'$ and trace out $\reg{R}$ only. Clearly $Tr_\reg{V}(\gamma)=\sigma$.  A few simple observations are also in order: $\hat \rho = \kb{1}_\reg{V}\otimes \rho$ by completeness of QAS; $\eta \wqc \hat \rho$ by hypothesis $\mac_{\Pi, \calA'} \wqc \mac_{\Gamma, \calS'}$; and $\td(\eta,\gamma) \leq \negl(n)$ by the soundness of QAS [\textbf{This needs a few lines to argue}]. 
Therefore $\rho\wqc \sigma$ because for any $\calZ_2$: 
\begin{align*}
& \left| \Pr(Z_2(\rho) = 1) - \Pr(Z_2(\sigma) =1) \right| \\
 \leq & \left| \Pr(Z_2(Tr_\reg{V}(\hat \rho)) = 1) - \Pr(Z_2(Tr_\reg{V}(\eta)) = 1)\right| \\
+  & \left| \Pr(Z_2(Tr_\reg{V}(\eta)) = 1) - \Pr(Z_2(Tr_\reg{V}(\gamma)) = 1) \right|\\
+ & \left| \Pr(Z_2(Tr_\reg{V}(\gamma)) = 1) - \Pr(Z_2(\sigma) = 1) \right|\\
\leq & \negl(n)+\negl(n) + 0  = \negl(n)
\end{align*}

\mypar{Statistical Setting} Similar techniques allow us to get better understandings of statistical stand-alone security models as well. Recall Definition~\ref{def:qssa} (Quantum Statistically Stand-Alone Emulation): \emph{Let $\Pi$ and $\Gamma$ be two poly-time protocols. We say $\Pi$ \emph{quantum statistical stand-alone} ($\sqsa$) emulates $\Gamma$ if for any $\calA$ there exists an $\calS$ that runs in poly-time of that of $\calA$, such that $M_{\Pi, \calA} \approx_{\diamond} M_{\Gamma, \calS}$}. We denote this security model $\calM$. Another model, denote it $\calM'$, which requires that $M_{\Pi,\calA} \approx_{tr} M_{\Gamma, \calS}$, seems to enforce weaker security guarantee.  Nonetheless, these two models are actually equivalent. 

\begin{theorem} $\calM \equiv \calM'$. Namely, for any adversary
  $\calA$, there exists an $\calS$ such that $M_{\Pi,\calA}
  \approx_{tr} M_{\Gamma, \calS}$ if and only if there exists an $\calS'$ such that $M_{\Pi,\calA} \approx_\diamond M_{\Gamma, \calS'}$. 
\label{thm:tr=diam}
\end{theorem}

This can be shown by a similar proof of Proposition~\ref{prop:aux1_sp}. We will need a statistical analogue of Proposition~\ref{prop:cqas}, which we state below. We omit its proof because one can follow almost line by line the proof of Proposition~\ref{prop:cqas}, replacing computational indistinguishability with trace distance  in a few places (e.g. in Claim~\ref{claim:cqas3}). 

\begin{prop}
Let $({A,B})$ be a secure QAS. Let $\calE,\calO$ and $\calO'$ be defined as in Proposition~\ref{prop:cqas}. 
If $\calE\otimes\I_{\lin{W}} \approx_{tr} \calO$, then $\calE \approx_{\diamond} \calO'$.
\label{prop:sqas}
\end{prop}

Note that Proposition~\ref{prop:cqas} and Proposition~\ref{prop:sqas}
have a similar flavor to the equivalence between various notions of quantum channel fidelities~(\citet{BKN00}).

%\mypar{Models without Quantum Advice to $\calZ_1$} 
\subsubsection{Models without Quantum Advice to $\calZ_1$ and Other Variants} 
\label{sssec:woadv}
When $\calZ_1$
does not receive quantum advice, the relationship between the
different models is less clear. State passing still does not make a
difference (see Figure~\ref{fig:equiv2_body}), but we do not know whether the possibility of quantum advice to $\calZ_2$ 
changes the model. This question appears closely related to the
question of whether BQP/qpoly equals BQP/poly. 
Appendix~\ref{subsec:woadv} contains precise statements and further discussion.

% \begin{theorem}
% $\model{\abar}{\bbar}{\cbar} \equiv \model{\abar}{\bbar}{\c}  \text{ and } \model{\abar}{\b}{\c} \equiv \model{\abar}{\b}{\cbar} \ . $
% \label{thm:nz1ps_body}
% \end{theorem}
\begin{figure}[h!] \centering 
\ifnum\ijqi=1
\mbox{ \subfigure[$ \model{\abar}{\bbar}{\cbar} \equiv \model{\abar}{\bbar}{\c} $ ]
{\includegraphics[width=0.45\columnwidth]{equiv2a}}\label{fig:naux2_body}
\qquad \subfigure[$\model{\abar}{\b}{\c} \equiv \model{\abar}{\b}{\cbar}$]{\includegraphics[width=0.45\columnwidth]{equiv2b}
}\label{fig:aux2_body}}
\else
\mbox{ \subfigure[$ \model{\abar}{\bbar}{\cbar} \equiv \model{\abar}{\bbar}{\c} $ ]
{\includegraphics[width=0.45\columnwidth]{Figures/equiv2a}}\label{fig:naux2_body}
\qquad \subfigure[$\model{\abar}{\b}{\c} \equiv \model{\abar}{\b}{\cbar}$]{\includegraphics[width=0.45\columnwidth]{Figures/equiv2b}
}\label{fig:aux2_body}}
\fi
\caption{Known equivalences between models where $\calZ_1$ takes no quantum advice}\label{fig:equiv2_body}
\end{figure}

Finally, in Appendix~\ref{subsec:mc}, we
discuss another constraint previously studied in the literature, a
``Markov condition'' on the inputs to the protocol and the adversary's
advice.

\fi

%============================================%
\subsection{Quantum UC Model: An Overview}
\label{subsec:qucmodel}
%============================================%

So far, our security model falls into the \emph{stand-alone} setting, where protocols are assumed to be executed in isolation. However, in practice we often encounter a \emph{network} setting, where many protocols are running concurrently. A protocol proven secure according to a stand-alone security definition ensures nothing if we run it in a network environment. In view of this issue, Canetti~\cite{Can01} proposed the (classical) Universally Composable (UC) security model. It differs from the stand-alone definition of security in that the environment is allowed to be
\emph{interactive}: during the execution of the protocol, the
environment may provide inputs and receive the outputs of the honest
players, and exchange arbitrary messages with the adversary. In
contrast, the environment in the stand-alone model runs only at the
end of the protocol execution (and, implicitly, before the protocol
starts, to prepare the inputs to all parties). UC-secure protocols
enjoy a property called {\em general} (or \emph{universal})
  \emph{composition}\footnote{There is a distinction between UC security (a
  definition that may be satisfied by a specific protocol and ideal functionality) and
  universal composition (a property of the class of protocols that satisfy a security definition). Not all definitions that admit universal
  composition theorems are  equivalent to UC security. See
  \cite{HU06,Lindell09} for discussion.}: loosely
speaking, the protocol remains secure even if it is run concurrently
with an unbounded number of other arbitrary protocols (whereas
proofs of security in the stand-alone model only guarantee security
when only a single protocol at a time is running).

Earlier work on defining UC security and proving universal
composition in the quantum setting appears
in~\cite{BOM04,Unr04}. We will adapt the somewhat simpler
formalism of Unruh~\cite{Unr10}.

Modulo a small change in Unruh's model (quantum advice, discussed below), our stand-alone model is exactly the restriction of Unruh's model to a \emph{non-interactive} environment, that is one which is
idle from the start to the finish of the protocol. 
The only apparent difference is that in the UC model, the environment
runs for some time before the protocol starts to prepare inputs, while
in Section~\ref{subsec:qsamodel} we simply quantify over all joint
states $\sigma$ of the honest players' and adversary's inputs. This
difference is only cosmetic, though: the state $\sigma$ can be taken
to be the joint state of the outputs and internal memory of the
environment at the time the protocol begins. %A more ``quantum-flavor'' issue is that the environment, though idle during the execution of the protocol, may keep a state that is entangled with the inputs to the honest party and the adversary, which may increase its distinguishing ability in the end. However, as we show in Section \ref{subsec:variants}, no actual difference occurs.

We make one change to Unruh's model in order to be consistent with our
earlier definitions and the work of Watrous on zero-knowledge
\cite{Wat09}: we allow the environment to take quantum advice, rather
than only classical advice.  In the language of \cite[p.\ 11]{Unr10},
we change the initialization phase of a network execution to create a
state $\rho \in \mathcal{P}(\mathcal{H}_{\mathrm{ \mathbf{ N}}})$
which equals the classical string $|(\varepsilon, \text{\tt
  environment}, \varepsilon)\rangle$ in $\mathcal{H}^{class}$ (instead
of $|(\varepsilon, \text{\tt environment}, z)\rangle$), and an
arbitrary state $\sigma$ in $\mathcal{H}^{quant}$ (instead of
$|\varepsilon\rangle$). Here $\varepsilon$ denotes the empty
string. Moreover, we change the definition of \emph{indistinguishable
  networks} \cite[p.\ 12]{Unr10} to quantify over all states $\sigma$
instead of all classical strings $z$. {This change is not significant
  for statistical security, since an unbounded adversary may
  reconstruct a quantum advice state from a (exponentially long)
  classical description. However, it may be significant for
  polynomial-time adversaries: it is not known how much quantum advice
  affects the power of, say BQP, relative to classical advice.} For
completeness, we state this modified definition of quantum UC security
below.

\begin{definition}[Computationally Quantum-UC Emulation]
Let $\Pi$ and $\Gamma$ be two-party protocols. We say $\Pi$ {\em computationally quantum-UC (\cquc) emulates} $\Gamma$, if for any poly-time QIM $\calA$, there is a poly-time QIM $\calS$ such that $M_{\Pi, \calA} \qcii M_{\Gamma, \calS}$ (as per Def.~\ref{def:qci}). 
\label{def:quc} 
\end{definition}
Here $\mac_{\Pi,\calA}$ (and $\mac_{\Gamma,\calA}$ likewise) denotes the composed system of $\Pi$ and $\calA$, which can be viewed as a QIM. Its network register consists of part of the adversary's network register, and is used for external communication with another party (e.g., an environment). Alternatively, define
$\mathbf{EXEC}_{\Pi,
\calA, \calZ} :=\{\langle \calZ (\sigma_n), M_{\Pi,\calA}\rangle\}_{n\in\bbN,\sigma_n \in \density{\mathcal{H}_n}}$ and 
 $\mathbf{EXEC}_{\calF, \calS, \calZ} := \{ \langle \calZ (\sigma_n), M_{\Gamma,\calS}\rangle\}_{n\in\bbN, \sigma_n \in \density{\mathcal{H}_n}} $.  We can rephrase the condition as ``\emph{for any poly-time QIM $\calA$, there is a poly-time QIM $\calS$, such that for any poly-time QIM $\calZ$,  $\mathbf{EXEC}_{\Pi, \calA, \calZ} \approx
\mathbf{EXEC}_{\Gamma, \calS, \calZ}$.}'' 

If we allow $\calA$ and $\calZ$ to be unbounded machines, i.e., we require that $\mac_{\Pi,\calA}\approx_{qsi} \mac_{\Gamma,\calS}$, then we get the notion of statistically quantum-UC (\squc) emulation. As suggested in~\cite{Gut13}, we can also use the $\|\cdot\|_{\diamond r}$ norm on strategies to define it. Namely, we require that for any $\calA$ there exists $\calS$ such that $\|\mac_{\Pi,\calA} - \mac_{\Gamma,\calS}\|_{\diamond r} \leq \text{negl}(n)$.  

%Later on, we will use Q-CUC (resp. C-CUC) and Q-SUC (resp. C-SUC) to refer to objects in \emph{quantum}- (resp. \emph{classical}-) computationally or statistically UC model. 
 
\mypar{General (Concurrent) Composition} The most striking feature of
UC model is that it admits a very general form of composition,
concurrent composition\footnote{People often refer to this type of
  composition as UC composition, presumably because security in the UC
  model implies generally concurrent composition. This should not
  cause any further confusion.}. Specifically, consider a protocol
$\Pi$ that makes subroutine calls to a protocol $\Gamma$. In contrast
to the stand-alone setting, we now allow multiple instances of $\Gamma$ running concurrently.
(For a
formal description of general composition operation, see \citet{Can01}.)
As before, we write
$\Pi^{\Gamma/{\Gamma'}}$ to denote the protocol obtained by $\Pi$ by
substituting $\Gamma'$ for subroutine calls to $\Gamma$. % However, as
% opposed to the composition operation in the stand-alone setting, there
% can be multiple instances of $\Gamma$ running concurrently.

%\fnote{add in UC composition operation: remark about the abuse of term: UC- model and the general composition operation}

Our modifications of Unruh's definition do not affect the validity of the \emph{universal composition} theorem:

\begin{theorem}[Quantum UC Composition Theorem (\citet{Unr10})] 
Let $\Pi, \Gamma$ and $\Gamma'$ be poly-time
protocols. Assume that $\Gamma$ quantum-UC emulates $\Gamma'$. Then
$\Pi^{\Gamma/{\Gamma'}}$ quantum-UC emulates $\Pi$. 
\label{thm:quccomp}
\end{theorem}

\def\dumA{\calA_{dummy}} There is another useful property that
simplifies the proof of UC emulation. In both classical and quantum UC
models, it suffices to consider a special adversary, which is called
the \emph{dummy} adversary $\dumA$. The dummy adversary $\dumA$ just
forwards messages between a protocol and an environment and leaves any
further processing to the environment. % See~\citet{Can01,Unr10} for more detailed and formal discussions about
% the dummy adversary. 
Here we only restate the completeness of dummy
adversary in the quantum setting:

\begin{theorem}[Completeness of the dummy adversary ({\citet[Lemma 12]{Unr10}})]
Assume that $\Pi$ quantum-UC emulates $\Gamma$ with respect to the dummy adversary (i.e., instead of quantifying over all adversaries, we fix $\calA:=\dumA$). Then $\Pi$ quantum-UC emulates $\Gamma$. This holds both for computational and statistical settings. 
\label{thm:qucdum}
\end{theorem}
\fi

%===============================%
\section{Classical Protocols with Quantum Security}
\label{sec:q2pc}
%===============================%

This section studies {what classical protocols remain secure against \emph{quantum} attacks} in the computational setting. Let $\calF$ be a classical two-party poly-time functionality. For technical reasons, $\calF$ needs to be \emph{well-formed}. See~\cite{Can01,CLOS02} for a formal definition and discussions. Throughout this paper, we only consider well-formed functionalities as well. Classically, there are two important families of secure protocols: 

\begin{itemize}
	\item {Stand-alone secure computation}~\cite{GMW87}: Assuming the existence of  enhanced trapdoor permutations, there exists poly-time protocols that computationally stand-alone emulates $\calF$.
	\item Universal-composable secure computation~\cite{CLOS02}: Assuming the existence of enhanced trapdoor permutations, there exists protocols in the $\fzk$-hybrid model that computationally UC emulates $\calF$.
\end{itemize}
 
Our main result shows that these general feasibility results largely remain unchanged against quantum attacks:

\begin{theorem*}[Informal]  For any classical two-party
functionality $\calF$, there exists a classical protocol $\pi$ that
quantum computationally stand-alone emulates $\calF$, under suitable quantum-resistant computational assumptions.
\label{thm:qgmwinformal}
\end{theorem*}

The proof of the theorem can be broken into two parts. First we show
a quantum analogue of~\cite{CLOS02} in Section~\ref{subsec:qclos}. Namely, there exist functionalities, such as $\fzk$, that are as powerful as to realizing any other functionalities based on them, even with respect to computationally quantum-UC security. To achieve this, we develop a framework called \emph{simple hybrid arguments} in Sect~\ref{sssec:sha} to capture a large family of classical security analyses that go through against quantum adversaries. As a result, it amounts to design a (stand-alone) secure protocol for $\fzk$, which is the content of Section~\ref{subsec:zkaok}. We stress that security of existing protocols for $\fzk$ relies on a sophisticated rewinding argument, and it is not clear if the arguments are still valid against quantum adversaries. Hence we need new ideas to get around this difficulty.

%Before we proceed to develop our main theorem, we review the quantum Universal-Composable security model. 

\subsection{Basing Quantum UC Secure Computation on $\fzk$}
\label{subsec:qclos}

We show here that $\fzk$ is sufficient for UC secure computation of any two-party functionality against any computational bounded quantum adversaries. That is, for any well-formed functionalities $\calF$,  there exists an $\fzk$-hybrid protocol that \cquc~emulates $\calF$. We stress that these protocols are all \emph{classical}, which can be implemented efficiently with classical communication and computation devices. 

\begin{theorem} 
Let $\calF$ be a two-party functionality. Under Assumptions~\ref{asn:prg} and \ref{asn:pkc}, there exists a  classical $\fzk$-hybrid protocol that \cquc~emulates $\calF$ in the presence of polynomial-time \emph{malicious quantum} adversaries with \emph{static} corruption.
\label{thm:qclos}
\end{theorem}

%The computational assumptions are stated below. 
\begin{assumption}
 There exists a classical pseudorandom generator secure against any poly-time quantum
    distinguisher.
\label{asn:prg}
\end{assumption}

Based on this assumption and the construction of~\cite{Nao91}, we
can obtain a statistically binding and quantum computationally
hiding commitment scheme $\pi_{\tt com}$. All commitment
schemes we use afterwards refer to this one. This assumption also
suffices for Watrous's ZK proof system for any NP-language against
quantum attacks.
\begin{assumption}
      There exists a {\em dense} classical public-key  crypto-system that is
      IND-CPA (chosen-plaintext attack) secure against quantum
      distinguishers. 
\label{asn:pkc}
\end{assumption}

A public-key crypto-system is dense if a valid public key is
indistinguishable in quantum poly-time from a uniformly random
string of the same length. Although it is likely that standard reductions would show that
\asnref{pkc} implies \asnref{prg}, we chose to keep the assumptions
separate because the instantiation one would normally use of the
pseudorandom generator would not be related to the public-key system
(instead, it would typically be based on a symmetric-key block or
stream cipher). Both assumptions hold, for instance, assuming the
hardness of {\em leaning with errors} (LWE) problem~\cite{Reg09}.

%==============================% 
\subsubsection{Simple Hybrid Argument.} 
\label{sssec:sha}
%==============================%
%
Our analysis is based on a new abstraction called a \emph{simple hybrid argument} (SHA). It captures a family of classical security arguments in the UC model which remains valid in the
quantum setting (as long as the underlying primitives are secure against quantum adversaries).

\begin{definition}[Simply related machines]
We say two QIMs $M_a$ and $M_b$ are \emph{$(t,\veps)$-simply related} if there is a time-$t$ QTM $M$ and a pair of classical
distributions $(D_a, D_b)$ such that
  \begin{enumerate}
  \item $M(D_a) \equiv M_a$ (for two QIMs $N_1$ and $N_2$, we say $N_1
    \equiv N_2$ if the two machines behave identically on all inputs,
   that is, if they can be described by the same circuits),
  \item $M(D_b) \equiv M_b$, and
  \item $D_a \approx_{qc}^{2t,\veps} D_b$.
  \end{enumerate}
\end{definition}

\begin{example} 
Figure~\ref{fig:srm} illustrates a pair of simply
related machines. 
\label{ex:srm}
\end{example}

\vspace{-4mm}
\begin{figure}[h!]
\centering 

\ifnum\ijqi=1
\mbox{ \subfigure
{\includegraphics[width=2in]{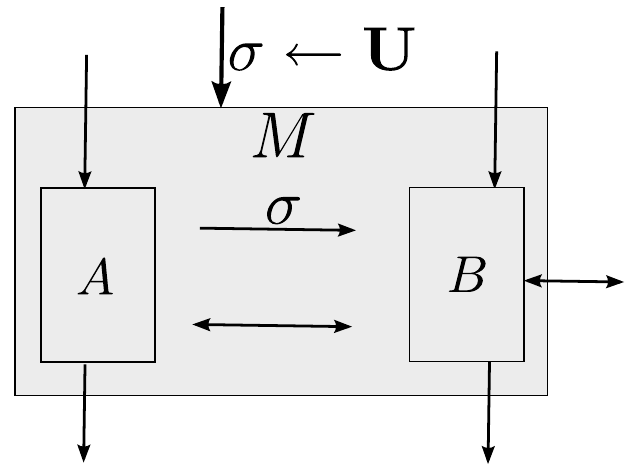}}\label{fig:srm1} \qquad
\subfigure
{\includegraphics[width=2in]{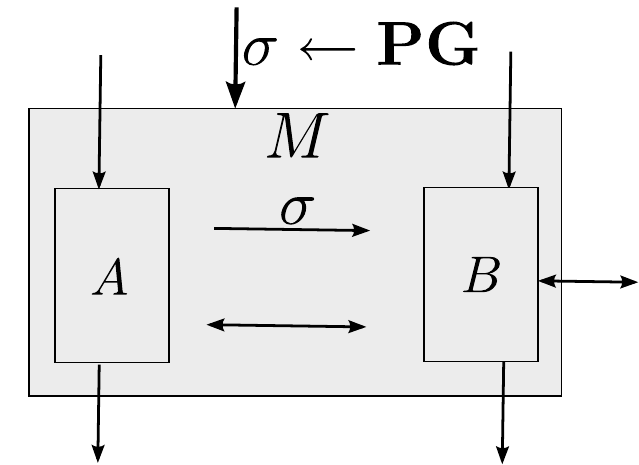}
}}\label{fig:srm2}
\else
\mbox{ \subfigure
{\includegraphics[width=2in]{Figures/sim_exp1}}\label{fig:srm1} \qquad
\subfigure
{\includegraphics[width=2in]{Figures/sim_exp2}
}}\label{fig:srm2}
\fi
\caption{Two simply related machines: $M_a$ is machine $M$ on input $\sigma$
chosen uniformly at random; $M_b$ is machine $M$ on input a pseudorandom
string $\pgen(r)$.}
\label{fig:srm} 
\end{figure}

\begin{lemma}
If two machines $M_a$ and $M_b$ are $(t, \veps)$-simply related, then
$M_a \qcii^{t,\veps} M_b$, i.e., they are $(t,
\veps)$-interactively indistinguishable (as per
Definition~\ref{def:qci}).
\end{lemma}

\begin{proof}
By definition, $M_a \equiv M(D_a)$ and $M_b \equiv M(D_b)$. If there
is a $\calZ$ with quantum advice $\sigma$ that distinguishes $M_a$ and
$M_b$ with advantage $\veps'
> \veps$ in time $t$, we can construct a time-$2t$ distinguisher
$\mathcal{D}$ for $D_a$ and $D_b$ with advantage $\veps'$ as well.
This contradicts $D_a\approx_{qc}^{2t,\veps}D_b$.  Distinguisher $\mathcal{D}$ works by
taking an input sample $d$ from either $D_a$ or $D_b$, 
simulates $\langle\calZ(\sigma), M(d)\rangle$, and outputs whatever $\calZ$
outputs. Obviously, $\mathcal{D}$ runs in time at most $2t$ and
distinguishes $D_a$ and $D_b$ with the same advantage that $\calZ$
distinguishes $M_a$ and $M_b$. Thus we conclude $|\Pr(\langle \calZ(\sigma), M_a \rangle =1) - \Pr(\langle \calZ(\sigma), M_b\rangle= 1)|
\leq \veps$ for any time-$t$ environment $\calZ$.
\end{proof}

\begin{definition} [Simple hybrid argument]
  Two machines $M_0$ and $M_\ell$ are related by a
  \emph{$(t,\veps)$-simple hybrid argument of length $\ell$} if there is a
  sequence of intermediate machines $M_1,M_2,...,M_{\ell-1}$ such
  that each adjacent pair $M_{i-1},M_{i}$ of machines,
  $i=1,\dots,\ell$, is $(t,\frac{\veps}{\ell})$-simply related.
  \label{def:sha}
\end{definition}

\begin{lemma}
  For any $t,\veps$ and $\ell$, if two machines are related by a $(t,\veps)$-simple hybrid
  argument of length $\ell$, then the machines are $(t,\veps)$-interactively indistinguishable.
  \label{lemma:sha}
\end{lemma}

\begin{proof}
  This is by a standard hybrid argument. Suppose, for contradiction, there
  exists a time-$t$ machine $\calZ$ with advice $\sigma$ such that
$$|\Pr(\langle\calZ(\sigma), M_0 \rangle =1) - \Pr(\langle \calZ(\sigma), M_\ell\rangle = 1)|
\geq \veps \ . $$ Then by triangle inequality we can infer that
there must exist some $i$ such that $$|\Pr(\langle \calZ(\sigma), M_i\rangle
=1) - \Pr(\langle\calZ(\sigma), M_{i+1} \rangle =
    1)| > {\veps}/{\ell} \ . $$ However, by assumption
$M_i$ and $M_{i+1}$ are ($t,\frac{\veps}{\ell})$-simply related and
in particular no time-$t$ machines can distinguish them with
advantage greater than ${\veps}/{\ell}$.
\end{proof}

%=================================%
\subsubsection{Lifting CLOS to Quantum UC Security.} 
\label{sssec:liftclos}
%=================================%

Now we apply our \emph{simple hybrid argument} framework to analyze the protocol in CLOS. We first review the structure of the construction of CLOS in the static setting:

\begin{enumerate}
	\item Let $\calF$ be a two-party functionality. Design a protocol $\pi$ that computationally (classical) UC (\ccuc) emulates $\calF$ against semi-honest adversaries. The protocol uses a semi-honest oblivious transfer (ShOT) protocol, which can be constructed assuming existence of enhanced trapdoor permutations. %(eTDPs). 
	\item Let $\fcp$ be the ``commit-and-prove'' functionality
          of~\cite[Figure 8]{CLOS02}. A protocol is constructed in
          $\fzk$-hybrid model that \ccuc~emulates $\fcp$, assuming
          existence of a statistically binding and computationally
          hiding commitment scheme. Such a commitment scheme in turn can be constructed from a pseudorandom generator~\cite{Nao91}. 
	\item In $\fcp$-hybrid model, a generic complier ${\tt COMP}$ is designed. Let $\pi' = {\tt COMP}(\pi)$ be the $\fcp$-hybrid protocol after compilation. It is shown in~\cite[Proposition 8.1]{CLOS02} that: \emph{for every classical adversary $\calA'$, there exists a classical adversary $\calA$ with running time polynomial in that of $\calA'$  such that $\mathbf{EXEC}_{\pi', \calA', \calZ} \equiv \mathbf{EXEC}_{\pi, \calA, \calZ}$}. That is, the interaction of $\calA'$ with honest players running $ \pi'$ is \emph{identical} to the interaction of $\calA$ with $\pi$ in the semi-honest model, i.e., $\mac_{\pi',\calA'} \equiv \mac_{\pi,\calA}$.
\end{enumerate}

It then follows that, by the UC composition theorem, $\pi'$	\ccuc~emulates $\calF$ in the $\fzk$-hybrid model. 

We then show how to make the construction secure against quantum adversaries using proper quantum-resistant assumptions. The key observation is that the security proofs of the semi-honest protocol and of the $\fcp$ protocol in the $\fzk$-hybrid model fall into our simple hybrid argument framework. Thus once we augment the computational assumptions to be quantum-resistant, they immediately become secure against quantum adversaries. This is stated more precisely below. 
 
\begin{obs}[CLOS proof structure]
%Except for the proof of security of protocol compilation from semi-honest to malicious adversaries, all the security proofs for
In CLOS, the security proofs for the semi-honest protocol and the protocol for $\fcp$ in $\fzk$-hybrid model against static adversaries consist of %either
%\begin{inparaenum}[\upshape(\itshape a\upshape)]
%\item simple hybrid arguments with $t=poly(n)$ and $\veps=\text{negl}(n)$, or
%\item applications of the UC composition theorem.
%\end{inparaenum}
simple hybrid arguments with $t=poly(n)$ and $\veps=\text{negl}(n)$. 

Moreover, the underlying indistinguishable distributions in the CLOS arguments consist of either %\fnote{NOT Fixed! this needs updating!}
\begin{inparaenum}[\upshape(\itshape i\upshape)]
\item switching between a real public key and a uniformly random string,
\item changing the plaintext of an encryption, or
\item changing the message in the commit phase of a commitment protocol. 
\end{inparaenum}

\label{obs:clos}
\end{obs}

From this observation, we get the corollary below.

%\fnote{NOT FIXED! Remark that in CLOS, the hybrid-machines can be made classical}. 

\begin{cor}[CLOS---simple hybrids]
\label{cor:CLOS-simple}

  \begin{enumerate}
  \item In the $\fzk$-hybrid model and under Assumption~\ref{asn:prg}, there is a
    non-trivial protocol that
    UC-emulates
    $\fcp$ in the presence of polynomial-time {\em malicious} static \emph{quantum} adversaries.

  \item Let $\calF$ be a well-formed two-party functionality. In the plain
    model, there is a protocol that UC-emulates $\calF$ in the presence of
     polynomial-time {\em semi-honest} static  \emph{quantum} adversaries under Assumption~\ref{asn:pkc}.
  \end{enumerate}
\end{cor}

\begin{proof}
Observation~\ref{obs:clos} tells us there are two types of proofs in
CLOS, so we only have to show both can be augmented to hold against
quantum adversaries. On the one hand, simple hybrid arguments in
CLOS still hold if we make assumptions~\ref{asn:prg}
and~\ref{asn:pkc}, because the underlying distributions in these
hybrid experiments will remain indistinguishable against quantum
distinguishers. On the other hand, we know quantum UC composition
also holds by Theorem~\ref{thm:quccomp}.

More specifically,  for the $\fcp$ protocol in $\fzk$-hybrid model,
the simply hybrid machines in its proof are related by switching the
messages being committed. Hence $\fcp$ protocol remains secure against
malicious static quantum adversaries under
Assumption~\ref{asn:prg}. In the semi-honest setting, an OT protocol
can be constructed from a dense crypto-system
(Assumption~\ref{asn:pkc}), see \citet{Goldreich-vol2}. %\fnote{Reference for OT from dense-crypto?}
Its proof consists of simply related machines that are related by either switching between a valid public key and a random string (when sender is corrupted) or switching the plaintext of an encryption (when receiver is corrupted). Therefore, this protocol \cquc~emulates $\fot$ against semi-honest quantum adversaries. Next in $\fot$-hybrid model, the construction for an arbitrary $\calF$ is unconditionally secure, which, by Unruh's lifting theorem, remains quantum-UC secure. Hence quantum UC composition theorem gives that there is a classical protocol that \cquc~ emulates $\calF$ in the presence of  {semi-honest} static {quantum} adversaries.
\end{proof}

Combining the previous arguments we can prove
Theorem~\ref{thm:qclos}.

\begin{proof}[Proof of Theorem~\ref{thm:qclos}]
  Fix a well-formed functionality
  $\calF$ and let $\pi$ be the protocol for $\calF$ in the semi-honest model guaranteed by the second part of Corollary~\ref{cor:CLOS-simple}. Now consider $\pi'={\tt COMP}(\pi)$. We want to show that it \cquc~emulates $\calF$. Theorem~\ref{thm:qucdum} tells us that it suffices to consider the classical dummy adversary $\dumA$. By \cite[Proposition 8.1]{CLOS02}, the interaction of the dummy adversary $\dumA$ with $\pi'$ (in the $\fcp$ hybrid model) is identical to the interaction of the adversary $\calA$ with $\pi$ (in the semi-honest model). By the security of $\pi$ in the semi-honest model, there exists an ideal-world adversary $\calS$ such that $\mac_{\calF,\calS} \qcii\mac_{\pi,\calA} \equiv \mac_{\pi',\dumA}$. Thus, $\pi'$ securely emulates $\calF$ in the $\fcp$-hybrid model against malicious adversaries. By the quantum UC composition theorem, we can compose $\pi'$ with the protocol for $\fcp$ to get a protocol secure against malicious quantum adversaries in the $\fzk$-hybrid model. 
\end{proof}

%----------------------------------------------------------%
\subsection{Realizing $\fzk$ with Stand-alone Security}
\label{subsec:zkaok}
%-----------------------------------------------------------%

In this section, we construct a protocol $\Pi_{ZK}$ that quantum stand-alone emulates $\fzk$. In the stand-alone model, $\fzk$ is more commonly referred to as \emph{zero-knowledge argument of knowledge}. 
%Our construction needs an encryption scheme that has one extra property than the one in Assumption~\ref{asn:pkc}.

%\begin{assumption}
%There exists, as in \asnref{pkc}, a {dense} classical public-key crypto-system that is IND-CPA secure against any quantum distinguisher. In addition, encryptions of two messages under a uniformly random string are statistically indistinguishable.
%\label{asn:spkc} 
%\end{assumption}

%Note that the {\em dense} property already implies encryptions under a random string are quantum computationally indistinguishable. Assumption~\ref{asn:spkc} strengthens this requirement to be statistically indistinguishable. 

%This allows ``cheating'' in the sense that if a ciphertext is generated under a uniformly random string, we can then claim it to be an encryption of an arbitrary message. 

%This type of encryption scheme is sometimes called Meaningful/Meaningless encryption (e.g., see \cite{KN08}). Again, the LWE assumption implies Assumption~\ref{asn:spkc}. 

We will use a dense encryption scheme $\mathcal{E} = ${ (\gen, \enc,
\dec)} as in Assumption~\ref{asn:pkc}%\ref{asn:spkc}
. Note that $\enc$ is a randomized algorithm and we denote by $\enc_{pk}(m,r)$ the encryption of a message $m$ under a public key $pk$ using randomness $r$, But unless when needed, we usually omit the randomness in the notation and write $\enc_{pk}(m)$. We will also need a result of Watrous's~\cite{Wat09}, where he showed that there exist classical zero-knowledge proofs for any \class{NP} language that are secure against any poly-time quantum verifiers. For completeness we give his definition (adapted to our terminology) of quantum computational zero-knowledge proof~\cite[Definition 7]{Wat09}.   

\begin{definition} 
An interactive proof system $(P,V)$ for a language $L$ is \emph{quantum computational zero-knowledge} if, for every poly-time QIM $V'$, there exists a poly-time QIM $\bfS_{V'}$ that satisfies the following requirements.

\begin{enumerate}
	\item The verifier $V'$ and simulator $\bfS_{V'}$ agree on the polynomially bounded functions $q$ and $r$ that specify the number of auxiliary input qubits and output qubits, respectively.
 	\item Let $M_{\langle P,V' \rangle (x)}$ be the machine describing that interaction between $V'$ and $P$ on input $x$, and let $M_{\bfS_{V'}(x)}$ be the simulator's machine on input $x$.  Then the ensembles $\{M_{\langle P, V'\rangle(x)}: x \in L\}$ and $\{M_{\bfS^{V'}(x)}: x \in L\}$ are quantum computationally indistinguishable as per Definition~\ref{def:qcm}.
\end{enumerate}
\label{def:qczk}
\end{definition}

Now that we have all building blocks ready, our construction of a classical ZKAoK protocol is as follows.

\begin{pffs}{ZKAoK Protocol
          $\Pi_{\tt ZK}$}{}{}

\noindent{\bf Phase 1}
        \begin{enumerate}
        \item $\bfV$ chooses $a \gets \{0,1\}^{n}$ at random, and sends
          $\bfP$ a commitment of $a$: $c = \bfcomm(a)$.

        \item $\bfP$ sends $b \gets \{0,1\}^{n}$ to $\bfV$.

        \item $\bfV$ sends $\bfP$ string $a$.

        \item $\bfV$ proves to $\bfP$ that $c$ is indeed a commitment of
          $a$ using Watrous's ZK protocol.

        \item $\bfP$ and $\bfV$ set $pk = a\oplus b$ and interpret it as a
          public key.
        \end{enumerate}\vspace{-2mm}
%        \hrule\vspace{1mm} 
\noindent{\bf Phase 2}\vspace{-2mm}
        \begin{enumerate}
        \item $\bfP$, holding an instance $x$ and a witness $w$,
          encrypts $w$ under $pk$. Let $e = \text{\bf
            Enc}_{pk}(w)$. $\bfP$ sends $(x,e)$ to $\bfV$.

        \item $\bfP$ proves to $\bfV$ that $e$ encodes a witness of $x$
          using Watrous's ZK protocol. $\bfV$ accepts if it accepts
          in this ZK protocol.  Otherwise it rejects and halts.

\end{enumerate}

\end{pffs}

\begin{theorem}
\label{thm:zkaok} {Protocol $\Pi_{ZK}$} quantum stand-alone-emulates
$\fzk$.
\end{theorem}

%\begin{proof}[Proof Sketch]
The full proof appears in Sect.~\ref{sssec:zkaokproof}. We provide a brief and intuitive justification here. Roughly speaking, Phase 1 constitutes what may be called a ``semi-simulatable'' coin-flipping protocol. Specifically we can simulate a corrupted Prover. This implies that a simulator $\calS$ can  ``cheat'' in Phase 1 and force the outcome to be a public key $pk$ of which he knows a corresponding secret key $sk$, so that $\calS$ can decrypt $e$ to recover $w$ in the end. This allows us to show argument of knowledge (in our stand-alone model). On the other side, although generally we cannot simulate a corrupted verifier in Phase 1, we can guarantee that the outcome $pk$ is uniformly random if the verifier behaves honestly. This is good enough to show zero-knowledge, because we can later encrypt an all-zero string and use the simulator for the ZK protocol in Phase 2 to produce a fake proof. In reality, a corrupted verifier may bias the coin-flipping outcome by aborting dependent on Prover's message $b$ for example. This technical subtlety, nonetheless, is not hard to deal with. Intuitively the verifier only sees ``less'' information about the witness if he/she decides to abort in Phase 1.

\subsubsection{Proof of Theorem~\ref{thm:zkaok}: Quantum Stand-alone Secure ZKAoK}
\label{sssec:zkaokproof}
%---------------------------------------------------------------------------%

For the sake of clarity, we propose a non-interactive notion of \emph{simple hybrid argument}, analogous to Def.~\ref{def:sha}, which formalizes a common structure in stand-alone security proofs.

\begin{definition}[Simply related non-interactive machines]
We say two QTMs $M_a$ and $M_b$ are \emph{$(t,\veps)$-simply related} if there is a time-$t$ QTM $M$ and a pair of QTMs $(N_a, N_b)$ such that
  \begin{enumerate}
  \item $M^{N_a} \equiv M_a$ (for two QTMs $N_1$ and $N_2$, we say $N_1  \equiv N_2$ if they can be described by the same circuits),
  \item $M^{N_b} \equiv M_b$, and
  \item $N_a \qc^{2t,\veps} N_b$.
  \end{enumerate}
\end{definition}

\begin{remark}\begin{inparaenum}[\upshape(i\upshape)]
\item $M^{N}$ is the machine that gives $M$ oracle access to $N$.
\item As a typical example of a pair of indistinguishable QTMs, consider $N_a$ being a QTM describing a ZK
protocol with a (dishonest) verifier, and $N_b$ being a simulator's
machine. Then by definition of a valid simulator, we have $N_a
\qc N_b$.
\item Machines $(N_a, N_b)$ in the definition also capture pair of
indistinguishable classical distributions that are efficiently samplable. Namely, we can let $N_a$ and $N_b$ be algorithms that sample from distributions $D_a$ and  $D_b$
respectively. 
%\fnote{Not really! model a distribution as a QTM. one possible issue is ``efficient sampling''}
\end{inparaenum}
\end{remark}
\begin{definition} [Simple hybrid argument (non-interactive version)]
  Two machines $M_0$ and $M_\ell$ are related by a
  \emph{$(t,\veps)$-simple hybrid argument of length $\ell$} if there is a
  sequence of intermediate machines $M_1,M_2,...,M_{\ell-1}$ such
  that each adjacent pair $M_{i-1},M_{i}$ of machines,
  $i=1,\dots,\ell$, is $(t,\frac{\veps}{\ell})$-simply related.
\end{definition}

\begin{lemma}
  For any $t,\veps$ and $\ell$, if two machines are related by a $(t,\veps)$-simple hybrid argument of length $\ell$, then the machines are $(t,\veps)$-indistinguishable.
  \label{lemma:nsha}
\end{lemma}

\begin{proof}
Suppose for contradiction, there exists a time-$t$ QTM $\calZ$ with
advice $\sigma$ such that $| \Pr[ \calZ((M_0\otimes
\I_{\lin{\reg{R}}})\sigma_n) = 1 ] - \Pr [ \calZ((M_\ell\otimes
\I_{\lin{\reg{R}}})\sigma_n) = 1]| > \veps$. Then by triangle
inequality we can infer that there must exist some $i$ s.t. $| \Pr[
\calZ((M_i\otimes \mathbb{1}_{\lin{\reg{R}}})\sigma_n) = 1 ] - \Pr [
\calZ((M_{i+1}\otimes \mathbb{1}_{\lin{\reg{R}}})\sigma_n) = 1]|  >
{\veps}/{\ell}$. However, by assumption $M_i$ and $M_{i+1}$ are
($t,{\veps}/{\ell})$-simply related and in particular no time-$t$
QTMs can distinguish them with advantage greater than
${\veps}/{\ell}$.
\end{proof}
\begin{remark} Actually, the proof of our modular composition can be seen as a simple hybrid argument. Specifically in step 3, $\mac_{\Pi',\calA}$ and  $\mac_{\Pi,\calS}$ are simply related by  $\mac_{\Gamma',\calA_{\Gamma'}}$ and $\mac_{\Gamma,\calA_{\Gamma}}$. 
\label{rmk:modsha}
\end{remark}

We now prove Theorem~\ref{thm:zkaok}
following the (non-interactive) simply hybrid argument framework.

\begin{proof}[Proof of Theorem~\ref{thm:zkaok}] We denote the two ZK proof systems in Phase 1 \& 2 by ZK$_1$ and ZK$_2$ respectively. The two NP languages, formalized below, are denoted by $L_1$ and  $L_2$ respectively.
{\small
\begin{eqnarray*}
    L_1 &=& \{(c, a): \exists r \in \{0,1\}^* \textrm{ s.t. } \bfcomm(a, r)=c\}\\
    L_2 &=& \{(pk, x, e): \exists w,r \in \{0,1\}^*, \textrm{ s.t. } \text{\bf Enc}_{pk}(w, r) = e
    \wedge   (x,w) \in R_L\}
\end{eqnarray*}
}
The simulators of ZK$_1$ and ZK$_2$ are denoted by $\mathbf{S}_1$ and $\mathbf{S}_2$ respectively. We stress that Watrous's ZK protocol has negligible completeness and soundness errors, and in addition the simulator succeeds for arbitrary quantum poly-time verifiers on true instances, except with negligible probability.

\mypar{Prover is Corrupted}  For any real-world adversary $\calA$, we
construct an ideal-world adversary $\calS$.

\begin{pffs}{Simulator $\calS$:}{Prover is corrupted}{} 
{\sf Input}: $\calA$ as a black box; security parameter $1^n$.
    \begin{enumerate}[1.]
        \item $\calS$ initializes $\calA$ with whatever input
        state it receives.
        \item In {\bf Phase 1}, $\calS$ does the following:
            \begin{enumerate}
                \item Compute $c = \bfcomm (0^{n})$ and send it to
                $\calA$.
                \item Obtain $b \in \{0,1\}^{n}$ from $\calA$.
                \item Run \gen$(1^n)$ to obtain $(pk,sk)$.
                Send $a = pk\oplus b$ to $\calA$.
                \item Run the simulator $\mathbf{S}_1$ for ZK$_1$ with input $(c,a)$.
            \end{enumerate}
        \item In {\bf Phase 2}, $\calS$ obtains $(x,e)$ and executes ZK$_2$ with $\calA$. If ZK$_2$ succeeds, $\calS$ decrypts $e$ to get $w = \dec_{sk}(e)$ and sends
        $(x,w)$ to $\fzk$.

        \item $\calS$ outputs whatever $\calA$ outputs.
    \end{enumerate}
\end{pffs}

Let $M_{\fzk,\calS}$ be the QTM of ideal-world interaction between
$\calS, \fzk$ and $\bfV$; and let $M_{\Pi_{\tt ZK}, \calA}(P)$ describing
real-world interaction between $\calA$ and $\bfV$.

    \begin{lemma}
$M_{\Pi_{\tt ZK}, \calA(P)} \qc M_{\fzk, \calS}$.    \label{lemma:sazkp}
    \end{lemma}

\begin{proof} We define a sequence of machines to form a hybrid argument:

\begin{pffs}{Hybrid Machines:}{relating $M_{\Pi_{\tt ZK}, \calA(P)}$ and $M_{\fzk,\calS}$}{}
\begin{itemize}
\item {\bf ${M}_0 := M_{\fzk, \calS}$}. Specifically, on any input state, the output has two parts: one part corresponds to the adversary $\calA$'s output state, and the other corresponds to the dummy verifier's output, which is accepting if $w$ obtained by $\calS$ in step 3 is a valid witness, i.e., $(x,w)\in R_L$. %Otherwise, this part is empty.
     %Here instead of thinking of $M_0$ as having three components $\hat\p_I, %%view it as a single machine that runs %$\hat\p_I$, and outputs what $\hat\p_I$ would output. $M_0$ also %outputs $x$ iff. $\hat\p_I$ accepts in the ZK$_2$ and the message %$w$ after decryption is a true witness, i.e., $(x,w) \in R_L$.

\item {$M_1$}: differ from $M_0$ only in that $M_1$ always let the dummy verifier accept as long as ZK$_2$
succeeds.

\item {$M_2$}: differs from $M_1$ in the message $a$ in Phase 1:
instead of sending $pk\oplus b$, in $M_2$, $a\gets\{0,1\}^{n}$ is
set to be a uniformly random string.

\item {$M_3$}: in the first step of Phase 1, $\calM_3$ commits
to $a$ instead of committing to $0^n$.

\item {$M_4$}: instead of running simulator $\mathbf{S}_1$, $M_4$
executes the actual ZK$_1$ protocol. Observe that $M_4 \equiv
M_{\Pi_{\tt ZK}, \calA(P)}$. 
\end{itemize}
\end{pffs}

Now it is easy to see:
\begin{itemize}
   \item $M_0 \qc M_1$. These two QTMs would behave differently only if ZK$_2$ succeeds but $w$ is not a valid witness. %Specifically, when ciphertext $e$ does not encode a valid witness, $M_0$ does not output $x$. On the other hand, $M_1$ might still output $x$, 
Namely $\calA$ (corrupted prover) has managed to prove a false statement that $e$ encodes a true witness. By soundness property of ZK$_2$, however, this only occurs with negligible probability.
    \item Machines $M_1,\ldots, M_4$ form a simple hybrid argument. More
    specifically, each adjacent pair of machines constitutes simply related
    machines:
    \begin{itemize}
        \item $M_1$ and $M_2$ are simply related by switching valid public keys to uniformly random strings.
        \item $M_2$ and $M_3$ are simply related by changing the
        messages being committed to.
        \item $M_3$ and $M_4$ are simply related via a pair of indistinguishable QTMs $N_a$
        and $N_b$, where $N_a$ is the simulator $\mathbf{S}_1$, and
        $N_b$ is the machine describing ZK$_1$.
    \end{itemize}
\end{itemize}
Thus $ M_{\Pi_{\tt ZK}, \calA(P)}\qc M_{\fzk, \calS}$. \qedhere%, which is equivalent to concluding $\mathbf{EXEC}_{\Pi_{ZK}, \calA, \calZ} \approx \mathbf{IDEAL}_{\fzk, \calS, \calZ}$ for any poly-time QTM $\calZ$.

\end{proof}

\mypar{Verifier is Corrupted} We construct ideal world $\calS$
for any adversary $\calA$ that corrupts the verifier as follows:

\begin{pffs}{Simulator $\calS$:}{Verifier is corrupted}{} 

{\sf Input}: $\calA$ as a black box; security parameter $1^n$. 
        \begin{enumerate}[1.]
                \item $\calS$ initializes $\calA$ with whatever input
        state it receives.
            \item Wait till get $x$ from $\fzk$. Then do the following.
            \item In {\bf Phase 1}, $\calS$ behave honestly and aborts if $\calA$ aborts. Let the outcome be $pk$.
            \item In {\bf Phase 2}:
            \begin{enumerate}
                \item $\calS$ picks an arbitrary string, say $0^{w(n)}$, and
                send $e = \text{\bf Enc}_{pk}(0^{w(n)})$ to $\calA$.
                \item $\calS$ runs the simulator $\mathbf{S}_2$ for ZK$_2$ with input $(pk, e, x)$.
            \end{enumerate}
            \item $\calS$ outputs whatever $\calA$ outputs.
        \end{enumerate}
\end{pffs}

Let $M_{\fzk, \calS}$ be the QTM of ideal-world interaction between
$\bfP, \fzk$ and $\calS$; and let $M_{\Pi_{\tt ZK}, \calA(V)}$ describing the
real-world interaction between $\bfP$ and $\calA$.

    \begin{lemma}
$M_{\Pi_{\tt ZK}, \calA(V)}\qc M_{\fzk, \calS}$.    \label{lemma:sazkv}
    \end{lemma}

\begin{proof} The proof again follows a hybrid argument. We define the following hybrids. 

\begin{pffs}{Hybrid Machines:}{relating $M_{\Pi_{\tt ZK}, \calA(V)}$ and $M_{\fzk,\calS}$}{}
\begin{itemize}
\item {${M}_0 := M_{\Pi_{\tt ZK}, \calA(V)}$}.
\item {$M_1$}: $M_1$ runs the simulator $\mathbf{S}_2$ instead of invoking the actual ZK$_2$ protocol. 
\item {$M_2$}:  {$M_2$} encrypts $0^{w(n)}$ instead of a valid witness $w$. Observe that $M_2 \equiv  M_{\fzk, \calS}$.
\end{itemize}
\end{pffs}

Clearly machines $M_0$ and $M_1$ are simply related via a pair of QTMs $N_a$ and $N_b$, where $N_a$ is the simulator $\mathbf{S}_2$, and $N_b$ is the machine describing ZK$_2$. Hence they are quantum computationally indistinguishable. Showing indistinguishability of $M_1$ and $M_2$ slightly deviates from our simple hybrid argument framework. We will modify $M_1$ and $M_2$ to get two machines $\hat M_1$ and $\hat M_2$ which may run in super-polynomial time. We can then show that $\hat M_1\qc \hat M_2$\footnote{Although the machines $\hat M$ are not necessarily poly-time, we can still talk about distinguishing them by poly-time distinguishers according to Definition~\ref{def:qcm}. If the output register of $\hat M$ exceeds the dimension of the input of the distinguisher, we assume that the distinguisher just takes an arbitrary portion that fits.}, and that $\hat M_1 \qc \hat M_2$ implies $M_1 \qc M_2$. 
        
Specifically $\hat M_1$ makes one change from $M_1$: if corrupted verifier aborts during Phase 1, $\bfP$ extracts $\hat a$ from $c$ using possibly super-polynomial-time brute-force search. Because the commitment scheme is statistically binding, there is a well-defined $\hat a$ with overwhelming probability. In addition, soundness of $\zka$ ensures that $\hat a = a$ except for negligible soundness error. In this way, $\bfP$ still gets $pk:=a\oplus b$ and we let $\bfP$ send $\enc_{pk}(w)$ to the verifier even in case of abort. $\hat M_2$ is modified similarly. Namely a (super-polynomial-time) simulator extracts $pk$ and sends $\enc_{pk}(0^{w(n)})$ in case of abort. 

Note that $\hat M_1$ and $\hat M_2$ are simply related by switching the plaintexts, and therefore $\hat M_1 \qc \hat M_2$ follows by our simple hybrid argument framework. Next we claim that if $\hat M_1 \qc \hat M_2$, then $M_1 \qc M_2$. This is because that if there exists a distinguisher $D$ that tells apart $M_1$ and  $M_2$, then one can as well distinguish $\hat M_1$ from  $\hat M_2$ by ignoring the ciphertext in case of aborting and then invoking $D$. 

Therefore we have that $M_{\Pi_{\tt ZK}, \calA(V)}\qc M_{\fzk, \calS}$. \qedhere 

\end{proof}

Finally, we conclude that Theorem~\ref{thm:zkaok} holds.
\end{proof}

%---------------------------------------------%
\subsection{Putting It Together}
\label{sec:qgmw}
%---------------------------------------------%

Recall the results that we have obtained so far: 

\begin{enumerate}
    \item Under Assumptions~\ref{asn:prg} and~\ref{asn:pkc}, for any well-formed two-party functionality $\calF$, there is a classical protocol $\pi^{\fzk}$ that quantum-UC emulates $\calF$ in the $\fzk$-hybrid model. (Theorem~\ref{thm:qclos})
    \item Under \asnref{prg} and~\ref{asn:pkc}, There exists classical protocol $\pi_{\tt ZK}$ that \cqsa~emulates $\fzk$. (Theorem~\ref{thm:zkaok})
\end{enumerate}

Applying our modular composition theorem  (Theorem~\ref{thm:sacomposition}) to the above, we obtain the main theorem: 
\begin{theorem} 
For any well-formed classical two-party functionality $\calF$, there exists a classical protocol $\Pi$ that \cqsa~realizes $\calF$ against malicious static quantum adversaries in the plain model, under Assumptions~\ref{asn:prg} and \ref{asn:pkc}. % and~\ref{asn:spkc}. 
\label{thm:qgmw}
\end{theorem}

%---------------------------------------------------------------------------%
%\section{Detailed Equivalence Between $\fzk$ and $\fcf$}
\section{Equivalence Between $\fzk$ and $\fcf$}
\label{sec:zk=cf}
%---------------------------------------------------------------------------%
\ifnum\old=1
In this section, we show that $\calG_{ZK}$ is equivalent to $\calG_{CF}$
in the quantum UC model, in the sense that we can construct a
classical protocol in the $\calG_{CF}$-hybrid model that quantum
UC-emulates $\calG_{ZK}$, and vice versa.
%---------------------------------------------------------------------------%
\subsection{From $\calG_{CF}$ to $\calG_{ZK}$}
\label{sssec:cf2zk}
%---------------------------------------------------------------------------%
We first show that a ZKAoK
protocol can be constructed in the $\calG_{CF}$-hybrid model.

\begin{prop}
Under Assumptions~\ref{asn:prg} and~\ref{asn:pkc}, there is a
protocol $\Pi_{ZK}^{\calG_{CF}}$ that quantum UC-emulates $\calG_{ZK}$ in
the $\calG_{CF}$-hybrid model. \label{prop:cf2zk}
\end{prop}

% \mypar{Quantum Witness Indistinguishable Proofs}
We start with a crucial building block: {\em
witness-indistinguishable proofs}. Let $\langle \bfP, \bfV\rangle$ be an
interactive proof (or argument) system for a $\class{NP}$-relation
$R_L$. Denote by $R_L(x)$ the set of all witnesses of $x$. We
consider a real-world execution of protocol $\langle\bfP, \bfV\rangle$
in the quantum UC model. Namely, there is an interactive environment
participating in the execution and outputting 1 or 0 at the end. Let
$\mathbf{EXEC}_{\bfP,\bfV,\calZ}(x, w)$ represent the output distribution
of an execution of $\Pi$, where $\calZ$ is an interactive environment
and $\bfP$ uses $w$ as a witness.

\begin{definition}(Quantum UC Witness-indistinguishability.) Let $\Pi = \langle \bfP,
\bfV\rangle$ be an interactive proof (or argument) system for a
language $L\in\class{NP}$. We say $\Pi$ is {\em quantum UC
witness-indistinguishable (quantum UC-WI)} for $R_L$, if for any
polynomial-time QIM $\bfV$, any polynomial-time QIM $\calZ$, and any
$\sigma$, $\mathbf{EXEC}_{\bfP,\bfV,\calZ}(x, w_1) \approx
\mathbf{EXEC}_{\bfP,\bfV,\calZ}(x, w_2)$ for  any $w_1, w_2 \in
R_L(x)$.
\end{definition}

%The well-known fact that WI is preserved under parallel repetition
%still holds in the quantum setting. We state it and omit the proof.

%\begin{prop}
%\label{prop:pcwi} Quantum WI is preserved under polynomially many
%parallel compositions.
%\end{prop}

%\begin{proof}(of Proposition~\ref{prop:pcwi}). Straightforward
%hybrid argument.
%\end{proof}

We can show that if we use a statistically binding and {\em quantum
computationally hiding} commitment scheme (as the one following
Assumption~\ref{asn:prg}) in Blum's zero-knowledge proof system for
Hamiltonian Cycle~\cite{Blum86}, then the resulting protocol, call
it HC$_Q$, is quantum UC-WI.

For completeness, let's recall Blum's zero-knowledge proof for
Hamiltonian Cycle (HC).
\begin{center}
\fbox{\parbox{0.8\textwidth}{
\small{\noindent{\bf HC (Zero-knowledge proof
for Hamiltonian Cycle)}: a directed graph $x$ with $n$ vertices, $\bfP$ is given a Hamiltonian cycle $w$ (the witness) in
$x$; security parameter $1^n$

\vspace{1mm} \hrule

\begin{enumerate*}

\item $\bfP$ picks a random permutation $\sigma \gets S_n$, let $y =
\sigma(x)$. Commit to $(y, \sigma)$ and send $(\bfcomm(y),
\bfcomm(\sigma))$ to $\bfV$. Here $\bfcomm(y)$ represents a bit-by-bit
commitment to the adjacency matrix of $y$.

\item $\bfV$ picks at random a challenge bit $ch \gets \{0,1\}$ and
sends to $\bfP$.

\item $\bfP$ responds according to $ch$: if $ch = 0$, $\bfP$ opens all commitments in step 1 to $\bfV$; if $ch =1$, $\bfP$ reveals a Hamiltonian cycle in $\sigma(x)$, namely it opens the entries in $\bfcomm(y)$ that correspond to
$\sigma(w)$.

\item $\bfV$ verifies $\bfP$'s response: if $ch = 0$, $\bfV$ checks that the commitment is opened correctly;
if $ch =1$, $\bfV$ checks that the revealed edges form a cycle in $y$.
$\bfV$ accepts if all checks succeed.
\end{enumerate*}
} }}
\end{center}

\begin{prop}
HC$_Q$ is quantum UC-WI. \label{prop:npqwi}
\end{prop}

\begin{proof}%[Proof of Proposition~\ref{prop:npqwi}]
Let $E_{1} = \mathbf{EXEC}_{\bfP,\bfV,\calZ}(x, w_1, 1^n)$ and $E_{2} =
\mathbf{EXEC}_{\bfP,\bfV,\calZ}(x, w_2, 1^n)$. We want to show that
$$| \Pr(E_{1} = 1) - \Pr(E_{2} = 1)| = \delta(n) \leq negl(n).$$
Note that the distribution of $ch$ in $E_1$ and $E_2$ are identical
since the first messages in $E_1$ and $E_2$ are identically
distributed (commitments to the same objects) and the distribution
of $ch$ is uniquely determined by prover's first message and $\bfV$'s
local configuration. So we know $\Pr(ch = b \text{ in } E_1) =
\Pr(ch = b \text{ in } E_2) \stackrel{def}{=} \Pr(ch = b)$ for $b =
0,1$.
%$\Pr(E_{i} = 1) = \Pr(E_{i} = 1 \wedge ch = 0) + \Pr(E_{i} = 1
%\wedge ch = 1 )$ for $i = 1, 2$.
Let
$$\delta_0(n) = |\Pr(E_1 = 1 \wedge ch = 0 \text{ in } E_1) -
\Pr(E_{2} = 1 \wedge ch = 0 \text{ in } E_2)| $$
$$\text{and } \delta_1 (n) = |\Pr(E_{1} = 1 \wedge ch = 1 \text{ in } E_1) -
\Pr(E_{2} = 1 \wedge ch = 1 \text{ in } E_2)|$$  Then $\delta(n)
\leq \delta_0(n) + \delta_1(n)$ by the triangle inequality. We show
both $\delta_0$ and $\delta_1$ are negligible.

\begin{claim}
$ \delta_0(n) = |\Pr(E_{1} = 1 \wedge ch = 0 \text{ in } E_1) -
\Pr(E_{2} = 1 \wedge ch = 0 \text{ in } E_2)| = 0.$
\label{claim:del0}
\end{claim}
\begin{proof}%[Proof of Claim~\ref{claim:del0}]
$\delta_0(n) = |\Pr(E_{1} = 1 | ch = 0 \text{ in } E_1)\cdot \Pr(ch
= 0\text{ in } E_1) - \Pr(E_{2} = 1 | ch = 0 \text{ in } E_2 )\cdot
\Pr(ch = 0\text{ in } E_2)| = \Pr(ch = 0)\cdot |\Pr(E_{1} = 1 | ch =
0 \text{ in } E_1) - \Pr(E_{2} = 1 | ch = 0 \text{ in } E_2 )|$.
However, provided that $ch = 0$, $E_1$ and $E_2$ become identical
since the response is just opening of the commitment. Therefore
$\Pr(E_{1} = 1 | ch = 0 \text{ in } E_1) = \Pr(E_{2} = 1 | ch = 0
\text{ in } E_2)$ and hence $\delta_0 (n) = 0$.
\end{proof}

\begin{claim}
$ \delta_1(n) = |\Pr(E_{1} = 1 \wedge ch = 1 \text{ in } E_1) -
\Pr(E_{2} = 1 \wedge ch = 1 \text{ in } E_2 )| \leq negl(n).$
\label{claim:del1}
\end{claim}
\begin{proof}%[Proof of Claim~\ref{claim:del1}]
We present an imaginary
experiment $E(1^n)$ such that $|\Pr(E_{i} = 1 \wedge ch = 1 \text{
in } E_i) - \Pr(E(1^n) = 1)| \leq negl(n)$ for both $i = 1, 2$. It
then follows easily by triangle-inequality that $\delta_1(n) =
|\Pr(E_{1} = 1 \wedge ch = 1 \text{ in } E_1) - \Pr(E_{2} = 1 \wedge
ch = 1 \text{ in } E_2)| \leq \sum_{i = 1,2}|\Pr(E_{i} = 1 \wedge ch
= 1 \text{ in } E_i ) - \Pr(E(1^n) = 1)| \leq negl(n)$.

\begin{center}
 \fbox{\parbox{0.8\textwidth}{
{\small{\noindent\bf Experiment $E(1^n)$} \vspace{1mm} \hrule

\begin{enumerate*}

\item $\bfP$ picks a random permutation $\sigma$,  commits to a complete
  graph $K_n$, and sends commitments $\bfcomm(\sigma)$ and $\bfcomm(\sigma(K_n))$ to
  $\bfV$. Here $K_n$ is the compete graph on $n$ vertices.

\item $\bfV$ picks the challenge bit $ch = 1$ and
sends it to $\bfP$.

\item $\bfP$ responds by opening a cycle in $\sigma(K_n)$.

\item $\bfV$ verifies $\bfP$'s response.

\end{enumerate*}
}}}
\end{center}% \hrule \vspace{1mm}

Compared to a real-execution in which $ch = 1$, $E(1^n)$ differs
only in step 1, where the commitment is to a permutation on
$\sigma(K_n)$ as opposed to $\sigma(x)$. However, the hiding
property of the commitment scheme ensures that no polynomial-time
machines can distinguish them. Hence $|\Pr(E_{i} = 1 \wedge ch = 1
\text{ in } E_i) - \Pr(E(1^n) = 1)| \leq negl(n)$ for both $i = 1,
2$.
\end{proof}

Therefore we conclude that HC$_Q$ is quantum UC-WI.
\end{proof}

Using a polynomial number of parallel repetitions of HC$_Q$, we have
a quantum UC-WI protocol for \class{NP} with negligible soundness
error which we call $\Pi_{WI}$ and will use in later constructions.

%---------------------------------------------------------------------------%
%\subsubsection{Quantum UC-secure ZKAoK}
%\label{sssec:qwi}
%---------------------------------------------------------------------------%

% \mypar{Quantum UC-secure ZKAoK}
We now construct $\Pi^{\calG_{CF}}_{ZK}$ that quantum UC-emulates
$\calG_{ZK}$ in the $\calG_{CF}$-hybrid model.

Let $L$ be an \class{NP} language and $R_L$ be the corresponding
NP-relation. Let $\pgen$ be a quantum secure pseudorandom generator
as in Assumption~\ref{asn:prg}, and let $\mathcal{E} = (\gen,\enc,
\dec)$ be an encryption scheme as in Assumption~\ref{asn:pkc}. We
define another relation $$R = \{((x_1,x_2,pk,e),w) | ( \exists r:
\enc_{pk}(w, r) = e \wedge (x_1,w) \in R_L) {\text{ or }} (\pgen(w)
= x_2) \}$$ It is clear that $R$ is an NP-relation, and thus there
is a WI proof for $R$. The key idea of constructing
$\Pi_{ZK}^{\calG_{CF}}$ is to exploit the common reference string (CRS)
in some clever way.
%Note that a true CRS, i.e., one generated by $\g_{CRS}$, is
%uniformly random.
We will interpret a CRS $s$ as two parts $(s_1, s_2)$, where $s_1 =
pk$ will be used as a public key $pk$ for $\mathcal{E}$, and $s_2$
will sometimes be an output string of $\pgen$. Our
$\Pi_{ZK}^{\calG_{CF}}$ has a simple form then: $\bfP$ and $\bfV$ get
$s=(s_1, s_2)$, $\bfP$ sends $x$ and $e = \enc_{s_1}(w)$ to $\bfV$, and
next they run a WI protocol on $(x_1 = x, x_2 = s_2, pk = s_1, e)$
using witness $w$.
%Here we only give an informal justification why $\Pi_{ZK}^{\calG_{CF}}$ is a ZKAoK for $R_L$.
%Completeness is straightforward.
Intuitively, if the adversary $\calA$ corrupts the verifier $\bfV$, then $\calS$ can choose a
fake CRS $s' = (s_1', s_2')$ where $s_2'$ is generated by $\pgen$
with random seed $r$, i.e., $s_2' = \pgen(r)$. Then it generates an
arbitrary ciphertext as $e$ and uses $r$ as a witness in the WI
proof, and witness-indistinguishability ensures the $\calA$ can not
distinguish from the case where $\bfP$ uses a real witness $w$ of $x$. If the prover is corrupted, $\calS$ can simply generate
$(pk,sk)\gets \gen$ and assign $pk$ as $s_1'$, while $s_2'$ is still
uniformly chosen. Therefore, whenever $\calA$ convinces $\calS$
in the WI protocol, $\calS$ then decrypts (it knows $sk$) $w =
\dec_{sk}(e)$. However, there is one subtlety. Namely, $R$ has two
witnesses, either a real $w$ (which is what we really ask for) s.t.
$(x,w)\in R_L$ or a random seed $r$ s.t. $\pgen(r) = s_2'$. We do
not want $\calA$ to be capable of achieving the latter case. This is
easy to guarantee though, because we can choose a generator $\pgen$
with sufficient expansion factor, e.g., if $\pgen: \{0,1\}^n\to
\{0,1\}^{3n}$. Then given a uniformly random $3n$-bit string $s_2'$,
the probability that there is a seed $r \in \{0,1\}^n$ getting
mapped to $s_2'$ is negligible. Thus whenever a prover succeeds in
WI, it must have proved the statement with respect to $R_L$ rather
than with respect to $\pgen$. The formal description of protocol
$\Pi^{\calG_{CF}}_{ZK}$ follows.

\begin{center}
\fbox{\parbox{0.8\textwidth}{
{\small {\noindent\bf UC-secure ZKAoK Protocol $\Pi^{\calG_{CF}}_{ZK}$}
%trusted party ${\g_{CF}}$; NP-relation $R_L$; security parameter
%$1^n$ \vspace{1mm}
\hrule
\begin{enumerate*}

\item $\bfP$ and  $\bfV$ get $s = (s_1, s_2) \in \{0,1\}^n\times \{0,1\}^{3n}$ from ${\calG_{CF}}$.

\item $\bfP$ sends $x$ and $e =\enc_{s_1}(w)$ to $\bfV$.

\item $\bfP$ and $\bfV$ invoke a WI protocol $\Pi_{WI}$ for relation $R$ with input
instance $(x_1 = x, x_2 = s_2, pk =s_1, e)$. $\bfP$ uses $w$ as a
witness for $(x_1, x_2, pk, e)$.

\item $\bfV$ outputs $x$ if it accepts in $\Pi_{WI}$.

\end{enumerate*}
}}}
\end{center} %\hrule \vspace{1mm}

\begin{prop}

\label{prop:za} The classical protocol $\Pi^{\calG_{CF}}_{ZK}$ quantum
UC-emulates $\calG_{ZK}$.
\end{prop}

\begin{proof} We first deal with the case
in which the prover is corrupted.
\begin{center}
\fbox{\parbox{0.8\textwidth}{
{\small{\noindent\bf Ideal-world adversary $\calS$}: given
adversary $\calA$ corrupting $\bfP$; security parameter $1^n$; \vspace{1mm}

\hrule

\begin{enumerate*}

\item $\calS$ initializes $\calA$ with whatever input
        state it receives from the environment.
\item $\calS$ internally generates $(pk,sk)\gets \gen(1^n)$ and set $s_1' = pk$. Choose $s_2' \gets \{0,1\}^{3n}$ uniformly at random. Let $s'=(s_1', s_2')$ be the fake CRS and it is given to $\calA$.
\item When $\calS$ receives $(x, e)$ from $\calA$, it decrypts $e$ to get $w = \dec_{sk}(e)$.

\item $\calS$ runs $\Pi_{WI}$ with $\calA$ on
input instance $(x, s_2', s_1', e)$ where $\calS$ plays the role
of a verifier. If $\calS$ accepts in WI, it sends $(x,w)$ to
${\calG_{ZK}}$.

\item $\calS$ outputs whatever $\calA$ outputs.

\end{enumerate*}
}
}}\end{center}%\hrule\vspace{1mm}

\begin{lemma}\label{lemma:aok}

For any poly-time QIM $\calZ$ and any $\sigma$,
$\mathbf{EXEC}_{\Pi_{ZK}^{\calG_{CF}}, \calA, \calZ} \approx
\mathbf{IDEAL}_{\calG_{ZK}, \calS, \calZ}$.
\end{lemma}
\vspace{-5mm}
\begin{proof}%(of Lemma~\ref{lemma:aok})
Consider the two QIMs $M_{\calS}$ and $M_{\calA}$. Similar to
the proof of Lemma~\ref{lemma:sazkp}, we can define an intermediate
machine $M$ which behaves differently from $M_{\calS}$ only in
that it does not check if $w$ is a true witness and outputs $x$ immediately $\Pi_{WI}$ succeeds. Then soundness of $\Pi_{WI}$ ensures $M_{\calS} \approx_I M$. Moreover $M$ and $M_{\calA}$ are
simply related by switching between a valid public key and a truly
random string. The lemma then follows from Assumption~\ref{asn:pkc}.
%\fnote{Fixed. uc model security proofs. check if all can be written consistently in the language of simple hybrid argument}
\end{proof}
Now we consider the case where $\calA $ corrupts the verifier.

\begin{center}
\fbox{\parbox{0.8\textwidth}{
{\small{\noindent\bf Ideal-world adversary $\calS$}: given
adversary $\calA$ corrupting $\bfV$; honest $\bfP$;
security parameter $1^n$;
\hrule

\begin{enumerate*}

%\item $\calZ$ prepares inputs for $\bfP_I$ and $\hat\bfV_I$ as usual.

%\item Upon receiving $(x,w)$ from $\p_I$, ${\g_{ZK}}$ verifies $(x,w)
%\stackrel{?}{\in} R_L$. If yes set $vb = 1$, otherwise set $vb = 0$.
%It then gives $(x,vb)$ to $\hat\bfV_I$.
\item $\calS$ initializes $\calA$ with whatever input
        state it receives from the environment.

\item Wait till receives $x$ from $\calG_{ZK}$. Then $\calS$ internally generates $s_1'\gets \{0,1\}^n$. It also generates $r \gets \{0,1\}^n$ and sets $s_2' =
\pgen(r)$. Let $s'=(s_1', s_2')$ be the fake CRS and it is given to
$\calA$.

\item $\calS$ sends $x$ and $e = \enc_{s_1'}(0^n)$ to $\calA$ and then invokes $\Pi_{WI}$ with $\calA$ on input instance $(x, s_2', s_1', e)$. $\calS$ uses $r$ as a witness.

\item $\calS$ outputs whatever $\calA$ outputs.

\end{enumerate*}
}
}}
\end{center}

\begin{lemma}
\label{lemma:zk} For any poly-time QIM $\calZ$ and any $\sigma$,
$\mathbf{EXEC}_{\Pi_{ZK}^{\calG_{CF}}, \calA, \calZ} \approx
\mathbf{IDEAL}_{\calG_{ZK}, \calS, \calZ} $.
\end{lemma}

\begin{proof}%(of Lemma~\ref{lemma:zk}).
We show QIMs $M_{\calS}$ and $M_{\calA}$ are
interactively-indistinguishable via a sequence of indistinguishable
machines defined as follow.

\begin{center}
 \fbox{\parbox{0.8\textwidth}{
{\small
\begin{itemize*}
\item {$M_0 := M_{\calS}$}. The ideal-world machine describing
$\bfP, \calS$ and $\calG_{ZK}$ as a single interactive machine.
\item{$M_1$}: same as $M_0$ except that the ciphertext is changed
from $\enc_{s_1'}(0^n)$ to $e = \enc_{s_1'}(w)$. Here $w$ is a
witness for $x$, i.e., $R_L(x,w) = 1$.
\item{$M_2$}: identical to $M_1$ except that
$M_2$ uses $w$ as a witness in the $\Pi_{WI}$.
\item{$M_3$}: $s_2'$ is also chosen uniformly random, rather than
pseudorandom. Note $M_3$ is exactly the real-world machine
$M_{\calA}$.
\end{itemize*}}

}}
\end{center}

Now we can see that:
\begin{itemize*}
    \item $M_0 \approx_I M_1$ because they simply related by changing the plaintext of encryption $e$.
    \item $M_1 \approx_I M_2$ because $\Pi_{WI}$ is quantum UC-WI.
    \item $M_2 \approx_I M_3$ because they are simply related by
    switching a pseudorandom string to a uniformly random string.
\end{itemize*}
Thus Lemma~\ref{lemma:zk} holds. \end{proof} We finally get
$\Pi_{ZK}^{\calG_{CF}}$ quantum UC-emulates $\calG_{ZK}$.
\end{proof}

\fi

We have seen that $\fzk$ functionality is sufficient to realize any other functionality. It is interesting to find out if this holds as well for other functionalities. More generally, we may ask what the relations of different functionalities are. In this section, we show that $\fzk$ and $\fcf$ are equivalent in the sense that one can be UC-realized from the other.  

\begin{theorem}[Equivalence between $\fzk$ and $\fcf$]
\label{thm:zk=cf}
\begin{enumerate}
	\item Under Assumption~\ref{asn:prg}, there is a constant-round protocol {$\Pi_{\tt CF}^{\fzk}$} that \cquc~ emulates $\fcf$ in the $\fzk$-hybrid model.
	\item Under Assumptions~\ref{asn:prg} and~\ref{asn:pkc}, there is a constant-round protocol $\Pi_{\tt ZK}^{\fcf}$ that   \cquc~emulates $\fzk$ in the $\fcf$-hybrid  model.
\end{enumerate}
\end{theorem}

It is possible to obtain more connections between different functionalities. For example, \cite{CF01} gives a ZK protocol that statistically UC and hence \cquc~emulates $\fzk$ in the $\fcom$-hybrid model. On the other hand, our Theorem~\ref{thm:qclos} implies that $\fcom$ can be \cquc~realized in $\fzk$-hybrid model. Thus $\fzk$ and $\fcom$ are equivalent in the computationally quantum-UC model. See~\cite{FKSZZ13} for a systematic study of the reducibility and characterizing functionalities in the quantum UC model.

%===================================%
\subsection{From $\fzk$ to $\fcf$}
\label{subsec:zk2cf}
%===================================%

Theorem~\ref{thm:qclos} already implies that $\fcf$ can be \cquc~ realized from $\fzk$-hybrid model. However, that relies on the generic construction of CLOS, which is typically not optimal in terms of the number of rounds (i.e., \emph{round complexity}) and the amount of messages exchanged (i.e., \emph{communication complexity}). Here we give a direct reduction which is simple and more efficient. Specifically, we show that  the parallel coin-flipping protocol of Lindell~\cite{Lin03}, once executed in $\fzk$-hybrid model, i.e., the (stand-alone) ZKAoK protocol is replaced by the ideal protocol for $\fzk$, is \cquc~secure. This yields a \emph{constant-round} protocol for $\fcf$, and we need only Assumption~\ref{asn:prg}: existence of a quantum-secure PRG. The protocol is shown below. 

%\begin{Lemma}
%Under \asnref{prg}, $\cqcucred{\fcf}{\fzk}$ with a constant-round classical reduction.
%\label{lemma:fzk2fcf}
%\end{Lemma}

\begin{pffs}{Coin-Flipping Protocol}{$\Pi_{\tt CF}^{\fzk}$}{}
\begin{enumerate}[1.]
    \item $A$ chooses $a \gets \{0,1\}^{n}$ at random, and sends $B$ a commitment of $a$: $c = \bfcomm(a,r)$.
    \item $A$ proves knowledge of $(a,r)$ using $\fzk$.

    \item $B$ sends $b \gets \{0,1\}^{n}$ to $A$.

    \item $A$ sends $B$ string $a$.

    \item $A$ proves to $B$ that $c$ is indeed a commitment of $a$ using $\fzk$.

    \item $A$ and $B$ set $s = a\oplus b$ as the outcome.

\end{enumerate}
\end{pffs}

 We give proofs for corrupted $A$ and corrupted $B$ separately. 

%\end{Proposition}

\mypar{Player $A$ is Corrupted} We construct an ideal world
$\calS$ for any adversary $\calA$ corrupting $A$.
\begin{pffs}{Simulator $\calS$:}{$A$ is corrupted}{}
Input: $\calA$ as a black box; security parameter $1^n$
        \begin{enumerate}[1.]
        \item $\calS$ initializes $\calA$ with whatever input
        state it receives from the environment.
            \item $\calS$ obtains $s$ from $\fcf$ which is
                chosen uniformly at random $s\gets\{0,1\}^n$.
                \item $\calS$ receives a commitment $c$ % = \bfcomm(a)$ 
                from $\calA$.
                \item $\calA$ shows knowledge of $(a,r)$ to
                $\fzk$, which is emulated by $\calS$ here.
                $\calS$ verifies if $c = \bfcomm(a,r)$ and aborts if
                not. This allows $\calS$ to learn $a$.
                \item $\calS$ sends $b = a\oplus s $ to $\calA$.
                \item $\calA$ sends $a$ to $\calS$. $\calS$ aborts if $\calA$ sends some other string not equal to $a$.
                \item $\calA$ needs to prove that $c$ is a valid commitment of $a$.  It sends $((c,a),r)$ to $\fzk$. $\calS$ verifies them. Abort if verification fails.
                \item If $\calA$ aborts at any point, $\calS$
                aborts $\fcf$. Otherwise, instruct $\fcf$ to send $s$ to the other (dummy) party $\tilde B$.
                \item $\calS$ outputs whatever $\calA$ outputs.
        \end{enumerate}
\end{pffs}

   \begin{claim} For any $\calA$ corrupting $A$, $\mac_{\Pi_{\tt CF}^{\fzk}, \calA} \qcii \mac_{\fcf, \calS}  $.
   \label{claim:pcfa}
    \end{claim}

  \begin{proof} Because $s$ is chosen
   uniformly, $b = a\oplus s$ is also uniformly random. The adversary must behave identical in the real world and the ideal world, and the two machines will look identical from the perspective
    of the environment.

\end{proof}

\mypar{Player $B$ is Corrupted} For any real-world adversary $\calA$ that corrupts $B$, we construct an ideal-world adversary $\calS$.

\begin{pffs}{Simulator $\calS$:}{$B$ is corrupted}{}

Input: $\calA$ as a black box; security parameter $1^n$
        \begin{enumerate}[1.]
            \item $\calS$ initializes $\calA$ with whatever input
        state it receives from the environment.
                \item $\calS$ obtains $s$ from $\fcf$ which is
                chosen uniformly at random $s\gets\{0,1\}^n$.
                \item $\calS$ computes $c = \bfcomm(0^n)$ and sends it to
                $\calA$.
                \item $\calS$ plays the role of $\fzk$ and sends $c$ to $\calA$.               
                \item Obtain $b \in \{0,1\}^n$ from $\calA$.
                \item $\calS$ sends $a = s\oplus b$ to $\calA$.
                \item $\calS$ mimics $\fzk$ and sends 
                $(c,a)$ to $\calA$.
                \item If $\calA$ aborts at any point, $\calS$
                aborts.
                \item $\calS$ outputs whatever $\calA$ outputs.
            \end{enumerate}
\end{pffs}

   \begin{claim} For any $\calA$ corrupting $B$, $M_{\Pi_{\tt CF}^{\fzk}, \calA} \qcii M_{\fcf, \calS}  $.
   \label{claim:cfb}
    \end{claim}

    \begin{proof} We define an intermediate machine $M$ which behaves differently from $\mac_{\fcf, \calS}$ merely in that a uniformly random string $a\gets\{0,1\}^n$ is chosen and sent to $\calA$ in $M$, instead of sending $a = s\oplus b$. Then observe that the only difference between $\mac$ and $\mac_{\Pi_{\tt CF}^{\fzk}, \calA}$ appears in the first commitment message: $M$ commits to $0^n$ while $\mac_{\Pi_{\tt CF}^{\fzk}, \calA}$ commits to $a$. Hence we can claim that: 

\begin{itemize}
    \item $\mac_{\fcf, \calS}  \equiv M$ since $s$ is chosen
    uniformly at random by $\fcf$ and hence $s\oplus b$ is still
    uniformly random just as $a$ in $M$. Thus the two
    machines are identical.
    \item $M \qcii \mac_{\Pi_{\tt CF}^{\fzk},\calA}$ because they are simply related by changing the underlying message of a commitment.
\end{itemize}
%Hence Claim~\ref{claim:cfb} holds. 
\end{proof}

%====================================%
\subsection{From $\fcf$ to $\fzk$}
\label{subsec:cf2zk}

We construct a classical \emph{constant-round} protocol for $\fzk$ in the $\fcf$-hybrid model. Our $\Pi_{\tt ZK}^{\fcf}$ protocol uses a standard transformation from a witness-indistinguishable (WI) proof system in
the $\fcf$-hybrid model.  We give a definition for WI against quantum adversaries and show a repetition theorem analogous to the classical setting to amplify the soundness of WI protocols. We also show that Blum's 3-round ZK protocol for Hamiltonian Cycle is in fact quantum-secure WI under suitable assumptions. %This means that quantum-secure WI exists assuming ...   

\begin{definition}[Quantum computationally witness-indistinguishable $\wiq$] 
Let $\Pi = \langle P, V\rangle$ be an interactive proof (or argument) system for a
language $L\in \class{NP}$. We say $\Pi$ is {\em
quantum computational witness-indistinguishable} for $R_L$, if for any
polynomial-time QIM $\bfV^*$, any two collections $\{w_x^1\}_{x\in L}$ and $\{w_x^2\}_{x\in L}$ with $w_x^i \in R_L(x), i=1,2$, the two machines $M_1: = \{M_{w_x^1,\bfV^*}\}_{x\in L}$ and $M_2:= \{ M_{w_x^2,\bfV^*}\}_{x\in L}$ are quantum computationally indistinguishable (i.e., $M_1 \qc M_2$). Here $M_{w_x^i,\bfV^*}$ denotes the composed machine of $P$ and $\bfV^*$ on instance $x$, and $P$ uses witness $w_x^i$. %$\langle P(w_x^1),\bfV^*\rangle(x)$. %$\mac_{P(w_1),\bfV^*} \qc \mac_{P(w_2),\bfV^*}$ for any  $w_1, w_2 \in R_L(x)$.
\label{def:qwi}
\end{definition}

\fnote{machine: collection of machines that work on specific input string or input length}

We know that classically WI is preserved under parallel repetition when the prover is efficient~\cite{FS90}. By a similar argument, one can also show that $\wiq$ protocols remain $\wiq$ under parallel reception. This is useful for reducing the soundness error of a $\wiq$ protocol. Here we only state this property and skip the proof. 

\begin{lemma}[Parallel composition of $\wiq$ protocols] 
Let $L\in \class{NP}$ and suppose that $\langle P,V \rangle$ is a $\wiq$ for $R_L$ and $P$ is polynomial time given a witness. Let $q(\cdot)$ be a polynomial and let $\langle P^q,V^q\rangle$ be machines so that they invoke $\langle P,V \rangle$ $q$-times in parallel. Specifically, on common input $\{x_i: i=1,\ldots,q\}$ and (private) input $\{w_i:i=1,\ldots,q\}$ to $P^q$, the $i^\mathrm{th}$ invocation is $\langle P(w_i),V\rangle(x_i)$.  Then $\langle P^q,V^q\rangle$ is $\wiq$ for the relation $$R_L^q: = \{(\{x_i:i=1,\ldots,q\},\{w_i:i=1,\ldots,q\}): \forall i, (x_i,w_i) \in R_L\}\, .$$
\label{lemma:prwiq}
\end{lemma}

It is easy to see that quantum ZK implies  $\wiq$. Meanwhile if we use a statistically binding and {\em quantum
computationally hiding} commitment (as the one following
Assumption~\ref{asn:prg}) in Blum's basic (3-round) ZK protocol~\cite{Blum86}, we can show that the resulting protocol, call it HC$_Q$, is quantum ZK using the techniques from~\cite{Wat09}. Therefore, we claim that HC$_Q$ is $\wiq$. Using a polynomial number of parallel repetitions of HC$_Q$, we have
a $\wiq$ protocol for \class{NP} with negligible soundness
error which we call $\Pi_{WI}$ and will use in later constructions.

\fnote{I removed the proof that Blum's protocol is qWI.}

We now construct $\pi^{\fcf}_{\tt ZK}$ that quantum-UC emulates
$\fzk$ in the $\fcf$-hybrid model.

Let $L$ be an \class{NP} language and $R_L$ be the corresponding NP-relation. Let $\pgen$ be a quantum secure pseudorandom generator as in Assumption~\ref{asn:prg}, and let $\mathcal{E} = (\gen,\enc,\dec)$ be an encryption scheme as in Assumption~\ref{asn:pkc}. We define another relation $$R = \{((x_1,x_2,pk,e),\hat w) | ( \hat w = (w,r) \wedge \enc_{pk}(w, r) = e \wedge (x_1,w) \in R_L) {\text{ or }} (\pgen(\hat w)
= x_2) \} \, .$$ 
It is clear that $R$ is an NP-relation, and thus there is a WI proof for $R$. The key idea of constructing $\Pi_{\fzk}^{\fcf}$ is to exploit the outcome of the coin-flipping in some clever way. We will interpret the coins $s$ as two parts $(s_1, s_2)$, where $s_1 = pk$ will be used as a public key $pk$ for $\mathcal{E}$, and $s_2$ will sometimes be an output string of $\pgen$. Our $\Pi_{\tt ZK}^{\fcf}$ has a simple form then: $\bfP$ and $\bfV$ get $s=(s_1, s_2)$, $\bfP$ sends $x$ and $e = \enc_{s_1}(w)$ to $\bfV$, and next they run a WI protocol on $(x_1 = x, x_2 = s_2, pk = s_1, e)$ using witness $w$. Intuitively, if the adversary $\calA$ corrupts the verifier $\bfV$, then $\calS$ can choose a fake $s' = (s_1', s_2')$ where $s_2'$ is generated by $\pgen$ with random seed $r$, i.e., $s_2' = \pgen(r)$. Then it generates an arbitrary ciphertext as $e$ and uses $r$ as a witness in the WI proof, and witness-indistinguishability ensures the $\calA$ can not distinguish from the case where $\bfP$ uses a real witness $w$ of $x$. If the prover is corrupted, $\calS$ can simply generate $(pk,sk)\gets \gen (1^n)$ and assign $pk$ as $s_1'$, while $s_2'$ is still uniformly chosen. Therefore, whenever $\calA$ convinces $\calS$ in the WI protocol, $\calS$ then decrypts (it knows $sk$) $w = \dec_{sk}(e)$. However, there is one subtlety. Namely, $R$ has two witnesses, either a real $w$ (which is what we really ask for) s.t. $(x,w)\in R_L$ or a random seed $r$ s.t. $\pgen(r) = s_2'$. We do not want $\calA$ to be capable of achieving the latter case. This is easy to guarantee though, because we can choose a generator $\pgen$ with sufficient expansion factor, e.g., if $\pgen: \{0,1\}^n\to\{0,1\}^{3n}$. Then given a uniformly random $3n$-bit string $s_2'$, the probability that there is a seed $r \in \{0,1\}^n$ getting mapped to $s_2'$ is negligible. Thus whenever a prover succeeds in WI, it must have proved the statement with respect to $R_L$ rather than with respect to $\pgen$. The formal description of protocol $\Pi_{\tt ZK}^{\fcf}$ follows.

\begin{pffs}
{\bf UC-secure ZKAoK Protocol}{ $\Pi^{\fcf}_{\tt ZK}$}{}
%trusted party ${\g_{CF}}$; NP-relation $R_L$; security parameter
%$1^n$ \vspace{1mm}

\begin{enumerate}

\item $\bfP$ and  $\bfV$ get $s = (s_1, s_2) \in \{0,1\}^n\times \{0,1\}^{3n}$ from $\fcf$.

\item $\bfP$ sends $x$ and $e =\enc_{s_1}(w,r)$ to $\bfV$.

\item $\bfP$ and $\bfV$ invoke a WI protocol $\Pi_{WI}$ for relation $R$ with input
instance $(x_1 = x, x_2 = s_2, pk =s_1, e)$. $\bfP$ uses $(w,r)$ as a
witness for $(x_1, x_2, pk, e)$.

\item $\bfV$ accepts if it accepts in $\Pi_{WI}$.

\end{enumerate}
\end{pffs}

\begin{lemma}
\label{lemma:za} 
The classical protocol $\Pi^{\fcf}_{\tt ZK}$ \cquc~emulates $\fzk$.
\end{lemma}

\begin{proof} We first deal with the case in which the prover is corrupted.

\begin{pffs}{\bf Simulator $\calS$:}{prover is corrupted}{}

Input: adversary $\calA$ ; security parameter $1^n$. 

\begin{enumerate}

\item $\calS$ initializes $\calA$ with whatever input
        state it receives from the environment.
\item $\calS$ internally generates $(pk,sk)\gets \gen(1^n)$ and set $s_1' = pk$. Choose $s_2' \gets \{0,1\}^{3n}$ uniformly at random. Let $s'=(s_1', s_2')$ be the fake coins and it is given to $\calA$.
\item When $\calS$ receives $(x, e)$ from $\calA$, it decrypts $e$ to get $w = \dec_{sk}(e)$.

\item $\calS$ runs $\Pi_{WI}$ with $\calA$ on
input instance $(x, s_2', s_1', e)$ where $\calS$ plays the role
of a verifier. If $\calS$ accepts in $\Pi_{WI}$, it sends $(x,w)$ to
$\fzk$.

\item $\calS$ outputs whatever $\calA$ outputs.

\end{enumerate}

\end{pffs}

\begin{claim}
\label{claim:aok}
For any $\calA$ corrupting the prover, $\mac_{\Pi_{\tt ZK}^{\fcf}, \calA(P)} \qcii \mac_{\fzk, \calS}  $.
\end{claim}

\begin{proof}%(of Lemma~\ref{lemma:aok})
%Similar to the proof of Lemma~\ref{lemma:sazkp}, we can 
Note that in the ideal world, if $(x,w)\notin R_L$, the dummy verifier will reject. Define an intermediate
machine $M$ in which $\fzk$ always sends $x$ to the dummy verifier (i.e. it accepts), and $M$ is identical to $M_{\fzk, \calS}$ otherwise. $M$ and $M_{\fzk,\calS}$ behave differently only when $\Pi_{WI}$ succeeds but somehow $(x,w)\notin R_L$. This however violates the soundness property of $\Pi_{WI}$. Hence $M_{\fzk, \calS} \qcii M$. Then $M$ and $M_{\Pi_{\tt ZK}^{\fcf}, \calA(P)}$ are simply related by switching between a valid public key and a truly
random string. The lemma then follows from Assumption~\ref{asn:pkc}.
%\fnote{Fixed. uc model security proofs. check if all can be written consistently in the language of simple hybrid argument}
\end{proof}
Now we consider the case where $\calA$ corrupts the verifier.

\begin{pffs}{Simulator $\calS$:}{verifier is corrupted}{}

Input: given adversary $\calA$; security parameter $1^n$;

\begin{enumerate}

\item $\calS$ initializes $\calA$ with whatever input
        state it receives from the environment.

\item Wait till it receives $x$ from $\fzk$. Then $\calS$ internally generates $s_1'\gets \{0,1\}^n$. It also generates $r \gets \{0,1\}^n$ and sets $s_2' =
\pgen(r)$. Let $s'=(s_1', s_2')$ be the fake coins and it is given to
$\calA$.

\item $\calS$ sends $x$ and $e = \enc_{s_1'}(0^n)$ to $\calA$ and then invokes $\Pi_{WI}$ with $\calA$ on input instance $(x, s_2', s_1', e)$. $\calS$ uses $r$ as a witness.

\item $\calS$ outputs whatever $\calA$ outputs.

\end{enumerate}

\end{pffs}

\begin{claim}
\label{claim:zk} 
For any $\calA$ corrupting the verifier, $\mac_{\Pi_{\tt ZK}^{\fcf}, \calA(V)} \qcii \mac_{\fzk, \calS}  $.
\end{claim}

\begin{proof}
We define a sequence of indistinguishable
machines as follows.

\begin{pffs}{}{}{}
\begin{itemize}
\item {$M_0 := M_{\fzk, \calS}$}. The ideal-world machine describing
$\bfP, \calS$ and $\fzk$ as a single interactive machine.
\item{$M_1$}: same as $M_0$ except that the ciphertext is changed
from $\enc_{s_1'}(0^n)$ to $e = \enc_{s_1'}(w)$. Here $w$ is a
witness for $x$, i.e., $R_L(x,w) = 1$.
\item{$M_2$}: identical to $M_1$ except that
$M_2$ uses $w$ as a witness in the $\Pi_{WI}$.
\item{$M_3$}: $s_2'$ is also chosen uniformly random, rather than
pseudorandom. Note $M_3$ is exactly the real-world machine
$\mac_{\Pi_{\tt ZK}^{\fcf}, \calA(V)}$.
\end{itemize}

\end{pffs}

Now we can see that:
\begin{itemize}
    \item $M_0 \qcii M_1$ because they simply related by changing the plaintext of encryption $e$.
    \item $M_1 \qcii M_2$ because $\Pi_{WI}$ is $\wiq$. Otherwise we can construct a malicious $\bfV^*$ such that $M_{w_x^1,\bfV^*}$ and $M_{w_x^2,\bfV^*}$ become distinguishable.
    \item $M_2 \qcii M_3$ because they are simply related by
    switching a pseudorandom string to a uniformly random string.
\end{itemize}

Thus Claim~\ref{claim:zk} holds. 

\end{proof} 

We finally get $\Pi_{\tt ZK}^{\fcf}$ quantum-UC emulates $\fzk$.

\end{proof}

%---------------------------------------------------------------------------%
%\section{Applications: Equivalence Between $\fzk$ and $\fcf$}
%\section{Applications and Discussions}
%\label{sec:zk=cf}
%---------------------------------------------------------------------------%

%\begin{Theorem} $ \fzk \equiv \fxor \equiv \fcom \equiv \fcf$ under  {\tt cQ-CUC} reductions.  
%\label{thm:morecomp}
%\end{Theorem}

%---------------------------------------------------------------------------%
%\section{Discussions}
%\label{sec:disc}
%---------------------------------------------------------------------------%
%\fnote{merge this section in introduction.}

\section*{Acknowledgments}

We would like to thank anonymous reviewers for valuable comments. This work was informed by insightful discussions with many
colleagues, notably Michael Ben-Or, Daniel Gottesman, Claude Cr{\'e}peau, Ivan
Damg{\aa}rd and Scott Aaronson. Several of the results were
obtained while A.S. was at the Institute for Pure and Applied
Mathematics (IPAM) at UCLA in the fall of 2006. He gratefully
acknowledges Rafi Ostrovksy and the IPAM staff for making his stay
there pleasant and productive.

%\bibliographystyle{prsty}
\bibliographystyle{plainnat}
%\ifnum\final=0
\bibliography{zkpoklong}
%\else
%\bibliography{zkpok}
%\fi
%-----------------------------------------------------------------------------%

\appendix

%-----------------------------------------------------------------------------%
\section{Proof of Modular Composition Theorem}
\label{sec:pmct}
%-----------------------------------------------------------------------------%

\begin{proof}[Proof of Theorem~\ref{thm:sacomposition}]
Let $\Pi' := \cprot{\Pi}{\Gamma}{\Gamma'}$ be the composed protocol. We show the theorem in the computational setting, and proofs for the statistical and perfect settings are analogous.  Specifically, we need to show that
$$ \forall\calA~\exists \calS: \mac_{\Pi',\calS} \approx_{qc} \mac_{\Pi,\calA} \, . $$
%(or equivalently $\forall \calA ~\exists \calS~ \forall \calZ:  \exec_{\Pi', \calA, \calZ} \approx \ideal_{\calF,\calS,\calZ}$.)  

Without loss of generality, we assume that $\Pi$ only calls $\Gamma$ once. The proof will proceed in three steps: 
\begin{enumerate}
	\item From any adversary $\calA$ attacking $\Pi'$, we construct another adversary $\calA_{\Gamma'}$ attacking $\Gamma'$. Notice that $\Gamma'$ is a subroutine in $\Pi'$. Basically $\calA_{\Gamma'}$ consists of the segment of the circuits of $\calA$ during the subroutine call of $\Gamma'$. 
	\item By the assumption that $\Gamma'$ \cqsa~emulates $\Gamma$, we know that $\forall \calA_{\Gamma'} \exists \calA_{\Gamma}: \mac_{\Gamma',\calA_{\Gamma'}} \qc \mac_{\Gamma,\calA_{\Gamma}}$. This gives us an adversary $\calA_{\Gamma}$. 
	\item Finally the adversary $\calS$ will be constructed by ``composing'' the machines $\calA$ and $\calA_{\Gamma}$: when $\Pi$ makes the subroutine call to $\Gamma$, $\calS$ runs  $\calA_{\Gamma}$, otherwise it follows the operations of $\calA$. Then $\mac_{\Pi',\calA} \qc \mac_{\Pi,\calS}$ basically follows from $\mac_{\Gamma',\calA_{\Gamma'}} \qc \mac_{\Gamma,\calA_{\Gamma}}$. 
\end{enumerate} 
Next we give the details. 

\mypar{Step 1 (Constructing $\calA_{\Gamma'}$ from $\calA$)} Adversary $\calA_{\Gamma'}$ represents the segment of $\calA$ during the subroutine $\Gamma'$.
It starts with some state that supposedly represents
the joint state in an execution of $\Pi'$ with $\calA$ right
before the invocation of $\Gamma'$. It then runs $\calA$ till completion of $\Gamma'$. 

\begin{pffs}{Adversary $\calA_{\Gamma'}$}{}{} 

{\sf Input}: adversary $\calA$; security parameter $1^n$; %input state $\sigma_n$.

\begin{enumerate}
\item $\calA_{\Gamma'}$ initiates $\calA$ with whatever input it receives from the environment. It then runs $\calA$ in the execution of $\Gamma'$.
\item When $\Gamma'$ terminates, $\calA_{\Gamma'}$ outputs the state on all of $\calA$'s registers.

\end{enumerate}
\end{pffs}

\mypar{Step 2 (Simulating $\calA_{\Gamma'}$ by $\calA_{\Gamma}$)} 
 This step is straightforward from the hypothesis that $\Gamma'$ \cqsa~emulates $\Gamma$, which means that $\forall \calA_{\Gamma'} \exists \calA_{\Gamma}: \mac_{\Gamma',\calA_{\Gamma'}} \wqc \mac_{\Gamma,\calA_{\Gamma}}$.

\mypar{Step 3 (Constructing $\calS$ from $\calA_{\Gamma}$ and $\calA$)} The construction is as described above. Here we show that $\mac_{\Pi',\calA} \qc \mac_{\Pi,\calS}$. Suppose for contradiction that there exists a distinguisher $\calZ$ and state\footnote{More precisely there exists a family of states $\{\sigma_n\}_{n\in\bbN}$.} $\sigma_n$ such that: 
$$\left| \Pr[\calZ((\I_{\lin{\reg{R}}}\otimes\mac_{\Pi,\calA})(\sigma_n)) = 1] - \Pr[\calZ((\I_{\lin{\reg{R}}}\otimes\mac_{\Pi',\calS})(\sigma_n)) = 1] \right| \geq \veps(n)\, ,$$
 with $\veps(n) \geq 1/{poly(n)}$.  We show a distinguisher $\tilde\calZ$ and state $\tilde \sigma_n$ such that on input $\tilde \sigma_n$, $\mac_{\Gamma',\calA_{\Gamma'}}$ and $\mac_{\Gamma,\calA_{\Gamma}}$ becomes distinguishable under $\tilde \calZ$.  
\begin{itemize}
	\item Let $\tilde \sigma_n$ be the joint state of executing
          $\Pi'$ in the presence of $\calA$ on input $\sigma_n$ right
          before the invocation of $\Gamma'$. Clearly it is identical
          to the joint state of executing $\Pi$ in the presence of $\calS$ on input $\sigma_n$ right before the invocation of $\Gamma$.
	\item Distinguisher $\tilde\calZ$ runs the circuits of $\calA$ after execution of the subroutine $\Gamma'$ (equivalently the circuits of $\calS$ after execution of the subroutine $\Gamma$) and then runs  $\calZ$. 
\end{itemize}
It is easy to see that 
\begin{align*}
& \tilde\calZ((\I_{\lin{\reg{R}}}\otimes\mac_{\Gamma',\calA_{\Gamma'}})(\tilde \sigma_n)) \equiv \calZ((\I_{\lin{\reg{R}}}\otimes\mac_{\Pi',\calA})(\sigma_n)) \\
\text{and }& \tilde\calZ((\I_{\lin{\reg{R}}}\otimes\mac_{\Gamma,\calA_{\Gamma}})(\tilde\sigma_n)) \equiv \calZ((\I_{\lin{\reg{R}}}\otimes\mac_{\Pi,\calS})(\sigma_n)) \, ,
\end{align*}
%$$\tilde\calZ((\I_{\lin{\reg{W}}}\otimes\mac_{\Gamma',\calA_{\Gamma'}})(\tilde \sigma_n)) \equiv \calZ((\I_{\lin{\reg{W}}}\otimes\mac_{\Pi',\calA})(\sigma_n)) \quad \text{and} \quad \tilde\calZ((\I_{\lin{\reg{W}}}\otimes\mac_{\Gamma,\calA_{\Gamma}})(\tilde\sigma_n)) \equiv \calZ((\I_{\lin{\reg{W}}}\otimes\mac_{\Pi,\calS})(\sigma_n)) \, ,$$
where ``$\equiv$'' means identical distributions.  
This implies that $$ \left| \Pr[\tilde\calZ((\I_{\lin{\reg{R}}}\otimes\mac_{\Gamma',\calA_{\Gamma'}})(\tilde\sigma_n)) = 1] - \Pr[\tilde\calZ((\I_{\lin{\reg{R}}}\otimes\mac_{\Gamma,\calA_{\Gamma}})(\tilde\sigma_n)) = 1] \right| \geq \veps(n) \, .$$ This contradicts the assumption that $\mac_{\Gamma',\calA_{\Gamma'}} \qc \mac_{\Gamma,\calA_{\Gamma}}$.

This concludes our proof for the modular composition theorem. 
\end{proof}

\ifnum\final=0

%-----------------------------------------------------------------------------%
\section{Variants of Quantum Stand-Alone Models. (More Details)}
\label{sec:variants-app}
%-----------------------------------------------------------------------------%

%\fnote{update pics}
%We give comprehensive discussions for the choices of defining stand-alone security models carefully.

%%%%%%%%%%%%%%%%%%%%%%%%%%%%%%%%%%%%%%%%%%%%%%%%%%%%%%%%%%%%%%%%%%%

%\subsection{$\calZ_1$ takes quantum advice}
%\label{subsec:wadv}
%%%%%%%%%%%%%%%%%%%%%%%%%%%%%%%%%%%%%%%%%%%%%%%%%%%%%%%%%%%%%%%%%%
\begin{comment}
Surprisingly as it may appear, we can show that once $\calZ_1$ takes quantum advice, the choices for b and c will no longer matter, and we get a single model as a result.

\begin{theorem}[Theorem~\ref{thm:eq_sp_body} Restated] The models where $\calZ_1$ passes state to $\calZ_2$ is equivalent to the models without state passing. Namely $\model{\a}{\bbar}{\cbar}\equiv \model{\a}{\bbar}{\c} \equiv \model{\a}{\b}{\c} \equiv\model{\a}{\b}{\cbar} \ . $

\begin{figure}[h!]
\centering
\ifnum\ijqi=1
{\includegraphics[width=5in]{equiv1}} 
\else
{\includegraphics[width=5in]{Figures/equiv1}} 
\fi
\caption{Quantum advice to $\calZ_1$ collapses all models} 
\label{fig:equiv1}
\end{figure}
\label{thm:eq_sp}
\end{theorem}
\end{comment}

%===Proof moved ahead=====%

\subsection{Proof of Proposition~\ref{prop:cqas}}
\label{ssec:pcqas}

\begin{prop*}[Proposition~\ref{prop:cqas} Restated]
Let $({A,B})$ be a secure QAS. Given two superoperators $\calE: \text{L}(X) \to \text{L}(X')$ and $\calO: \text{L}(W\otimes X) \to \text{L}(W \otimes X')$, define $\calO': \text{L}(X) \to \text{L}(X')$ by $$\rho \mapsto Tr_{W'}\left[\calO \left(\I_{\lin{X}} \otimes {\bar A} (\rho \otimes |0\rangle \langle 0|_{W'})\right)\right]\, .$$
If $\calE\otimes\I_{\lin{W}} \approx_{wqc} \calO$, then $\calE \approx_{qc} \calO'$.

\begin{figure}[h!]

\centering 
\ifnum\ijqi=1
{\includegraphics[width=3in]{oprime}} 
\else
{\includegraphics[width=3in]{Figures/oprime}} 
\fi
\caption{Illustration of $\calO$ and $\calO'$} \label{fig:oprime}
\end{figure}
\label{prop:acqas}
\end{prop*}

\mypar{More on Quantum Authentication Schemes} We now define formally
what a {\em secure} QAS is and prove
Prop.~\ref{prop:cqas}. In~\citet{BCGST02}, security is only defined with respect to a closed message system. Here we extend their definition to capturing the entangled input case as well, i.e., the message being authenticated is entangled with another reference system. In addition, we define an operation $cE$, which we call {\em conditional-erasure}. Given a state $\rho \in \density{{X}\otimes{Y}}$ where ${X}$ is a qubit system, $cE$ first measures ${X}$ in the computational basis. If the outcome is $1$, $cE$ leaves ${Y}$ untouched; otherwise, ${Y}$ is replaced with $|0\rangle$.

\begin{figure}[h!]
\centering 
\ifnum\ijqi=1
{\includegraphics[width=4in]{qas-sound}} 
\else
{\includegraphics[width=4in]{Figures/qas-sound}} 
\fi
\caption{Soundness of QAS. Here $\rho_1$ denotes the final state in a
  real execution of the QAS,  while $\rho_2$ denotes the state
  obtained when a dummy message $|0\rangle$ is authenticated.} \label{fig:qas-s}
\end{figure}

\begin{definition}(Secure QAS with entangled input) A QAS $(A,B)$ is secure with error $\varepsilon$ for a state
$|\phi\rangle \in {R\otimes M}$ if it satisfies:
\begin{itemize}
    \item {\em Completeness}. For all keys $ k\in\mathcal{K}$: $(B_k\otimes\I_{\lin{R}}) (A_k\otimes\I_{\lin{R}})(|\phi\rangle\langle\phi|) = |\phi\rangle\langle\phi| \otimes |ACC \rangle\langle ACC|$.
    \item {\em Soundness}. For all superoperators $\calO$ acting on $M\otimes R$, define (see Figure~\ref{fig:qas-s}):
        $$\rho_1 := cE( \bar{B}\otimes \I_{\lin{R}}) {\calO} (\bar{A} \otimes \I_{\lin{R}}) (|\phi\rangle \langle \phi|_{MR})\, ,$$
        $$ \rho_2 := Tr_{M'}[(\I_{\lin{M'}} \otimes cE)  (\bar{B}\otimes \I_{\lin{M\otimes R}}) ({\calO} \otimes \I_{\lin{M}})( \bar{A}\otimes \I_{\lin{M\otimes R}})(|0\rangle \langle 0|_{M'}\otimes|\phi\rangle \langle \phi|_{MR})] \, ,$$ where $M'$ is an auxiliary system with the same dimension as $M$ and is initialized with $|0\rangle$.
        
        Soundness requires that the trace distance between $\rho_1$ and $\rho_2$ is at most $\varepsilon$: $\td(\rho_1, \rho_2) \leq \varepsilon$.
\end{itemize}
A QAS is $\varepsilon$-secure if it is secure with error
$\varepsilon$ for all state $|\phi\rangle$\label{def:sqas}. When $\varepsilon$ is a negligible function in the security parameter $n$, we simply call the scheme a secure QAS.
\end{definition}

The authentication scheme of~\citet{BCGST02} satisfies this stronger notion. This follows from the work of \citet{LHM11}, which showed that this scheme is quantum-UC secure. Our soundness definition is essentially equivalent to the formulation of quantum-UC secure authentication. But we restrict the simulator to have a particular form (implicit in the operation that defines $\rho_2$). It is easy to verify that the proof in~\cite{LHM11} goes through with the simulator of the desired form.  

%\mypar{Proof of Prop.~\ref{prop:cqas}} Now we can give the proof. 
\begin{proof}[Proof of Prop.~\ref{prop:cqas}]
Give any $|\psi\rangle \in X\otimes W$, let
$$\rho = (\calE\otimes \I_{\lin{W}} ) (|\psi\rangle \langle \psi|), \quad \sigma = ( \calO'\otimes \I_{\lin{W}} )(|\psi\rangle \langle \psi|) \ .$$
We need to show that $ \rho \approx_{wqc} \sigma $. Towards this goal, define
\begin{eqnarray*}
\tau_1' :=  (\calE \otimes \I_{\lin{W}}) ( \I_{\lin{X}} \otimes \bar{A}) (|\psi\rangle \langle \psi|), &&\quad \tau_1 : =( \I_{\lin{X'}} \otimes \bar{B})(\tau_1')\, ;\\
\tau_2' := \calO (\I_{\lin{X}} \otimes \bar{A}) (|\psi\rangle \langle \psi|), &&\quad \tau_2 : = (\I_{\lin{X'}} \otimes \bar{B}) (\tau_2') \, .
\end{eqnarray*}
By assumption that  $\calE\otimes\I_{\lin{W}} \approx_{wqc} \calO$, we know that $ \tau_1' \wqc \tau_2' \quad (*)$. 
We then show $\rho \approx_{wqc} \sigma$ via a sequence of claims below. 

\begin{claim}
$\rho = Tr_V(\tau_1)$ %and $D(Tr_V(\tau_2), \sigma) \leq %2(\delta(n)+\varepsilon(n))$.
\label{claim:cqas1}
\end{claim}
\begin{proof}[Proof of Claim~\ref{claim:cqas1}]
This is actually the completeness condition of QAS in disguise.
\begin{eqnarray*}
\tau_1 & = &  (\I_{\lin{X}} \otimes \bar{B})( \calE \otimes \I_{\lin{W}} )( \I_{\lin{X}} \otimes \bar{A}) (|\psi\rangle \langle \psi|)\\
& = & (\calE \otimes \I_{\lin{W}}) (\I_{\lin{X}} \otimes \bar{B})(\I_{\lin{X}} \otimes \bar{A}) (|\psi\rangle \langle \psi|) \quad (\text{Commutativity})\\
& = & (\I_{\lin{W}} \otimes \calE)(|\psi\rangle \langle \psi|) \otimes |1\rangle \langle 1|_V \quad (\text{Completeness of QAS})\\
& = & \rho \otimes |1\rangle \langle 1|_V \, .
\end{eqnarray*}
\end{proof} %of claim:cqas1

\begin{claim}
Let $P^{acc}:= |1\rangle \langle 1|_V \otimes \I_{\lin{W\otimes X'}}$ and $p:= Tr(P^{acc}\tau_2) $. Then $$ \td(\sigma, Tr_{V}(\tau_2)) \leq 2\varepsilon(n) + 2 (1 - p) \ , $$ where the negligible function $\varepsilon(n)$ is the soundness error of the QAS.
\label{claim:cqas2}
\end{claim}

\begin{proof}[Proof of Claim~\ref{claim:cqas2}]
Write
\begin{eqnarray*}
\tau_2 &=& p |1\rangle \langle 1|_V \otimes \tau_2^{acc} + (1 - p) |0 \rangle \langle 0|_V \otimes \tau_2^{rej} \, ,\\
\text{and } \sigma & = & q \sigma^{acc} + (1 - q) \sigma^{rej} \, .
\end{eqnarray*}

As one may have already expected, this claim relies heavily on the soundness property of secure-QAS, we make the connection precise by defining
\begin{eqnarray*}
\eta &:=& p |1\rangle \langle 1|_V \otimes \tau_2^{acc} + (1 - p) |0 \rangle \langle 0|_V \otimes |0 \rangle \langle 0|_{X'\otimes W} \, ,\\
\text{and } \eta' & := & q |1\rangle \langle 1|_V \otimes \sigma^{acc} + (1 - q) |0 \rangle \langle 0|_V \otimes |0 \rangle \langle 0|_{X'\otimes W} \, .
\end{eqnarray*}
Then by soundness of QAS, we have that $\td(\eta, \eta') \leq \varepsilon(n) \quad (1)$. This in particular implies that $|p - q| \leq \varepsilon(n)$ and hence $q \geq p - \veps(n)$.  
Moreover 
$$\td(\tau_2, \eta) = (1-p)\cdot \|\tau_2^{rej} - |0 \rangle \langle 0|_{X'\otimes W} \|_1 \leq 1 - p \, , \quad (2) 
$$
$$ \text{and } \td(\sigma, Tr_{V}(\eta')) = (1-q) \cdot \|\sigma^{rej} - |0 \rangle \langle 0|_{X'\otimes W} \|_1 \leq 1 - q \leq 1 - p + \varepsilon(n) \, . \quad (3) $$
Combine (1) - (3), we conclude, by the triangle inequality, that 
\begin{align*}
& \td(\sigma, Tr_{V} (\tau_2)) \\
\leq & \td(\tau_2, \eta) + \td(\eta, \eta') + \td(Tr_{V} (\eta'), \sigma) \\
\leq & (1- p) + \veps(n) + (1 - p + \veps(n))  \\
= & 2\veps(n) + 2(1 - p) \, .
\end{align*}
\end{proof} % of claim:cqas2

\begin{claim}
$\tau_1\approx_{wqc} \tau_2$. 
\label{claim:cqas3}
\end{claim}
\begin{proof}[Proof of Claim~\ref{claim:cqas3}]
We show that for any poly-time $\calZ$, $ \delta(n):= \left|\Pr(\calZ(\tau_1) = 1) - \Pr(\calZ(\tau_2) = 1)\right| $ must be negligible. Suppose for contradiction that $\delta(n)  > 1/{\kappa(n)}$ for some polynomial $\kappa(n)$. We design $\calZ'$  that distinguishes $\tau_1'$ and $\tau_2'$, and this contradicts (*). Distinguisher $\calZ'$ will have the random classical key $k$ that $\mathcal{A}$ used in authentication hardwired into its circuit. Then upon receiving either $\tau_1'$ or $\tau_2'$, it first apply $B_k$--the verification operation of QAS on system $\calW$, and pass the resulting state to $\calZ$. Output whatever $\calZ$ outputs. It is easy to see that
\begin{eqnarray*}
&& |\Pr(\calZ'(\tau_1')=1) - \Pr(\calZ'(\tau_2') =1)|\\
& = & |\Pr(\calZ ((\I_{X'}\otimes B_k)(\tau_1'))=1) - \Pr(\calZ((\I_{X'}\otimes B_k)(\tau_2')) =1) |\\
& = & |\Pr(\calZ(\tau_1)=1) - \Pr(\calZ(\tau_2) =1)| > 1/{\kappa(n)} \, .
\end{eqnarray*}
\end{proof} % of Claim:cqas3

Immediately we obtain as a corollary of Claim~\ref{claim:cqas3} that $p \geq 1 - \delta(n)$. This is because $Tr(P^{acc}\tau_1) = 1,$ and $Tr(P^{acc}\tau_2) = p$. If we consider a QTM $\calZ$ that measures $V$ in computational basis to distinguish $\tau_1$ and $\tau_2$. Then Claim~\ref{claim:cqas3} tells us that
$$1 - p_{acc} = | Tr(P^{acc}\tau_1) - Tr(P^{acc}\tau_2)| = |\Pr(\calZ(\tau_1) = 1) - \Pr(\calZ( \tau_2) = 1)|  \leq \delta(n) \ .$$

Combining the above claims, we conclude that $\forall |\psi\rangle \in X\otimes W$ and for all poly-time QTM $\calZ$, 
\begin{align*}
& |\Pr(\calZ(\rho) = 1) - \Pr(\calZ( \sigma) = 1)| \\
= & \left| \Pr(\calZ(\rho) = 1) - \Pr(\calZ(Tr_V(\tau_1)) = 1) \right| \\
+ & \left|\Pr(\calZ(Tr_V(\tau_1) ) = 1) - \Pr(\calZ(Tr_V(\tau_2) ) =1)\right| \\
+ & \left|\Pr(\calZ(Tr_V(\tau_2) ) =1) - \Pr(\calZ( \sigma) = 1) \right |\\
\leq & 0 + \delta(n)+ 2 \veps(n)+ 2(1 - p_{acc}) \\
\leq & 2(\delta(n) + \veps(n)) = \delta'(n) \, , 
\end{align*}
and $\delta'(n)$ is still negligible. This means that $\rho \approx_{wqc} \sigma$ and as a result $\calE \approx_{qc} \calO'$. 
\end{proof} % of lemma:cqas

%%%%%%%%%%%%%%%%%%%%%%%%%%%%%%%%%%%%%%%%%%%%%%%
\subsection{$\calZ_1$ Does Not Take Quantum Advice}
\label{subsec:woadv}
%%%%%%%%%%%%%%%%%%%%%%%%%%%%%%%%%%%%%%%%%%%%%%%

When $\calZ_1$ takes no quantum advice, namely the inputs to the players are generated by a poly-time QTM, the situation becomes less clear.

\mypar{State Passing still Makes No Difference} Since we can think of $\calZ_1$'s that do not take quantum advice as a subclass of $\calZ_1$'s that do take quantum advice (i.e., empty advice), an analogous argument as in Lemma~\ref{lemma:eqsp2} shows that

\begin{theorem} Without
  quantum advice to $\calZ_2$, state passing still does not change the
  model. That is,
$$\model{\abar}{\bbar}{\cbar} \equiv \model{\abar}{\bbar}{\c}  \text{ and } \model{\abar}{\b}{\c} \equiv \model{\abar}{\b}{\cbar} \ . $$
\label{thm:nz1ps}
\end{theorem}
\vspace{-18pt}
\begin{figure}[h!] \centering
\ifnum\ijqi=1
\mbox{ \subfigure[$ \model{\abar}{\bbar}{\cbar} \equiv \model{\abar}{\bbar}{\c} $ ]
{\includegraphics[width=0.45\columnwidth]{equiv2a}}\label{fig:naux2}
\qquad \subfigure[$\model{\abar}{\b}{\c} \equiv \model{\abar}{\b}{\cbar}$]{\includegraphics[width=0.45\columnwidth]{equiv2b}
}\label{fig:aux2}}
\else
\mbox{ \subfigure[$ \model{\abar}{\bbar}{\cbar} \equiv \model{\abar}{\bbar}{\c} $ ]
{\includegraphics[width=0.45\columnwidth]{Figures/equiv2a}}\label{fig:naux2}
\qquad \subfigure[$\model{\abar}{\b}{\c} \equiv \model{\abar}{\b}{\cbar}$]{\includegraphics[width=0.45\columnwidth]{Figures/equiv2b}
}\label{fig:aux2}}
\fi
\caption{Equivalent models when $\calZ_1$ takes no quantum advice}
\label{fig:equiv2}
\end{figure}

\mypar{Quantum advice to $\calZ_2$} Recall that when $\calZ_1$ takes quantum advice, quantum advice to $\calZ_2$ turns out to be irrelevant. However, here it is unclear if such a result still holds. A closely related question would be asking whether $\text{BQP/qpoly} \stackrel{?}{=} \text{BQP/{poly}}$, which is one of the most important open problems in quantum complexity theory. %We may connect complexity theory with our security models in different ways:
%\begin{itemize}
%\item 
For example, assume $\text{BQP/qpoly} {=} \text{BQP/poly}$, does that imply that $\model{\abar}{\bbar}{\c} \equiv \model{\abar}{\b}{\c}\, ?$ This would collapse all four models under the choice that $\calZ_1$ takes no quantum advice.  
Otherwise what is the minimal complexity assumption that suffices to derive $\model{\abar}{\bbar}{\c} \equiv \model{\abar}{\b}{\c}$?
%\item 
Conversely, if it turns out that $\model{\abar}{\b}{\c} \not\equiv \model{\abar}{\bbar}{\c}$, are there any interesting complexity implications?
%\end{itemize}

%%%%%%%%%%%%%%%%%%%%%%%%%%%%%%%%%%%%%%%%%%%%%%%
\subsection{A Special Constraint: Markov Condition}
\label{subsec:mc}
%%%%%%%%%%%%%%%%%%%%%%%%%%%%%%%%%%%%%%%%%%%%%%%
\fi

%==========================%
\section{A Special Constraint in Quantum Stand-Alone Model: Markov Condition}
\label{sec:mc}
%==========================%
Another choice exists in the literature~\cite{DFLSS09,FS09}, where a stand-alone model was proposed to
capture secure emulation of {\em classical} functionalities. Only a
special form of inputs is allowed there, which satisfy what we call the {\em Markov condition}. As opposed to a general bipartite state with one part being classical (a.k.a {\em cq-states}): $\rho_{AB} = \sum_{a} \lambda_a |a \rangle\langle a | \otimes \rho_B^a$, the Markov condition requires that the input to dishonest Bob contains a classical subsystem $Y$ such that conditioned on $Y$ Bob's quantum input is independent of Alice's classical input. Such states are denoted as $$\rho_{A \leftrightarrow Y \leftrightarrow B} = \sum_{a,b} \lambda_{a,b} |a\rangle \langle a|_A \otimes |b\rangle \langle b|_Y \otimes \rho_{B}^b\, .$$

Now let us analyze how Markov condition affects our abstract model discussed above. It turns out that the effect of Markov condition, again, depends on whether $\calZ_1$ takes quantum advice.

\mypar{$\calZ_1$ takes quantum advice: Markov condition becomes redundant}
We denote models with Markov condition $\calM^*$.
\begin{lemma}
$\model{\a}{\cdot}{\cdot}^* \equiv \model{\a}{\cdot}{\cdot}$ regardless of the choices for $aux_2$ and state passing. Namely, the
model where inputs must satisfy Markov condition is equivalent to the model where inputs could be any bipartite cq-states. \label{lemma:mar1}
\end{lemma}

\begin{proof} To be concrete, we consider two models $\calM:= \model{\a}{\bbar}{\cbar}$ and $\calM' := \model{\a}{\bbar}{\cbar}^*$. The same argument applies to other cases.

One direction is obvious, namely, if a protocol $\Pi$ emulates $\Gamma$ in $\calM$
then it automatically holds that $\Pi$ emulates $\Gamma$ in $\calM'$.  This is because we can think of the Markov condition as specifying a subclass of possible $\calZ_1$ allowed in $\calM$. Now we show the converse by contradiction. Specifically, we prove that if there is an adversary $\calA$ in $\calM$, and $\forall \calS$, there exist $(\calZ_1, \calZ_2)$ such that $\calZ_2$ can
distinguish $\mac_{\Pi, \calA}$ and $\mac_{\Gamma,\calS}$, then in model $\calM'$ we construct $\calA', \calZ_1'$ and $\calZ_2'$ such that no $\calS'$ that is
able to simulate $\calA'$. By our hypothesis, there is
an input state $\sigma_n$, which can always be written as $\sum_{a} \lambda_a |a \rangle \langle
a|_A\otimes \sigma_{B}^a$ with $\sum_i {\lambda_i} = 1$ such that
$$\left| \Pr[\calZ_2(M_{\Pi, \calA}(\sigma_n)) =1] - \Pr[\calZ_2(M_{\Gamma,\calS}(\sigma_n))=1] \right| \geq 1/{poly(n)}$$ holds for any poly-time $\calS$. Observe that each summand $|a \rangle \langle a|_A\otimes \sigma_{B}^a$ of $\sigma_n$ trivially satisfies Markov condition. Since $\sigma_n$ is a convex combination of $|a \rangle \langle a|_A\otimes \sigma_{B}^a$, there
must be a $\tilde{\sigma}_n = |\tilde a \rangle \langle \tilde a|_A\otimes \sigma_{B}^{\tilde a}$ such that $$\left|
\Pr[\calZ_2(M_{\Pi,\calA}(\tilde \sigma_n))=1] -
\Pr[\calZ_2(M_{\Gamma,\calS}(\tilde \sigma_n))=1] \right| \geq 1/{poly(n)} \, ,$$ for
any poly-time $\calS$. %\fnote{Fixed. convexity? how does this follow?}
This observation tells us that we can simply
let $\calA' := \calA, \calZ_2' := \calZ_2$, and let $\calZ_1'$ be the machine that takes quantum advice $\{\tilde \sigma_n\}$ and hands  $\tilde \sigma_n$ to players as input. Then for any poly-time $\calS'$,
\begin{align*}
& \left|\Pr[\calZ_2'(M_{\Pi,\calA'}(\tilde \sigma_n))=1] -
\Pr[\calZ_2'(M_{\Gamma,\calS'}(\tilde \sigma_n))=1] \right| \\
 = &\left|\Pr[\calZ_2(M_{\Pi,\calA}(\tilde \sigma_n))=1] - \Pr[\calZ_2(M_{\Gamma,\calS'}(\tilde \sigma_n))=1] \right|  \geq 1/{poly(n)} \, .
\end{align*}
This shows that emulation in $\calM'$ implies emulation in $\calM$.
\end{proof}

\mypar{$\calZ_1$ Takes No Quantum Advice: Markov Condition may Matter}
The argument in Lemma~\ref{lemma:mar1} does not necessarily apply here
because previously we could simply give $\sigma_n^*$ to $\calZ_1$
directly as an advice. However, $\sigma_n^*$ might be impossible to
generate on a poly-time QTM. It is interesting to either construct a
concrete example to show a separation or otherwise showing a proof of
equivalence. We do not have clear insight into the Markov condition in
this case, and leave the possibility of a separation as an open question.

%\fnote{NOT fixed. big message? the two points below are not justifications.}

%\mypar{Reflections}  Given the above discussion, we feel that it is reasonable to suggest that, when dealing with efficient quantum adversaries, allowing quantum advice to adversaries should be the ``default'' choice. A few justifications are in order. First of all, poly-time QTMs with quantum advice is the most general model under the quantum circuit computation formalism, and as demonstrated earlier, quantum advice often simplifies the discussion dramatically. Moreover, despite the seemingly strong power of adversaries, we can still get secure cryptographic protocols (e.g., SFE) against such adversaries under plausible computational assumptions. Finally, the use of nonuniform modeling is standard in classical cryptography; allowing quantum advice makes the quantum models strict generalizations of their classical counterparts. Many existing works do not consider quantum advice to adversaries (or do not fully specify the adversarial model). We think it makes sense to revisit such results against adversaries with quantum advice, though it should not be surprising that most, if not all, results should still hold, subject possibly to adaptation of the underlying assumptions.
\end{document}